\documentclass[letterpaper]{article} % DO NOT CHANGE THIS
\usepackage{aaai25}  % DO NOT CHANGE THIS
\usepackage{times}  % DO NOT CHANGE THIS
\usepackage{helvet}  % DO NOT CHANGE THIS
\usepackage{courier}  % DO NOT CHANGE THIS
\usepackage[hyphens]{url}  % DO NOT CHANGE THIS
\usepackage{graphicx} % DO NOT CHANGE THIS
\usepackage{natbib}  % DO NOT CHANGE THIS AND DO NOT ADD ANY OPTIONS TO IT
\usepackage{caption} % DO NOT CHANGE THIS AND DO NOT ADD ANY OPTIONS TO IT
\usepackage{algorithm}
\usepackage{algorithmic}

\usepackage{newfloat}
\usepackage{listings}
\DeclareCaptionStyle{ruled}{labelfont=normalfont,labelsep=colon,strut=off} % DO NOT CHANGE THIS
\floatstyle{ruled}
\newfloat{listing}{tb}{lst}{}
\floatname{listing}{Listing}
\usepackage{amsthm}
\usepackage{amsmath}
\usepackage{amssymb}
\usepackage{dsfont}
\usepackage{xcolor}

\global\long\def\real{\mathbb{R}}
\global\long\def\t{\mathrm{T}}
\global\long\def\prob{\mathbb{P}}
\global\long\def\expect{\mathbb{E}}
\global\long\def\cov{\mathrm{cov}}

\global\long\def\tdomain{\mathcal{T}}
\global\long\def\bdelta{\boldsymbol{\delta}}
\global\long\def\bw{\boldsymbol{w}}

\global\long\def\bY{\boldsymbol{Y}}
\global\long\def\bb{\boldsymbol{b}}
\global\long\def\trace{\mathrm{tr}}
\global\long\def\bepsilon{\boldsymbol{\epsilon}}

\global\long\def\bxi{\boldsymbol{\xi}}
\global\long\def\bv{\boldsymbol{v}}
\global\long\def\bgamma{\boldsymbol{\gamma}}
\global\long\def\brho{\boldsymbol{\rho}}

\newtheorem{theorem}{Theorem}
\newtheorem{lemma}{Lemma}
\newtheorem{corollary}{Corollary}
\newtheorem{proposition}{Proposition}
\newtheorem{remark}{Remark}

\newtheorem{assumption}{Assumption}

\title{Transfer Learning Meets Functional Linear Regression: No Negative Transfer under Posterior Drift}
\author{
Xiaoyu Hu,
Zhenhua Lin
}
\affiliations{
Department of Statistics and Data Science, National Univeristy of Singapore 
}

\usepackage{bibentry}
\begin{document}

\maketitle	
	
	\begin{abstract}
		Posterior drift refers to changes in the relationship between responses and covariates while the distributions of the covariates remain unchanged. 
		In this work, we explore functional linear regression under posterior drift with transfer learning.
		Specifically, we investigate when and how auxiliary data can be leveraged to improve the estimation accuracy of the slope function in the target model when posterior drift occurs.
		We employ the approximated least square method together with a lasso penalty to construct an estimator that transfers beneficial knowledge from source data. 
		Theoretical analysis indicates that our method avoids negative transfer under posterior drift, even when the contrast between slope functions is quite large. Specifically, the estimator is shown to perform at least as well as the classical estimator using only target data, and it enhances the learning of the target model when the source and target models are sufficiently similar. 
		Furthermore, to address scenarios where covariate distributions may change, we propose an adaptive algorithm using aggregation techniques. This algorithm is robust against non-informative source samples and effectively prevents negative transfer.
		Simulation and real data examples are provided to demonstrate the effectiveness of the proposed algorithm.
	\end{abstract}
	
		%	     \link{Code}{https://aaai.org/example/code}
		%	     \link{Datasets}{https://aaai.org/example/datasets}
		%	     \link{Extended version}{https://aaai.org/example/extended-version}
		%	 \end{links}
	
	\section{Introduction}
	
	Functional data analysis (FDA) has gained increasing attention over the past two decades. Two monographs, \citet{ramsay2005functional} and \citet{hsing2015theoretical}, provide comprehensive treatments on methodologies and theories of FDA. In particular, functional linear regression has emerged as a crucial tool and has been extensively studied in the literature, including \citet{yao2005functional,cai2006prediction,hall2007methodology,li2007rates,crambes2009smoothing,yuan2010reproducing}. The aforementioned methodologies are developed and can be successful when there is sufficient training data for the target task. However, in many real-world applications, the available training data is often limited, leading to unsatisfactory estimation results. Fortunately, samples from different but related sources can provide beneficial information to boost performance on the target problem. 
	The process of transferring knowledge from additional data to improve the task on the target data is a popular topic known as transfer learning \citep{pan2009survey,zhuang2020comprehensive}.
	
	Transfer learning has been widely applied to various tasks, including text classification \citep{xue2008topic}, recommendation systems \citep{pan2013transfer}, and medical diagnosis \citep{hajiramezanali2018bayesian}. 
	Although many methodologies have been developed in the machine learning community, less attention is paid to the statistical properties and theoretical guarantees. Recently, some works have begun to explore transfer learning algorithms in different statistical models \citep{cai2021transfer,li2022transferlasso,tian2022transfer,tian2022unsupervised,zhang2022transfer,li2023estimation,li2023transfer,jin2024transfer}.
	However, in the context of functional linear regression, the estimation performance under transfer learning remains unclear. 
	%Concurrent with our work, 
	For prediction, \citet{lin2024hypothesis} considered using reproducing kernel Hilbert spaces (RKHS) and derived the error bound on the excess prediction risk under transfer learning. In contrast, estimation in functional linear regression is a more difficult problem than prediction, as measured by the minimax optimal rate \citep{cai2006prediction,hall2007methodology}. The impact on slope estimation, which matters for understanding the effect of functional explanatory variables on the response, has not yet been investigated in the context of transfer learning. 
	%In conventional functional linear regression, many widely-used methods considered estimation of the slope function based on functional principal component analysis (FPCA) \citep{yao2005functional,hall2007methodology}. In particular, \citet{hall2007methodology} established the minimax optimal rate of slope estimation.
	
	In this paper, we study the slope estimation problem in functional linear regression under the transfer learning setting. Specifically, we focus on the case of posterior drift, where the relationship between responses and covariates changes while the covariate distributions remain unchanged \citep{kouw2018introduction,cai2021transfer,li2022transferlasso,maity2024linear}. Notably, we can relax the condition of unchanged covariate distributions to aligned eigenspaces of the functional covariates, see Sections \ref{sec:method} and \ref{sec:thm} for details.
	Since the slope function is intrinsically infinite-dimensional, truncation is indispensable to balance bias and variance during estimation.   
	Our main idea involves constructing the estimator using the approximated least square method.
	We project all functional covariates onto the common eigenspace estimated from the auxiliary data and define the transformed variables as scores.
	We use the scores from the source data to obtain an initial estimator, then correct the bias with a lasso penalty using the target data \citep{tibshirani1996regression}.  
	
	Moreover, we derive the convergence rate of the proposed estimator and demonstrate that there is no negative transfer in the presence of posterior drift. 
	When the contrast between slope functions is sufficiently small, the proposed algorithm effectively improves estimation performance on the target model by transferring knowledge from source data.
	It is noteworthy that the truncation parameter depends not only on the smoothness of the slope function and the decay rate of the eigenvalues of covariance functions, but also on the relatedness between target and source models.
	
	The theoretical challenges mainly lie in the following aspects. First, the projection scores of the target data are correlated, differing from the conventional case in \citet{hall2007methodology}. Second, the truncation level is potentially larger than the target sample size due to additional information from auxiliary data, which brings new theoretical issues.
	To tackle these challenges, we employ the oracle inequality to quantify the error bounds and generalize the results regarding the restricted eigenvalue condition in high-dimensional regression to the functional data setting \citep{raskutti2010restricted,negahban2012unified}. 
	
	In practice, we may not know whether the covariate distributions are equal or share an aligned eigenspace. To avoid performance deterioration, we provide an adaptive algorithm utilizing aggregation techniques. We demonstrate its robustness and effectiveness in alleviating negative transfer caused by non-informative source samples through numerical studies.
	
	The main contributions of this paper are summarized as follows:
	\begin{itemize}
		\item We propose an effective algorithm to enhance the learning of the target slope function by transferring knowledge from source samples under posterior drift. % or the relaxed condition where covariate distributions share an aligned eigenspace.
		\item We theoretically demonstrate the absence of negative transfer in slope estimation for functional linear regression under posterior drift.% or the relaxed condition where covariate distributions share an aligned eigenspace.
		\item Under scenarios where covariate distributions may differ substantially, we propose an adaptive algorithm through sparse aggregation to prevent performance degradation.
	\end{itemize}
	
	\paragraph{Notations.} 
	We use bold letters to denote vectors. For a given $p$-dimensional vector $\bw=(w_1, \dots, w_p)^\t \in \real^p$, the $l_q$ norm is given by $\|\bw\|_q = (\sum_{j=1}^p |w_j|^q)^{1/q}, q>0$ and $\|\bw\|_{\infty} = \max_{j}|w_j|$. For a function $f:[0,1] \to \real$, let $\|f\|_2^2 = \int_0^1 f^2(t) dt$.
	For a random variable $\xi$, define the $\psi_p$-Orlicz norm by $\|\xi\|_{\psi_p} = \inf\{c>0: E \exp(|\xi|^p/c^p) \le 2 \} $ for $p\ge 1$.
	For $a, b \in \real$, define $a \vee b = \max(a,b)$ and $a \land b = \min(a, b)$. 
	For two deterministic and non-negative sequences $\{a_n\}_{n=1}^\infty$ and $\{b_n\}_{n=1}^\infty$, we use $a_n \ll b_n$ or $a_n = o(b_n)$ if $a_n/b_n \to 0$ as $n \to \infty$. And $a_n = O(b_n)$ or $a_n \lesssim b_n$ if $\sup_n a_n/b_n < \infty$. Let $a_n \asymp b_n$ if $a_n \lesssim b_n$ and $b_n \lesssim a_n$. 
	%Write $a~ \widetilde \in~ [l, u]$ if $l \lesssim a \lesssim u$.
	For two random sequences $\{x_n\}_{n=1}^\infty$ and $\{y_n\}_{n=1}^\infty$, let $x_n = O_P(y_n)$ denote $P(|x_n/y_n| \le c) \to 1$ for some finite constant $c>0$ and $x_n=o_P(y_n)$ denote $P(|x_n/y_n|>c) \to 0$ for any constant $c>0$.
	Unless otherwise stated, let $c, c_1, c_2, \dots$ and $C, C_1, C_2, \dots$ denote positive constants, not depending on the sample sizes $n, n_1, \dots, n_L$. We allow $c$ and $C$ to be different at different appearances.
	
	\section{Methodology Under Posterior Drift}\label{sec:method}
	
	\subsection{Models}
	%We consider transfer learning in functional linear regression. 
	Given the target distribution $Q$, the observations $(X_{1}, Y_{1}), \cdots, (X_{n}, Y_{n})$ are independent and identically distributed (i.i.d.) from $Q$, where $X_i \in L^2(\tdomain)$ is the random covariate function in the space of square integrable functions on a compact interval $\tdomain$ and $Y_i$ is the scale response. Without loss of generality, let $\tdomain = [0, 1]$. 
	The target model is 
	\begin{equation}\label{eq:target_model}
		Y_{i} - EY = \int_{\tdomain} b(t)\big(X_{i}(t) - \mu(t) \big) dt +\epsilon_{i},  ~~~~i=1, \dots, n,
	\end{equation}
	where $b(t)$ is an unknown slope function, $\mu(t) = \expect \{X(t)\}$ and $\epsilon_i$ is random noise with $\expect(\epsilon_i | X_i) = 0$ and variance $\sigma^2$, independent of the covariates. 
	
	In the context of transfer learning, we observe additional data from source distributions $P^{(l)}$, where $l=1,\dots, L$. Denote $\mathcal A = \{ 1, \dots, L\}$.
	The independent random samples $(X_{1}^{(l)}, Y_{1}^{(l)}), \cdots, (X_{n_l}^{(l)}, Y_{n_l}^{(l)}) \stackrel{i.i.d.}{\sim} P^{(l)}$ are generated by the following source models,
	\begin{equation}\label{eq:source_model}
		Y_{i}^{(l)} - \expect (Y^{(l)}) = \int_{\tdomain} w^{(l)}(t)\big(X_{i}^{(l)}(t) - \mu^{(l)}(t) \big) dt +\epsilon_{i}^{(l)}, 
	\end{equation}
	$i=1, \dots, n_l$, where $w^{(l)}(t)$ is an unknown slope function, $\mu^{(l)}(t) = \expect \{X^{(l)}(t)\}$ and $\epsilon_i^{(l)}$ is random noise with $\expect (\epsilon_i^{(l)} | X_i^{(l)} ) = 0$ and variance $(\sigma^2)^{(l)}$, independent of the covariates.
	
	\subsection{Contrast under Posterior Drift}
	Denote the covariance functions by $K(s,t) = \cov(X(s), X(t))$, $K^{(l)}(s,t) = \cov(X^{(l)}(s), X^{(l)}(t))$. Under posterior drift where $K^{(l)}(s,t)$ and $K(s,t)$ share an aligned eigenspace, the spectral expansion is given by
	\[ K(s,t) = \sum_k \lambda_k \phi_k(s) \phi_k(t),\] 
	\[K^{(l)}(s,t) = \sum_k \lambda_k^{(l)} \phi_k(s) \phi_k(t), \]
	where $\lambda_1 \ge \lambda_2 \ge \cdots \ge 0$, $\lambda_1^{(l)} \ge \lambda_2^{(l)} \ge \cdots \ge 0$ and $\int_\tdomain \phi_k(t) \phi_l(t)dt = \mathds{1}(k=l)$.
	Consequently, the Karhunen-L$\grave{\mathrm o}$eve expansion is
	\[ X_i(t) = \mu(t) +  \sum_{k} \xi_{i,k} \phi_k(t), ~~ i=1, \dots, n, \] 
	\[ X_i^{(l)}(t) = \mu^{(l)}(t) + \sum_k \xi_{i,k}^{(l)}\phi_k(t), ~~ i =1, \dots, n_l, \]
	where $\xi_{i,k} = \int_\tdomain \big( X_i(t) - \mu(t) \big) \phi_k(t) dt $ and $\xi_{i,k}^{(l)} = \int_\tdomain \big(X_i^{(l)}(t) - \mu^{(l)}(t) \big) \phi_k(t)dt $.
	Given the complete orthonormal basis $\phi_1, \phi_2, \dots$, write
	\[ w^{(l)}(t) = \sum_{k=1}^\infty w_{k}^{(l)} \phi_k(t), ~~~~ b(t) = \sum_{k=1}^\infty b_{k} \phi_k(t). \] 
	In general, the slope functions $w^{(l)}(t)$ are different from $b(t)$.
	Let $\delta_k^{(l)} = b_k - w_k^{(l)}$ and $\|\bdelta^{(l)} \|_1 = \sum_{k=1}^\infty |\delta_k^{(l)}|$ for $l=1, \dots, L$. 
	%Denote $\mathcal A_h = \{1 \le l \le L: \|\bdelta^{(l)}\|_1 \le h \}$. 
	Note that the contrast $\|\bdelta^{(l)}\|_1$ measures the relatedness between the target model and source models.
	Clearly, the smaller the contrast is, the more information can be transferred from the source data. 
	Intuitively, if we expect to improve the target estimation by borrowing information from source data, the source models should be sufficiently close to the target model, that is, the contrast should be sufficiently small. This will be discussed further in Section \ref{sec:thm}.
	
	\subsection{Transfer Learning under Posterior Drift}
	
	%In practice, we construct the estimators based on the observed source and target data. 
	To handle the infinite dimensionality, we first perform functional principal component analysis (FPCA) on the source data to estimate the eigenfunctions and obtain the projected score variables for both source and target data.
	Denote 
	\begin{align*}
		&\quad \hat{K}^{(l)}(s,t)  \\
		& =  \frac{1}{n_l-1} \sum_{i=1}^{n_l}\left\{X_{i}^{(l)}(s) - \bar{X}^{(l)}(s) \right\}\left\{X_{i}^{(l)}(t) - \bar{X}^{(l)}(t)\right\}, 
	\end{align*}
	\begin{align}\label{eq:cov_source}
		\hat K^{\mathcal A}(s,t) = \sum_{l=1}^L \pi_l \hat K^{(l)}(s,t),
	\end{align}
	where $\bar{X}^{(l)}(t) = n_1^{-1}\sum_{i=1}^{n_1}X_{i}^{(l)}(t)$ and $\pi_l = n_l/ N$, $N = \sum_{l=1}^L n_l$. 
	Note that
	\begin{align}\label{eq:fpca}
		\hat K^{\mathcal A}(s,t) = \sum_k \hat{\lambda}_k^{\mathcal A} \hat{\phi}_k^{\mathcal A}(s) \hat{\phi}_k^{\mathcal A}(t), 
	\end{align}
	where $\hat{\lambda}_1^{\mathcal A} \ge \hat{\lambda}_2^{\mathcal A} \ge \cdots$ are eigenvalues and $\hat{\phi}_1^{\mathcal A}, \hat{\phi}_2^{\mathcal A}, \cdots$ are corresponding eigenfunctions. 
	Let 
	\begin{align}\label{eq:scores}
		\begin{split}
			\hat{\xi}_{i,k}^{(l)} & = \int_{\tdomain} \big(X_{i}^{(l)}(t) - \bar{X}^{(l)}(t)\big) \hat{\phi}_k^{\mathcal A}(t) dt, ~ i=1,\dots, n_l, \\
			\hat{\xi}_{i,k} & = \int_{\tdomain} \big(X_{i}(t) - \bar{X}(t)\big) \hat{\phi}_k^{\mathcal A}(t) dt, ~ i=1, \dots, n,
		\end{split}
	\end{align}
	where $\bar{X}(t) = n^{-1} \sum_{i=1}^{n}X_i(t)$. 
	
	We introduce the following notations for vectors and matrices. 
	Let $\bY^{(l)} = (Y_1^{(l)}, \dots, Y_{n_l}^{(l)})^{\t}$ and $\bar{\bY}^{(l)}$ be an $n_l$-dimensional vector with each element equal to $\bar{Y}^{(l)} = n_l^{-1}\sum_{i=1}^{n_l} Y_i^{(l)}$. Define $\boldsymbol{Y}$ and $\bar{\bY}$ similarly for the target sample.
	Denote
	\[ \hat{\Xi}^{(l)} = \left( \begin{array}{ccc}
		\hat{\xi}_{1,1}^{(l)} & \cdots & \hat{\xi}_{1,m}^{(l)} \\
		\hat{\xi}_{2,1}^{(l)} & \cdots & \hat{\xi}_{2,m}^{(l)} \\
		\vdots & \ddots & \vdots \\
		\hat{\xi}_{n_l,1}^{(l)} & \cdots & \hat{\xi}_{n_l,m}^{(l)}
	\end{array} \right),\] 
	%\[
	%\hat{\Xi} = \left( \begin{array}{ccc}
		%	\hat{\xi}_{1,1} & \cdots & \hat{\xi}_{1,m} \\
		%	\hat{\xi}_{2,1} & \cdots & \hat{\xi}_{2,m} \\
		%	\vdots & \ddots & \vdots \\
		%	\hat{\xi}_{n,1} & \cdots & \hat{\xi}_{n,m}
		%\end{array} \right), \]
		where the truncation parameter $m$ is allowed to grow with the sample sizes. Define $\hat \Xi \in \real^{n\times m}$ analogously.
		
		We use the approximated least square method to transfer knowledge from the source data and obtain an initial estimator. Since the slope functions of the target model and the source models are generally different, this initial estimator may be biased. To correct the bias, we use the target data and apply the approximated least square with a lasso penalty. The proposed algorithm is presented in Algorithm \ref{alg:tl-flr}.
		
		\begin{algorithm}[!h]
			\caption{\label{alg:tl-flr} The transfer learning algorithm under posterior drift.}
			\begin{algorithmic}
				\STATE {\bf Input:} Target data $(X_i, Y_i), i=1,\dots, n$ and auxiliary data $(X_i^{(l)}, Y_i^{(l)}), i=1,\dots, n_l; l=1, \dots, L$.
				\STATE Compute $\hat K^{\mathcal A}(s,t)$ as defined in \eqref{eq:cov_source} and obtain the estimates $\hat \lambda_k^{\mathcal A}, \hat \phi_k^{\mathcal A}$ in \eqref{eq:fpca}.
				\STATE Obtain the score variables $\hat \xi_{i,k}$ and $\hat \xi_{i,k}^{(l)}$ in \eqref{eq:scores}, $k=1, \dots, m$.
				\STATE Step 1: An initial estimator. 
				\[ \hat{\bw} = \mathop{\arg\min}_{\bw \in \real^m} \sum_{l=1}^L \pi_l (n_l-1)^{-1} \|\bY^{(l)} - \bar{\bY}^{(l)} - \hat{\Xi}^{(l)} \bw \|_2^2 .\]	
				\STATE Step 2: Bias correction.
				\[ \hat{\bdelta} = \mathop{\arg\min}_{\bdelta \in \real^m} \frac{1}{2n} \|\bY - \bar{\bY} - \hat{\Xi} \hat{\bw} - \hat{\Xi} \bdelta\|_2^2 + \tau \|\bdelta\|_1,  \]
				where $\tau\ge 0$. Let $\hat b_k = \hat w_k + \hat \delta_k$.
				\STATE {\bf Output:} $\hat b(t) = \sum_{k=1}^m \hat b_k \hat \phi_k^{\mathcal A}(t)$.
			\end{algorithmic}
		\end{algorithm}
		
		In Step 1, the true parameter of interest is $\bw = \big( \sum_{l=1}^L \pi_l \Sigma^{(l)})^{-1}\big( \sum_{l=1}^L \pi_l \expect( \bxi^{(l)} Y^{(l)}) \big) $, where $\Sigma^{(l)}$ is a diagonal matrix with elements $\lambda_1^{(l)}, \dots, \lambda_m^{(l)}$ and $\bxi^{(l)} = (\xi_1^{(l)}, \dots, \xi_m^{(l)})^\t$. In step 2, the targeted parameter is $\bdelta = \bb - \bw = \big(\sum_{l=1}^L \pi_l \Sigma^{(l)} \big)^{-1}\sum_{l=1}^L (\pi_l \Sigma^{(l)} \bdelta^{(l)})$, where $\bb=(b_1, \dots, b_m)^{\t}$ and $\bdelta^{(l)} = (\delta_1^{(l)}, \dots, \delta_m^{(l)})$.
		By computation, we obtain
		\begin{align*}
			\hat{\bw} & = \bigg( \sum_{l=1}^L \pi_l (n_l-1)^{-1} \hat{\Xi}^{(l) \t} \hat{\Xi}^{(l)} \bigg)^{-1}  \bigg\{ \\
			& \quad \quad \sum_{l=1}^L \pi_l (n_l-1)^{-1} \hat{\Xi}^{(l) \t} (\bY^{(l)}-\bar{\bY}^{(l)}) \bigg\}. 
		\end{align*}
		Based on the estimation procedure for $\hat K^{\mathcal A}$ and $\hat \phi_k^{\mathcal A}$, it follows that $\sum_{l=1}^L \pi_l (n_l-1)^{-1} \hat{\Xi}^{(l) \t} \hat{\Xi}^{(l)}$ is a diagonal matrix. 
		However, the $\hat \xi_{i,k}$'s in Step 2 may be correlated, meaning that $ \hat \Xi^{\t} \hat \Xi$ is not necessarily a diagonal matrix, which differs from the conventional case in \citet{hall2007methodology}. Furthermore, there is no explicit solution for $\hat \bdelta$ due to the lasso penalty.
		
		\section{Theoretical Properties}\label{sec:thm}
		
		We investigate the theoretical properties of the proposed estimator. To begin with, we provide some necessary conditions. 
		Assumption \ref{assump:xdist} states that the score variables of the target distribution are sub-Gaussian, which is common in the literature on functional and nonparametric analysis \citep{lin2021unified,tian2022transfer}. 
		In Assumption \ref{assump:common}, we assume that the covariance functions share common eigenfunctions \citep{dai2017optimal}. Moreover, Assumption \ref{assump:cov-decay} characterizes the decay rate of the eigenvalues of different covariance functions using the parameter $\alpha$, which simplifies the exposition. Assumptions \ref{assump:common} and \ref{assump:cov-decay} are milder than assuming equal covariate distributions under posterior drift. Assumption \ref{assump:slope-decay} concerns the slope parameters of interest and the contrast between the target model and source models. 
		Assumption \ref{assump:x-moment} about the moments of the source data is quite standard \citep{hall2007methodology}.
		
		\begin{assumption}\label{assump:xdist}
			The score vector $\bxi = (\xi_1, \xi_2, \cdots, \xi_m)^{\t}$ is sub-Gaussian, i.e., $\| \bv^\t \bxi \|_{\psi_2} \le K(\bv^{\t}\Sigma_{\bxi} \bv)^{1/2} $, for some constant $K>0$, any vector $v \in \real^m$ and any integer $m>0$, where $\Sigma_{\bxi}$ is the covariance of $\bxi$.
		\end{assumption}
		
		\begin{assumption}\label{assump:common}
			The covariance functions $K(s,t)$ and $K^{(l)}(s,t)$ share common eigenfunctions.
		\end{assumption}
		
		\begin{assumption}\label{assump:cov-decay}
			Foe some universal constant $c_1>0$, the eigenvalues satisfy $\lambda_k \le c_1 k^{-\alpha}, \lambda_k - \lambda_{k+1} \ge c_1^{-1} k^{-\alpha-1}$ and $\lambda_k^{(l)} \le c_1k^{-\alpha}, \lambda_k^{(l)} - \lambda_{k+1}^{(l)} \ge c_1^{-1} k^{-\alpha-1} $ for $l=1, \dots, L$ and $\alpha>1$.
		\end{assumption}
		
		\begin{assumption}\label{assump:slope-decay}
			For some universal constant $c_2 >0$, assume $b_k \le c_2 k^{-\beta}, \beta > \alpha/2 + 1$ and $\sum_k |\delta_k^{(l)}| \le h$ for $l=1,\dots,L$. 
		\end{assumption}
		
		\begin{assumption}\label{assump:x-moment}
			For some universal constant $c_3>0$, assume $X^{(l)}$ has finite fourth moment, $\int_\tdomain \expect\{ \big(X^{(l)}(t)\big)^4\} dt \le c_3 < \infty $, and $\expect(\xi_k^{(l)})^4 \le c_3 \{\expect(\xi_k^{(l)})^2\}^2$ for all $k$. Moreover, $\expect(\epsilon^{(l)})^4 \le c_3 < \infty$. 
		\end{assumption}
		
		In the context of transfer learning, the truncation parameter $m$ can potentially be much larger than $n$, and the sample correlations between $\hat \xi_{ik}$'s are nonzero. These features distinguish the problem from the classical functional linear regression.
		To tackle these issues, we establish the restricted eigenvalue property under the scenario of functional data, and leverage the oracle inequalities to obtain error bounds. 
		
		Denote $\Xi = (\xi_{i,k}) \in \real^{n \times m}$ and $D$ is a diagonal matrix with elements $\lambda_1^{1/2}, \lambda_2^{1/2}, \dots, \lambda_m^{1/2}$. 
		If $m > n$, the smallest eigenvalue of $\Xi^\t \Xi /n$ is 0. This means that small perturbations in $\|\Xi \bv\|_2/n$ can turn into large changes in $\|\bv\|_2$ for $\bv \in \real^m$, leading to unstable solutions.
		In Proposition \ref{prop:re}, we establish the connection between $n^{-1}\|\Xi \bv\|_2^2$ and $\|D \bv\|_2^2$ in the functional data setting, which is known as the restricted eigenvalue property in high-dimensional regression analysis \citep{raskutti2010restricted}.
		
		\begin{proposition}\label{prop:re}
			Under Assumption \ref{assump:xdist}, for any vector $\bv \in \real^m$, there exists some constant $c_K$ depending on the sub-Gaussian parameter $K$, such that 
			\begin{equation}\label{eq:re}
				\frac{1}{n}\|\Xi \bv \|_2^2 \ge \frac{1}{4} \|D \bv \|_2^2 - \frac{c_K\|D\|_F}{n^{1/2}}\|\bv\|_1 \|D\bv \|_2 , 
			\end{equation}
			with probability at least $1-c_4\exp(-c_5n)$ for some constants $c_4, c_5>0$.
		\end{proposition}
		
		With the result in Proposition \ref{prop:re}, we quantify the estimation error via oracle inequalities in Theorem \ref{thm:upper}. Theorem \ref{thm:upper} reveals an interesting phenomenon indicating that the sparsity parameter $\tau$ plays a crucial role.
		If we take $\tau \asymp n^{-1/2}$, the rate of convergence consists of several components. First, the error term $m^{1-2\beta}$ represents the bias caused by truncation. Second, the bound $(m^{1+\alpha} + m^3h^2)N^{-1}$ quantifies the estimation error of the initial estimator to its probabilistic limit in Step 1. These two components are standard in functional linear regression \citep{hall2007methodology} with $h=0$. The extra term $\big(n^{-1/2}m^\alpha  h\big) \land h^2$ characterizes the error in Step 2. If we set $\tau = 0$, the rate of convergence is consistent with that in classical functional linear regression using only target data.
		
		\begin{theorem}\label{thm:upper}
			Suppose Assumptions \ref{assump:xdist}-\ref{assump:x-moment} hold. 
			If $\tau \asymp n^{-1/2}$, $N^{-1}m^{2(\alpha+1)} = o(1)$ and $h = O(1)$,  then
			\[ \|\hat b - b\|_2^2 = O_P\left( \frac{m^\alpha h}{n^{1/2}} \land h^2 + \frac{m^{1+\alpha}}{N} + m^{1-2\beta} + \frac{m^3h^2}{N}  \right). \]
			If $\tau=0$, $n^{-1}m^2=o(1)$ and $N^{-1}m^{\alpha+3}=o(1)$, then 
			\[ \|\hat b - b\|_2^2 = O_P\left( \frac{m^{1+\alpha}}{n} + \frac{m^{1+\alpha}}{N} + m^{1-2\beta} \right). \]
		\end{theorem}
		
		The truncation parameter $m$ plays an important role in the final convergence rate, which is determined by the bias-variance trade-off. In Corollary \ref{cor:rate}, we elucidate the choice of the parameter $m$ and the corresponding rate of convergence under different bias levels $h$.
		Compared to the minimax rate $n^{-(2\beta-1)/(\alpha+2\beta)}$ in the conventional functional linear regression, we identify when the transfer learning algorithm improves estimation performance. When $N \gg n$ and $h \ll n^{-(\beta-1/2)/(\alpha+2\beta)}$, the obtained rate of convergence is faster than that of the classical estimator using only target data. This is intuitive, as effective information transfer from the source data is possible only when there are sufficient source data and the source models are sufficiently close to the target model. 
		Note that the truncation level $m$ can be larger than $n$ by taking advantage of the additional information from the source data.   
		
		\begin{corollary}\label{cor:rate}
			Suppose Assumptions \ref{assump:xdist}-\ref{assump:x-moment} hold. Assume $n \lesssim N$.
			\begin{itemize}
				\item If $h \lesssim N^{-\frac{2\beta-1}{2(\alpha+2\beta)}}$,  we take $\tau \asymp n^{-1/2}$ and $ m \asymp N^{1/(\alpha+2\beta)}$, then
				\[ \|\hat b - b\|_2^2 = O_P\big( N^{-\frac{2\beta-1}{\alpha+2\beta}} \big). \]
				\item If $ N^{-\frac{2\beta-1}{2(\alpha+2\beta)}} \lesssim  h \lesssim n^{-\frac{2\beta-1}{2(\alpha+2\beta)}}$, we take $\tau \asymp n^{-1/2}$ and $  h^{-2/(2\beta-1)} \lesssim m \lesssim (Nh^2)^{1/(1+\alpha)}$, $ m^{2(\alpha+1)}N^{-1}=o(1)$, then
				\[ \|\hat b - b\|_2^2 = O_P( h^2 ). \]
				\item If $ h \gtrsim n^{-\frac{2\beta-1}{2(\alpha+2\beta)}}$, we take $\tau = 0$ and $ m \asymp n^{1/(\alpha+2\beta)} $, then
				\[ \|\hat b - b\|_2^2 = O_P\big( n^{-\frac{2\beta-1}{\alpha+2\beta}} \big). \]
			\end{itemize}
		\end{corollary}
		
		Corollary \ref{cor:rate} implies that if the bias $h$ is sufficiently small, the lasso penalty can leverage this and help achieve a faster convergence rate. Otherwise, the ordinary least square suffices for bias correction.
		Notably, there is no negative transfer under posterior drift in functional linear regression, which contrasts with other high-dimensional regression problems \citep{li2022transferlasso,tian2022transfer}.
		See the Supplement for a detailed discussion of the pediction performance.
		In practice, the tuning parameters $m$ and $\tau$ can be determined through cross validation.
		
		\section{Adaptive Estimation} \label{sec:adaptive}
		
		In some applications, we may not know whether the functional covariates share aligned eigenspace. If the condition is violated, our estimator may be subject to  performance deterioration due to the unaligned eigenfunctions. Similar phenomena have been observed for the FPCA-based approaches in the conventional functional linear regression \citep{yuan2010reproducing,cai2012minimax}.
		
		To avoid potential performance degradation when covariate distributions have different eigenspaces, we propose an adaptive algorithm consisting of two main steps. First, we construct a collection of candidate source sets and obtain candidate estimators using Algorithm \ref{alg:tl-flr} along with auxiliary samples from these source sets. Second, we perform a sparse aggregation step on these candidate estimators \citep{gaiffas2011hyper}. The aggregated estimator is expected to be not much worse than the best candidate estimator under consideration. 
		
		We randomly split the target data into two subsets $\mathcal D_1 = \{ (X_i, Y_i), i \in \mathcal I_1 \}$ and $\mathcal D_2 = \{ (X_i, Y_i), i \in \mathcal I_2 \}$, where $\tilde{n} = |\mathcal I_1|, \bar{n}=|\mathcal I_2|$.  
		We use $\mathcal D_1$ to construct candidate sets and candidate estimators, and $\mathcal D_2$ for aggregation. 
		
		\subsection{Candidate Source Sets}
		
		To avoid the unstable inverse of covariance operators and truncation operations during the construction of candidate sets, we model the discrepancy between cross covariance functions, defined as $\zeta^{(l)} = \| g^{(l)}- g\|_2^2$, where $g^{(l)}(t) = \cov\big(X^{(l)}(t), Y^{(l)}\big) $ and $g(t) = \cov\big(X(t), Y\big) $. 
		The statistic $\zeta^{(l)}$ is a viable choice as it typically increases with the contrast $\|\bdelta^{(l)}\|_1$ and has the potential to characterize differences in eigenspaces.
		The empirical estimates are given by
		\begin{align}\label{eq:zeta_statistic}
			\begin{split}
				\hat \zeta^{(l)} & = \int_0^1 \bigg\{\frac{1}{n_l} \sum_{i=1}^{n_l} X_i^{(l)}(t)(Y_i^{(l)} -  \bar Y^{(l)}) - \\
				& \quad \quad \frac{1}{|\tilde{n}|} \sum_{i \in \mathcal I_1} X_i(t) (Y_i - \bar Y)  \bigg\}^2 dt.  
			\end{split}
		\end{align}
		Then, we construct the candidate sets as
		\begin{align}\label{eq:candidate_set}
			\hat{\mathcal A}_l = \{ 1 \le k \le L: \hat \zeta^{(k)} \text{is among the first } l \text{ smallest of all} \},
		\end{align} 
		for $l=1, \dots, L$. 
		
		\begin{algorithm}[!h]
			\caption{\label{alg:adap-tl-flr} The adaptive transfer learning algorithm for functional linear regression.}
			\begin{algorithmic}
				\STATE {\bf Input:} Target data $(X_i, Y_i), i=1,\dots, n$ and auxiliary data $(X_i^{(l)}, Y_i^{(l)}), i=1,\dots, n_l$.
				\STATE Randomly split the target data into two sub-samples $\mathcal D_1 = \{ (X_i, Y_i), i \in \mathcal I_1 \}$ and $\mathcal D_2 = \{ (X_i, Y_i), i \in \mathcal I_2 \}$.
				\STATE Construct the $L+1$ candidate sets $\hat{\mathcal A_0}, \hat{\mathcal A_1}, \dots, \hat{\mathcal A_L}$ such that $\hat{\mathcal A_0} = \emptyset$ using \eqref{eq:zeta_statistic} and \eqref{eq:candidate_set}. 
				\FOR{$l = 0, \dots, L$}  
				\STATE Obtain the candidate estimator $\hat b_l(t)$ using data $\mathcal D_1 \cup \{ (X_i^{(l)}, Y_i^{(l)}), i=1, \dots, n_l;  l \in \hat{\mathcal A_l} \}$ with Algorithm \ref{alg:tl-flr}.
				\ENDFOR
				\STATE Sparse aggregation:
				\[ \hat b_{sagg}(t) = \hat \lambda \hat b_{l_{1,\star}}(t) + (1-\hat \lambda) \hat b_{l_{2,\star}} (t), \]
				where $l_{1, \star}, l_{2, \star}$ and $\hat \lambda$ are obtained from \eqref{eq:sagg_index_erm} and \eqref{eq:sagg_other}.
				\STATE {\bf Output:} $\hat b_{sagg} (t)$.
			\end{algorithmic}
		\end{algorithm}
		
		\subsection{Sparse Aggregation}
		
		For each constructed candidate set $\hat{\mathcal A}_l$, we obtain the candidate estimator $\hat b_l(t)$ using $\mathcal D_1 \cup \hat{\mathcal A}_l$ with Algorithm \ref{alg:tl-flr} for $l=0, 1, \dots, L$. Motivated by the principle of model aggregation, which aims to obtain an estimator not much worse than the best estimator under consideration, we use sparse aggregation to achieve adaptivity and prevent negative transfer.
		To reduce computational cost, we adopt sparse aggregation without the preselection step as in \citet{gaiffas2011hyper}.
		
		Define $R_{n,2}(b) = \sum_{i \in \mathcal I_2} \{Y_i - \bar{Y}_2- \int_0^1 b(t) \big(X_i(t) -\bar{X}_2(t)\big)dt\}^2 / \bar{n}$, where $\bar{Y}_2 =\sum_{i \in \mathcal I_2}Y_i/ \bar{n}$ and $\bar{X}_2(t) = \sum_{i \in \mathcal I_2} X_i(t)/ \bar{n}$. 
		Denote
		\begin{align}\label{eq:sagg_index_erm}
			l_{1,\star} = \mathop{\arg\min}_{ l=0, 1, \dots, L} R_{n,2}(\hat b_l).
		\end{align}
		Moreover,
		\begin{align}\label{eq:sagg_other}
			\hat \lambda, l_{2,\star} = \mathop{\arg\min}_{\lambda \in [0,1], l=0, \dots, L}  R_{n,2}( \lambda \hat b_{l_{1, \star}} + (1-\lambda) \hat b_l). 
		\end{align}
		The sparse aggregate estimator is $\hat b_{sagg}(t) = \hat \lambda \hat b_{l_{1,\star}}(t) + (1-\hat \lambda) \hat b_{l_{2,\star}} (t)$.
		
		\begin{remark}
			In addition to sparse aggregation, other methods such as exponential aggregation \citep{rigollet2011exponential} and Q-aggregation \citep{dai2012deviation} can be utilized for adaptive estimation. However, their aggregate performance heavily depends on the critical temperature parameter. In our numerical experiments, we compare their performance by tuning the temperature parameter through 5-fold cross-validation against the performance of sparse aggregation. 
			The detailed results are provided in the Supplementary Material. 
			Based on numerical results, we advocate for sparse aggregation because of its competitive, robust performance and computational efficiency, which does not require parameter tuning.
		\end{remark}
		
		\section{Synthetic Data}\label{sec:sim}
		
		We conduct several experiments to demonstrate the performance of the transfer learning algorithm. 
		Let $\tdomain=[0,1]$, $n=150$, $n_l=100$ and the total number of source samples $L=20$. We take $X_i(t) = \sum_{k=1}^{50} \sqrt{\lambda_k} Z_{ik} \phi_k(t)$, $i=1, \dots, n$, where $\lambda_k = k^{-\alpha}$, $\phi_k(t) = 2^{1/2}\mathrm{cos}(k\pi t)$ for $k \ge 1$, and $Z_{ik}$'s are uniformly distributed on $[-3^{1/2}, 3^{1/2}]$. The target slope function is $b(t) = \sum_{k=1}^{50} b_{k} \phi_k(t)$ with $b_{k} = 4k^{-\beta}(-1)^{k+1}$.  
		For the slope functions in the auxiliary samples, we have $w^{(l)}(t) = \sum_{k=1}^{50} w_k^{(l)} \phi_k(t)$ for $l=1, \dots, L$.  
		The generation mechanisms of $w_k^{(l)}$ and $X^{(l)}$ will be provided later.
		The responses $Y$ and $Y^{(l)}$ are generated from $\eqref{eq:target_model}$ and $\eqref{eq:source_model}$, respectively, with $\expect Y = \expect Y^{(l)} = 0$, $\mu(t) = \mu^{(l)}(t) = 0$ and the errors $\epsilon, \epsilon^{(l)} \sim N(0, \sigma_{\epsilon}^2)$ where $\sigma_\epsilon = 0.5$.
		%For the slope functions in the auxiliary samples $w^{(l)}(t) = \sum_{k=1}^{50} w_k^{(l)} \phi_k(t)$,   Let $w_k^{(l)} = b_k - R_k h/s $, where $R_k$'s are independent Rademacher random variables.
		%	we consider two configurations.
		%	\begin{itemize}
			%		\item  Let $w_k^{(l)} = b_k - R_k h/50 $, where $R_k$'s are independent Rademacher random variables.
			%		\item Let $w_k^{(l)} = b_k - \zeta_k^{(l)}$, where $\zeta_k^{(l)} \stackrel{i.i.d.}{\sim} \text{Laplace}(0, h/50)$ such that $\sum_{k=1}^{50} \expect |w_k^{(l)} - b_k| = h$.
			%	\end{itemize}
		Denote $\mathcal A_h = \{ 1\le l \le L: \|\bdelta^{(l)}\|_1 \le h \}$.
		
		\subsection{Transfer Learning on $\mathcal A_h$ with Aligned Eigenspace}
		
		We explore the numerical performance of Algorithm \ref{alg:tl-flr} when the eigenspaces of the target sample and source samples are aligned.
		
		\begin{itemize}
			\item[(I)] For $l \in \mathcal A_h$, let $X_i^{(l)}(t) = \sum_{k=1}^{50} \sqrt{\lambda_k} Z_{ik}^{(l)} \phi_k(t)$ for $i=1, \dots, n_l$ and $w_k^{(l)} = b_k - R_k h/s $ for $k=1, \dots, 50$, where $Z_{ik}^{(l)}$'s are generated from $N(0,1)$, $s$ is some positive integer within $[1, 50]$ and $R_k$'s are independent Rademacher random variables.  
		\end{itemize}
		
		Let $K=|\mathcal A_h|$. We consider different combinations of $s=1, 5, 20, 50$ and $h=2, 20, 200, 2000$, respectively.
		The proposed method in Algorithm \ref{alg:tl-flr} is denoted by ``$\mathcal A_h$ TL-FLR".
		For comparison purposes,  we include the FPCA-based estimation \citep{hall2007methodology} using only target data, which is denoted by ``FLR".
		We evaluate the performance using the mean integrated squared error (MISE), i.e., $\int_0^1 (\hat b(t) - b(t))^2 dt$, which is approximated on a grid of 100 equally spaced points on $[0,1]$.
		The tuning parameters of all considered approaches are chosen by 5-fold cross-validation. 
		
		The results under $h=2000$ are deferred to the Supplementary Material. 
		As shown in Figure \ref{fig:model-I}, the ``$\mathcal A_h$ TL-FLR" consistently outperforms ``FLR" across all considered scenarios, even when the contrast is extremely large. The numerical results align with theoretical findings in Theorem \ref{thm:upper} and Corollary \ref{cor:rate}. In functional linear regression, when the eigenspaces of the target sample and source samples are aligned, the proposed transfer learning algorithm avoids negative transfer, meaning its performance is at least as good as that of the ``FLR" using only target data. 
		
		\begin{figure}
			\centering
			\newcommand{\thiswidth}{0.3\linewidth}
			\newcommand{\thisgap}{0mm}
			\begin{tabular}{ccc}
				\hspace{\thisgap}\includegraphics[width=\thiswidth]{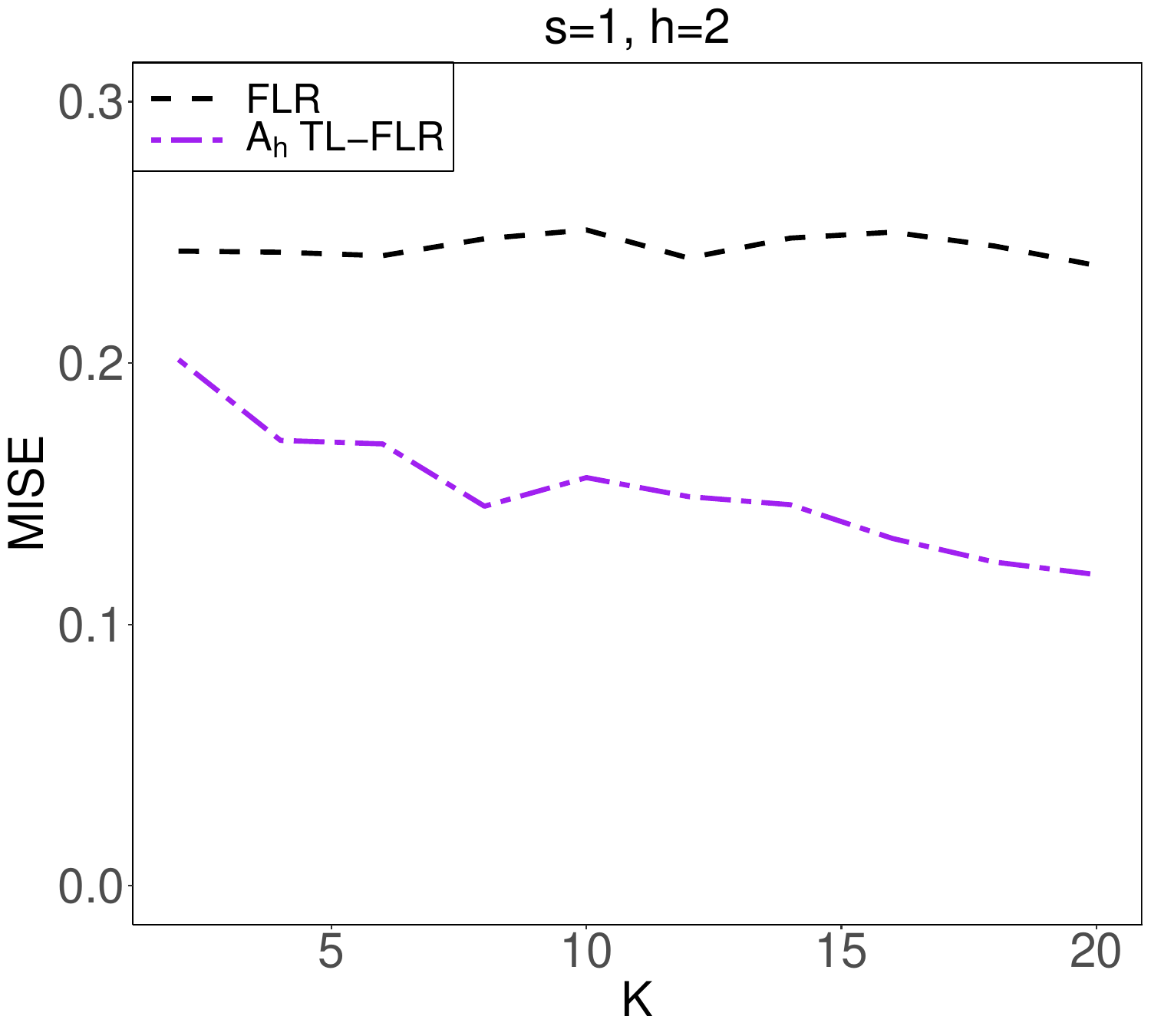} &
				\hspace{\thisgap}\includegraphics[width=\thiswidth]{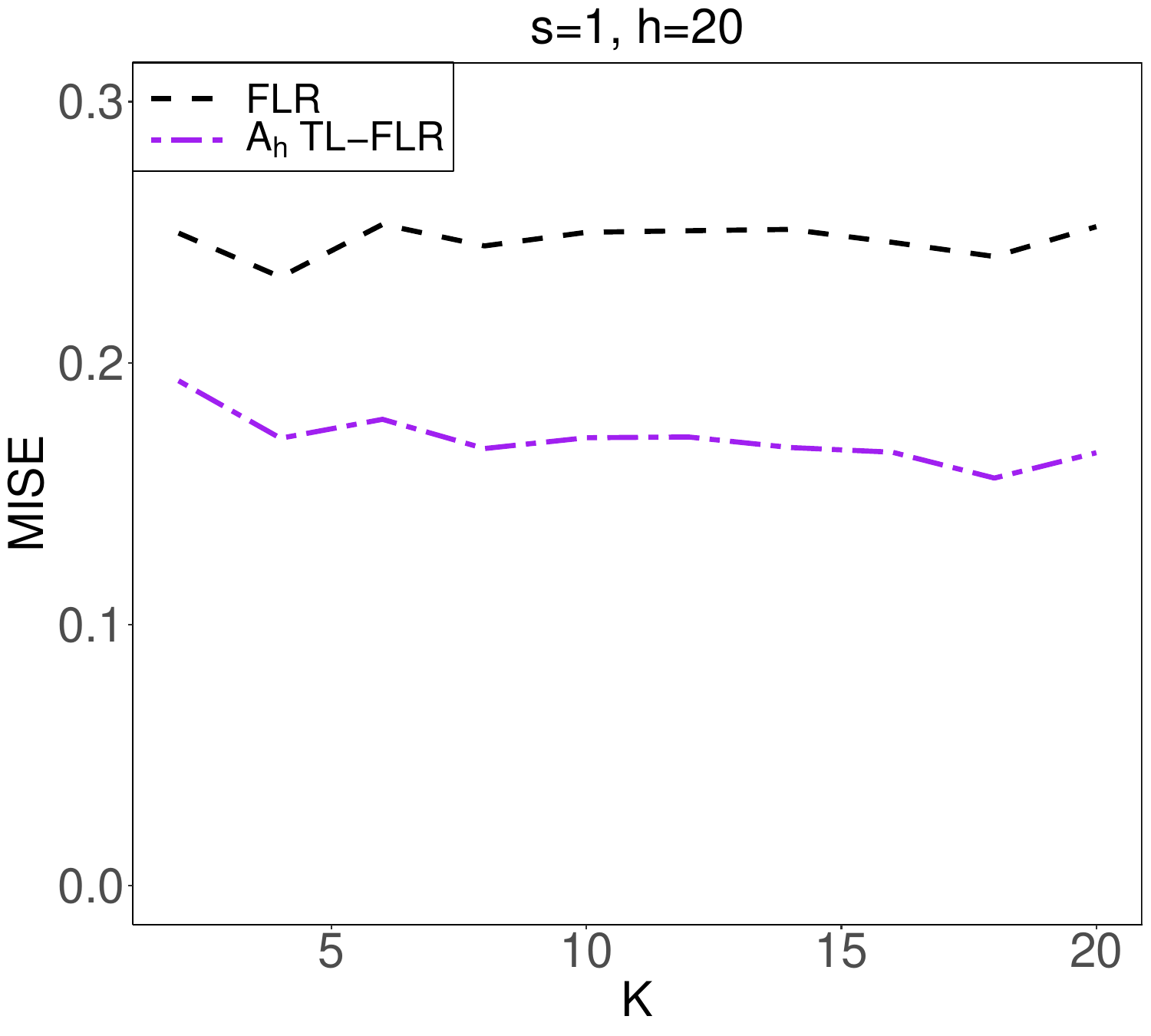} &
				\hspace{\thisgap}\includegraphics[width=\thiswidth]{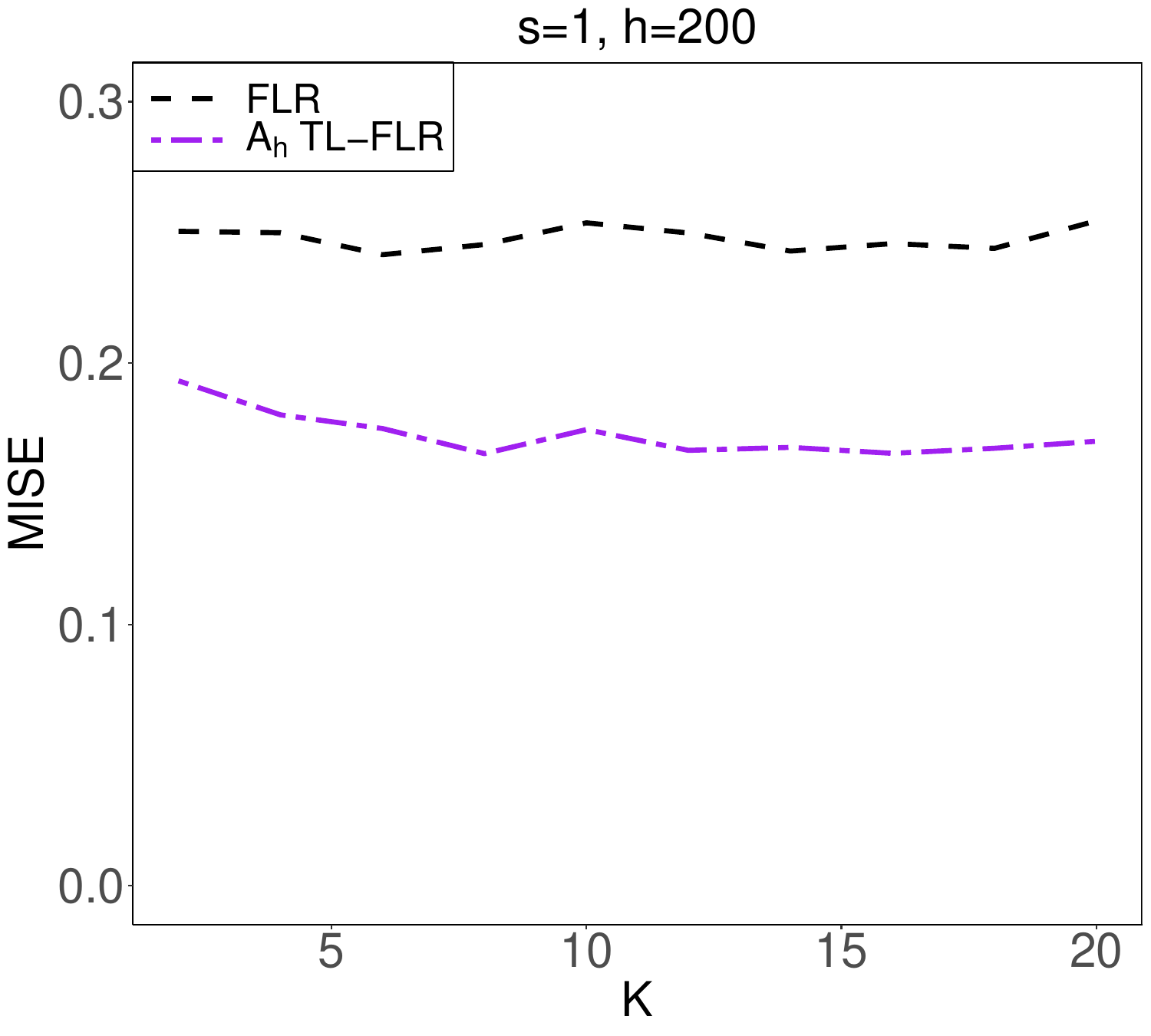} \\
				%		\hspace{\thisgap}\includegraphics[width=\thiswidth]{figs_AAAI/MISE_same_exact_TRUE_ndiffcoef_1_n_150_alpha_1.5_beta_2_h_2000_noise_0.5}  \\
				\hspace{\thisgap}\includegraphics[width=\thiswidth]{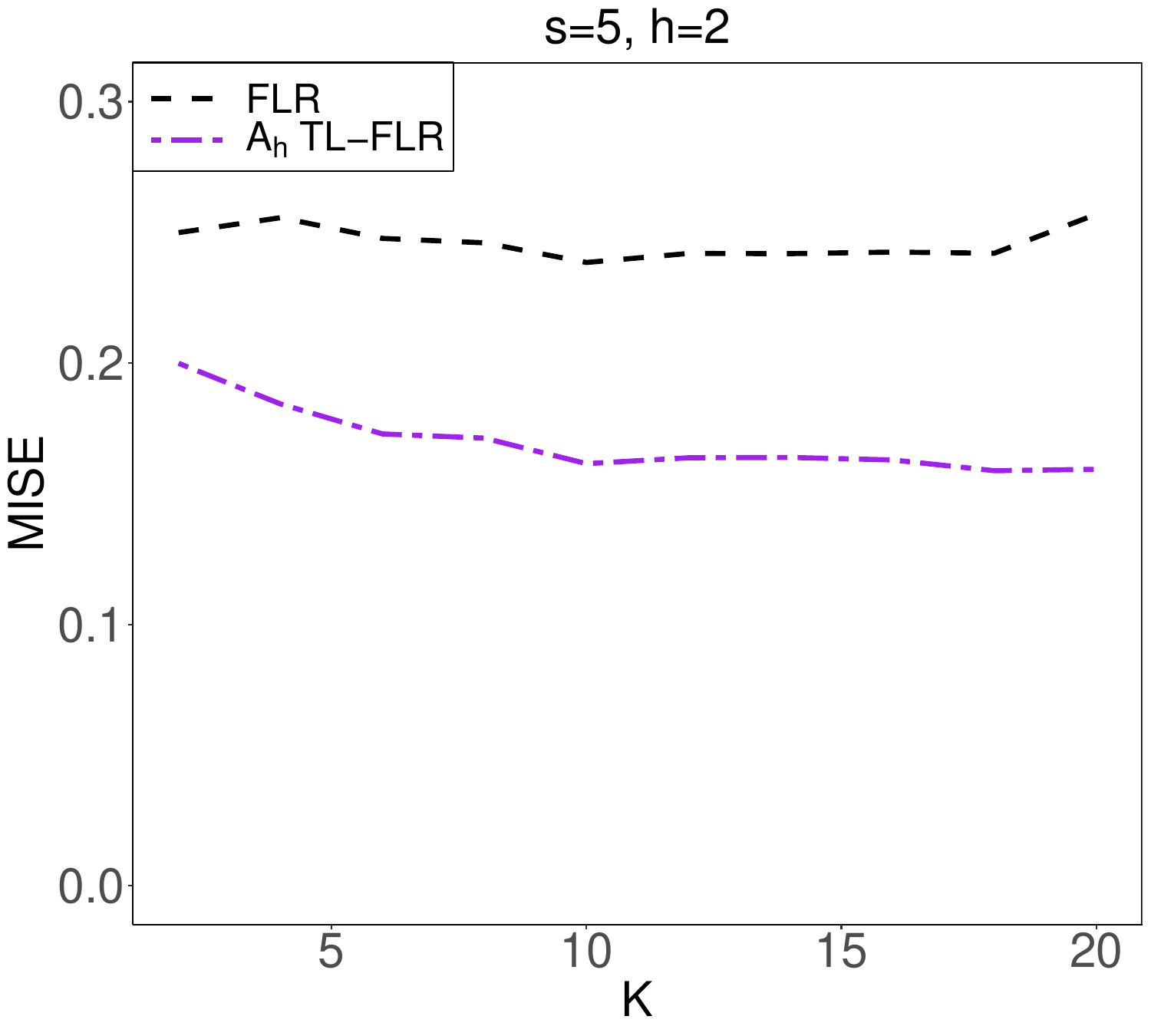} &
				\hspace{\thisgap}\includegraphics[width=\thiswidth]{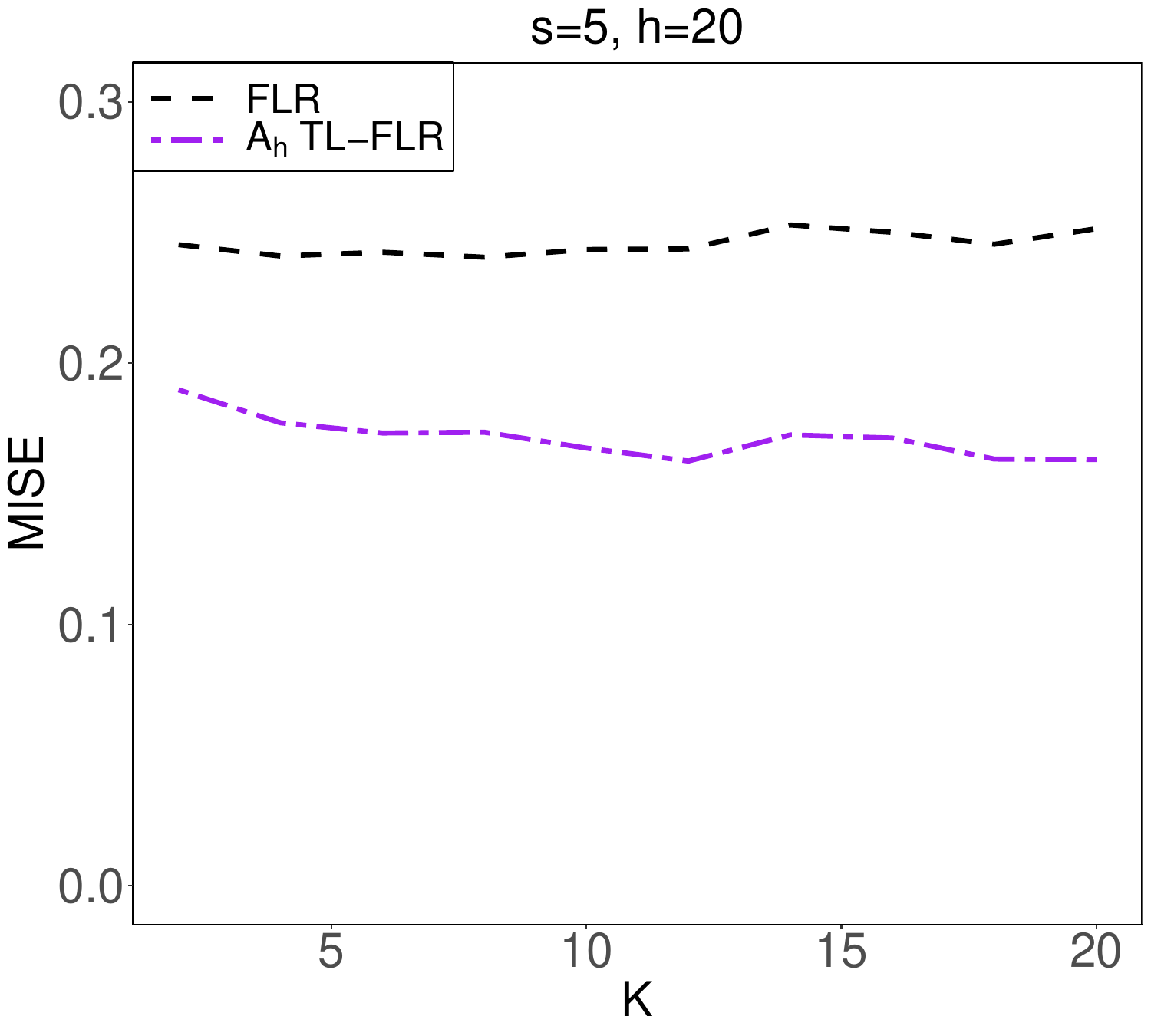} &
				\hspace{\thisgap}\includegraphics[width=\thiswidth]{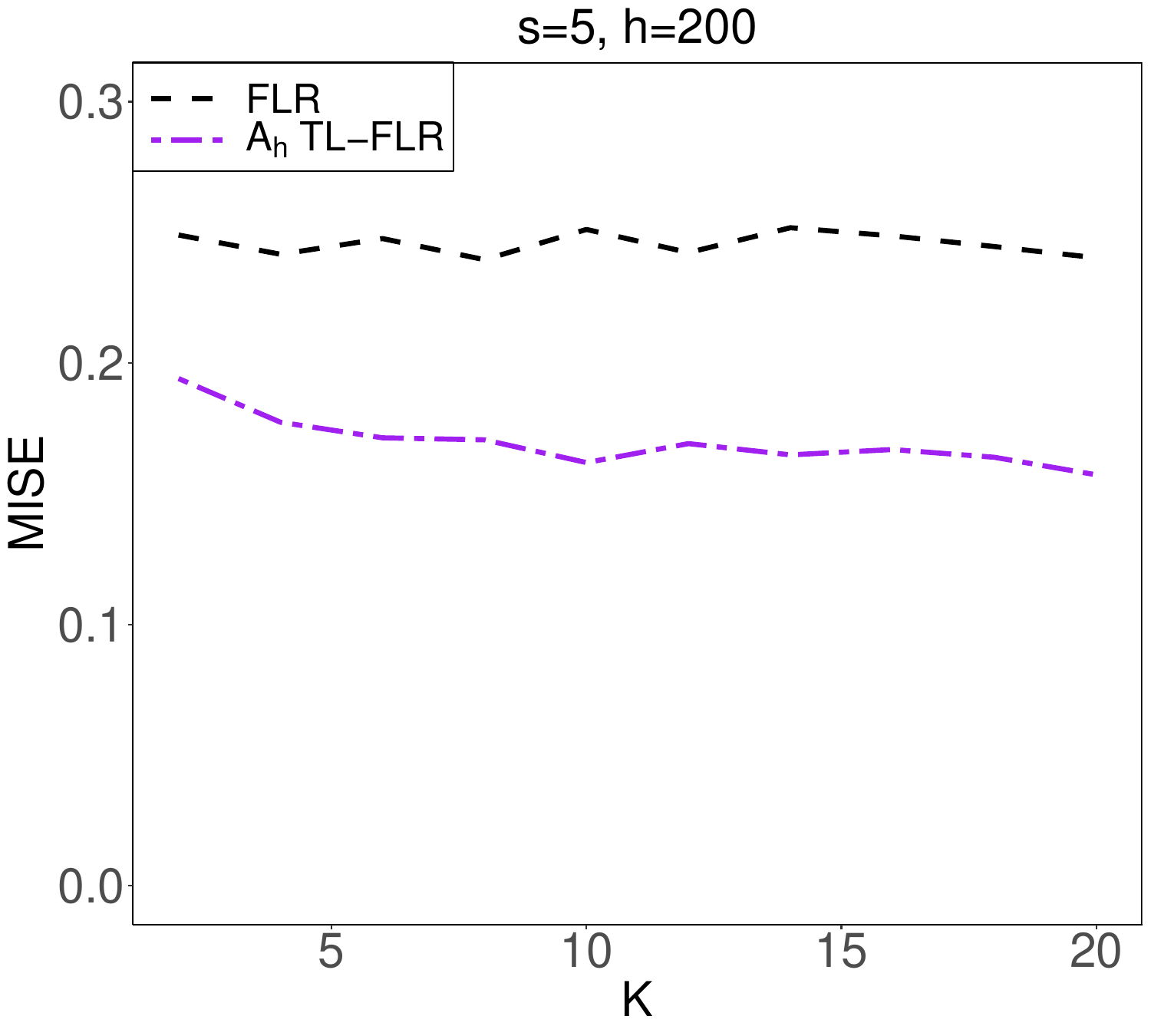} \\
				%		\hspace{\thisgap}\includegraphics[width=\thiswidth]{figs_AAAI/MISE_same_exact_TRUE_ndiffcoef_5_n_150_alpha_1.5_beta_2_h_2000_noise_0.5}  \\
				\hspace{\thisgap}\includegraphics[width=\thiswidth]{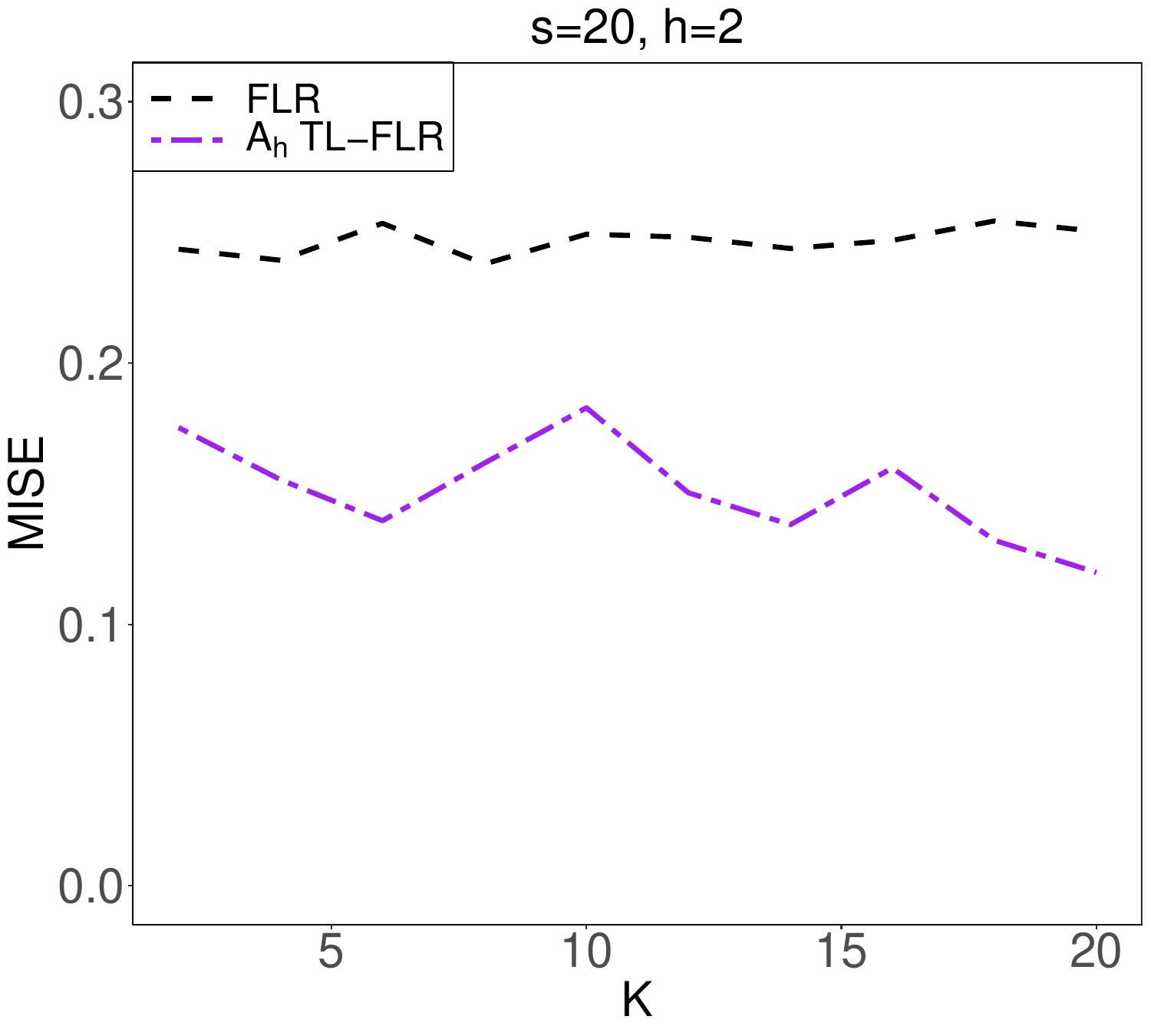} &
				\hspace{\thisgap}\includegraphics[width=\thiswidth]{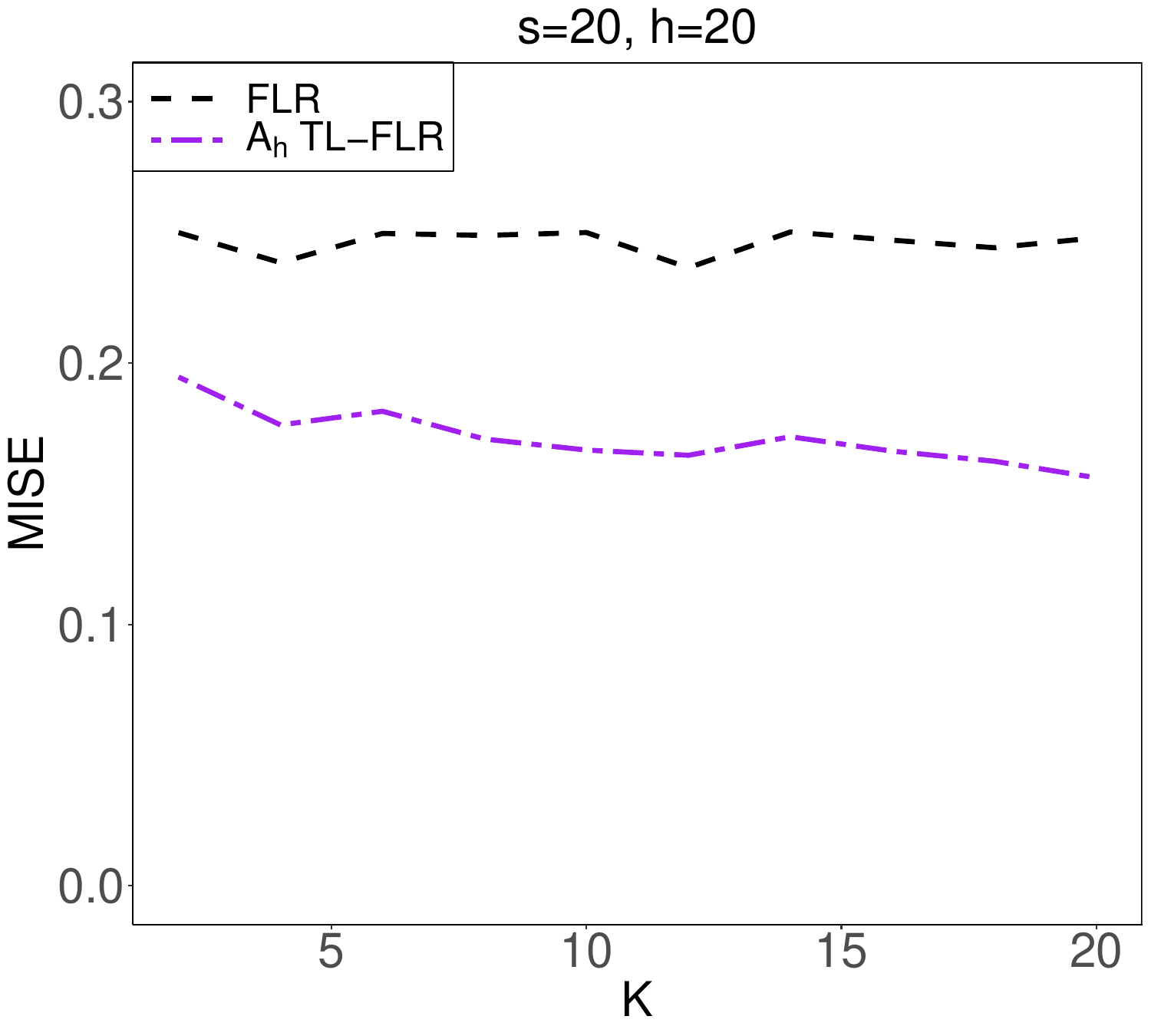} &
				\hspace{\thisgap}\includegraphics[width=\thiswidth]{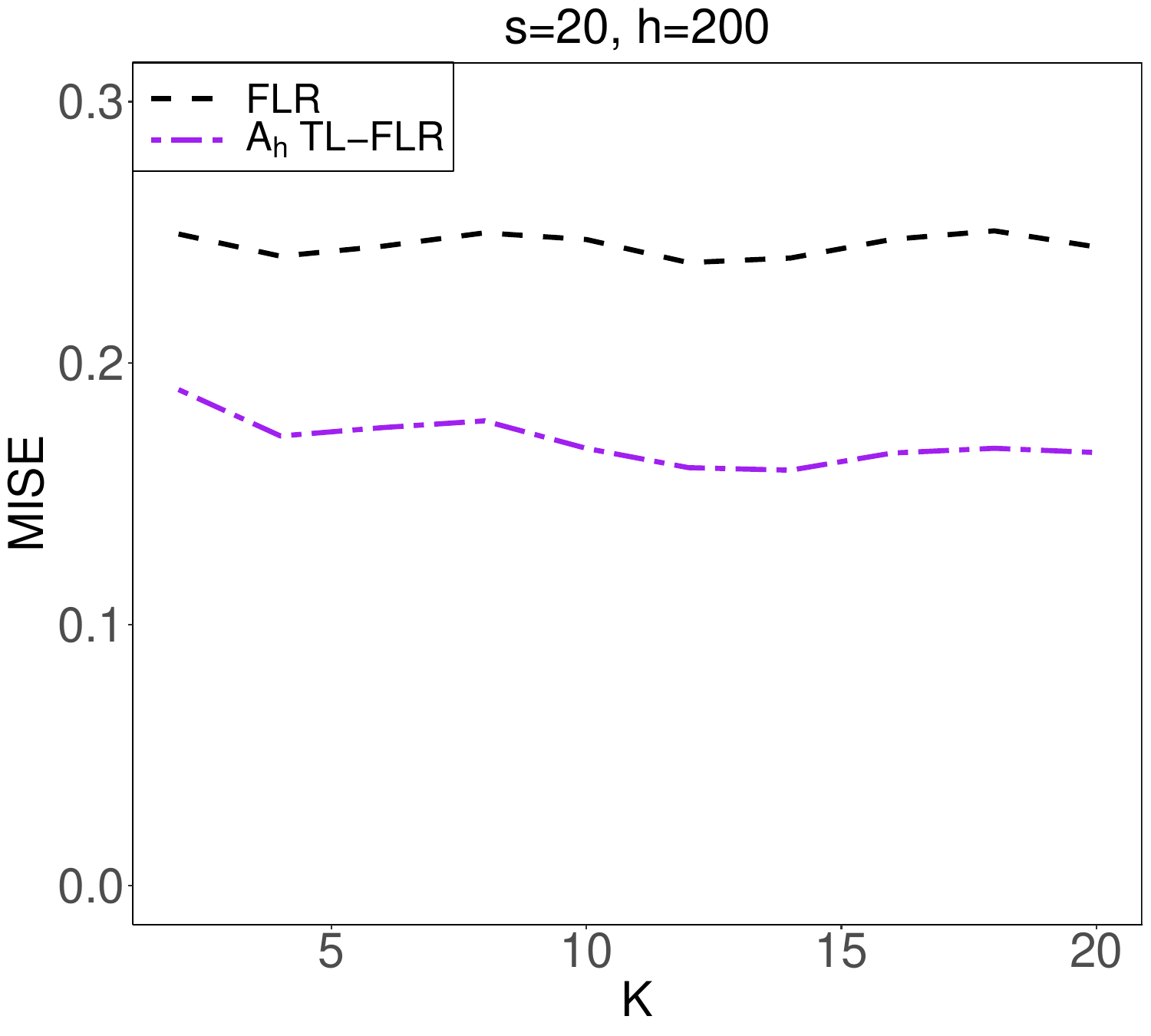} \\
				%		\hspace{\thisgap}\includegraphics[width=\thiswidth]{figs_AAAI/MISE_same_exact_TRUE_ndiffcoef_20_n_150_alpha_1.5_beta_2_h_2000_noise_0.5}	\\	
				\hspace{\thisgap}\includegraphics[width=\thiswidth]{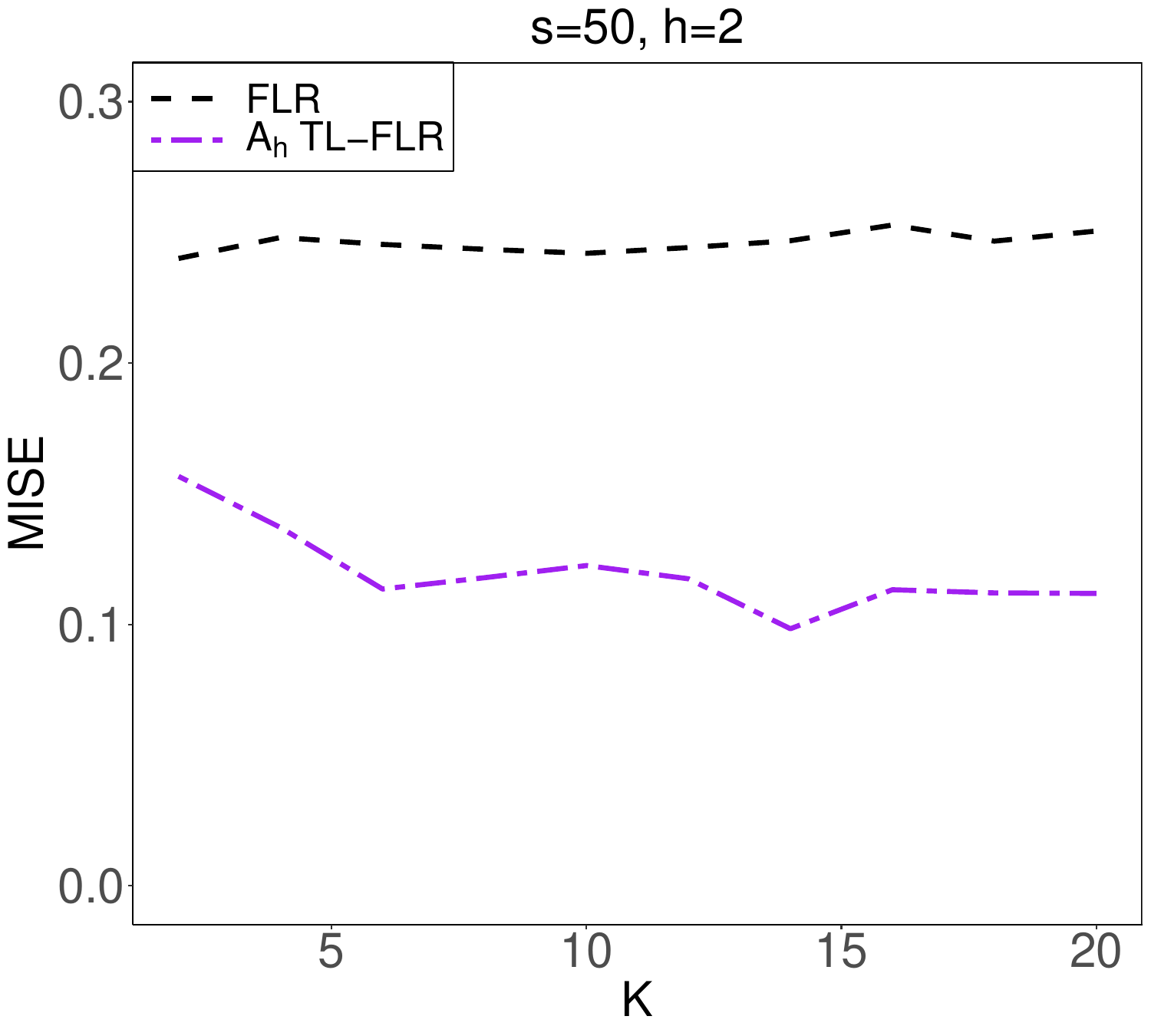} &
				\hspace{\thisgap}\includegraphics[width=\thiswidth]{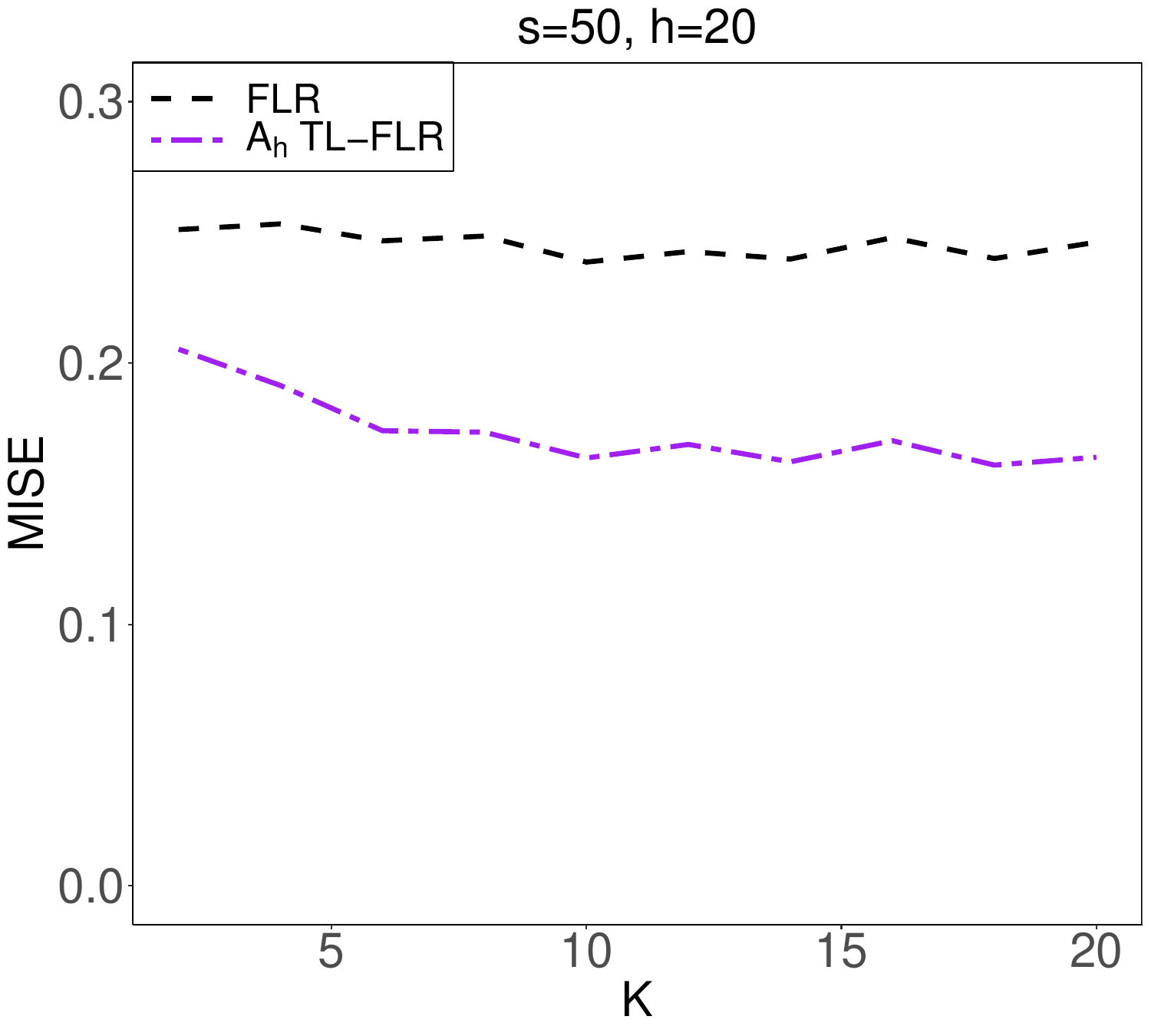} &
				\hspace{\thisgap}\includegraphics[width=\thiswidth]{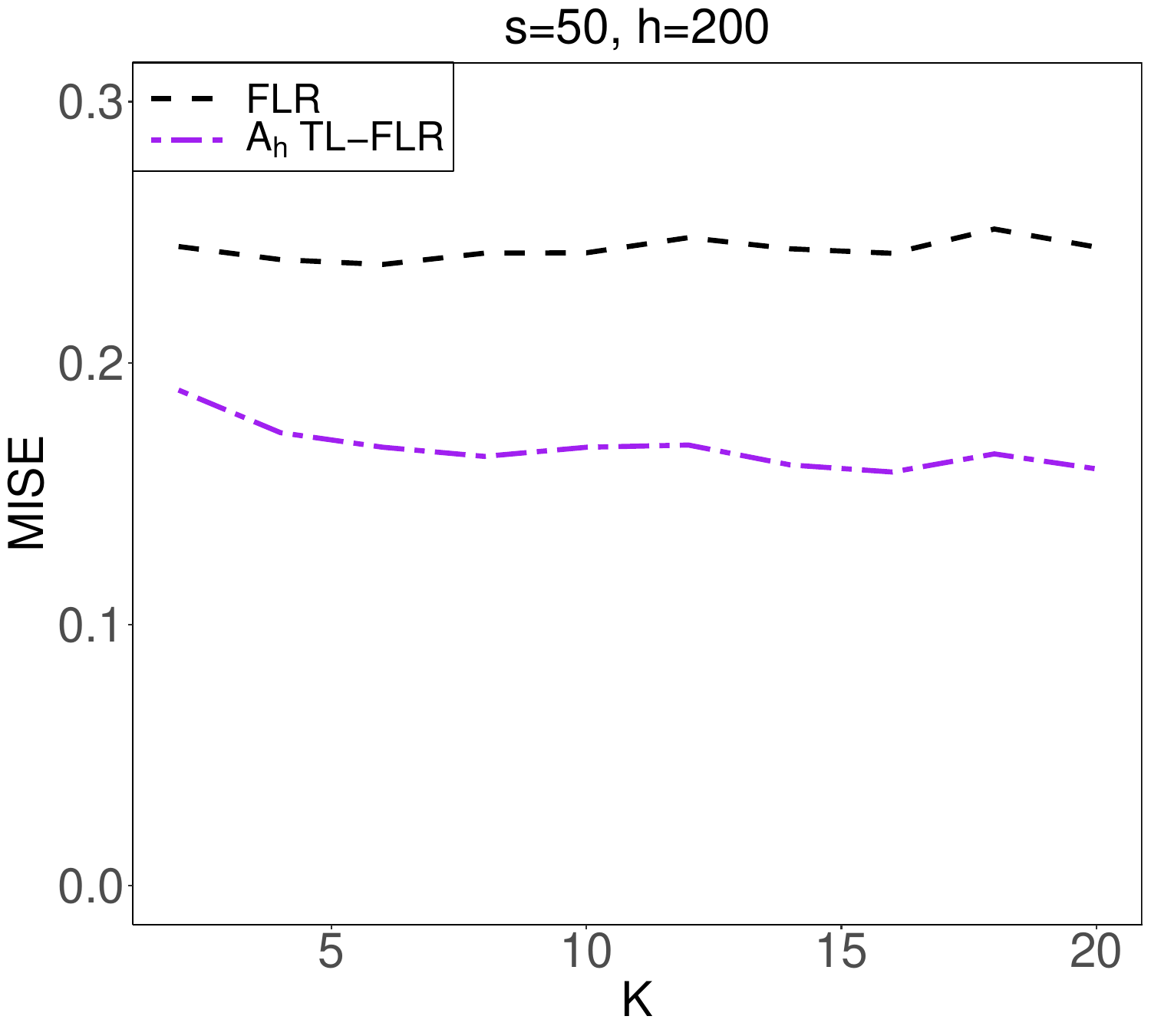} 
				%		\hspace{\thisgap}\includegraphics[width=\thiswidth]{figs_AAAI/MISE_same_exact_TRUE_ndiffcoef_50_n_150_alpha_1.5_beta_2_h_2000_noise_0.5}	
			\end{tabular}
			\caption{Estimation errors of different methods under Model (I) over 1000 repetitions.}
			\label{fig:model-I}
		\end{figure}
		
		\subsection{Adaptive Transfer Learning in General Settings}
		
		We further include the adaptive method in Algorithm \ref{alg:adap-tl-flr}, denoted by ``Agg TL-FLR", and the naive transfer learning method using all source samples, denoted by ``Naive TL-FLR" for comparison under various data generation mechanisms.
		
		\begin{itemize}	
			\item[(II).] For $l=1, \dots, L$, let $X_i^{(l)}(t) = \sum_{k=1}^{50} \sqrt{\lambda_k} Z_{ik}^{(l)} \phi_k(t)$, $i=1, \dots, n_l$, where $Z_{ik}^{(l)}$'s are generated from $N(0,1)$. For $l \in \mathcal A_h$, $w_k^{(l)} = b_k - R_k h/50 $; for $l \in \mathcal A_h^c$, $w_k^{(l)} = b_k - 40R_k $, where $R_k$'s are independent Rademacher random variables.
			
			\item[(III)] For $l \in \mathcal A_h$, let $X_i^{(l)}(t) = \sum_{k=1}^{50} \sqrt{\lambda_k} Z_{ik}^{(l)} \phi_k(t)$; for  $l \in \mathcal A_h^c$, $X_i^{(l)}(t) = \sum_{k=1}^{50} \sqrt{\lambda_k} Z_{ik}^{(l)} \psi_k(t)$, where $\psi_k$'s are Haar functions, i.e.,
			\begin{equation*}
				\psi_{2^j+\ell}(t) = \left\{
				\begin{array}{cc}
					2^{j/2}, & t \in [\frac{\ell}{2^j}, \frac{\ell+0.5}{2^j}), \\
					-2^{j/2}, & t \in [\frac{\ell+0.5}{2^j}, \frac{\ell+1}{2^j}], \\
					0, & \text{otherwise},
				\end{array}
				\right.
			\end{equation*}
			for $j=0, 1, \dots$ and $\ell=0, 1, \dots$. The regression coefficients $w_k^{(l)}$ are the same as in Model (II).
			
			\item[(IV)] For $l=1, \dots, L$, let $X_i^{(l)}(t) = \sum_{k=1}^{50} \sqrt{\lambda_k} Z_{ik}^{(l)} \psi_k(t)$, $i=1, \dots, n_l$, where $\psi_k$'s are the same as in Model (III). The regression coefficients $w_k^{(l)}$ are the same as in Model (II).
		\end{itemize}
		
		Model (II) corresponds to the case where all the source samples share the same eigenspace as the target sample. In contrast, the eigenspaces for each source sample in Model (IV) are completely unaligned with the eigenspace of the target data. In addition, Model (III) represents an intermediate case where the eigenspaces for source samples in $\mathcal A_h$ are aligned with the eigenspace of the target data.
		
		The results are depicted in Figure \ref{fig:adap}. In Model (II) where the eigenspaces are perfectly aligned, the three transfer learning methods consistently outperform ``FLR", as there is no negative transfer under this model. The ``Agg TL-FLR" exhibits the ability to adaptively select the most informative source samples for estimation, leading to the best performance. In Model (III), the method ``Agg TL-FLR" achieves comparable performance to ``$\mathcal A_h$ TL-FLR". However, the ``Naive TL-FLR" leads to inferior performance when $K$ is small, due to the inclusion of the source samples in $\mathcal A_h^c$ whose eigenspaces are misaligned with the eigenspace of the target data.
		In the most challenging case, i.e., Model (IV), as expected, both ``Naive TL-FLR" and ``$\mathcal A_h$ TL-FLR" perform worse than the ``FLR", since the eigenspaces in all source samples are different from the eigenspace of the target data. Even the contrast between $b(t)$ and $w^{(l)}(t)$ for $l \in \mathcal A_h$ is relatively small, the FPCA based estimation fails because of the inaccurate estimation of the leading eigenspace of the target data. In contrast, the adaptive method ``Agg TL-FLR" still performs decently in this extremely challenging scenario, comparable to or a bit worse than the ``FLR".
		
		In summary, the proposed algorithm ``Agg TL-FLR" adaptively selects the truly informative source samples for estimation, effectively alleviating negative transfer even in the case where all source samples are not informative at all.
		
		\begin{figure}
			\centering
			\newcommand{\thiswidth}{0.3\linewidth}
			\newcommand{\thisgap}{0\linewidth}
			\begin{tabular}{ccc}
				\hspace{\thisgap}\includegraphics[width=\thiswidth]{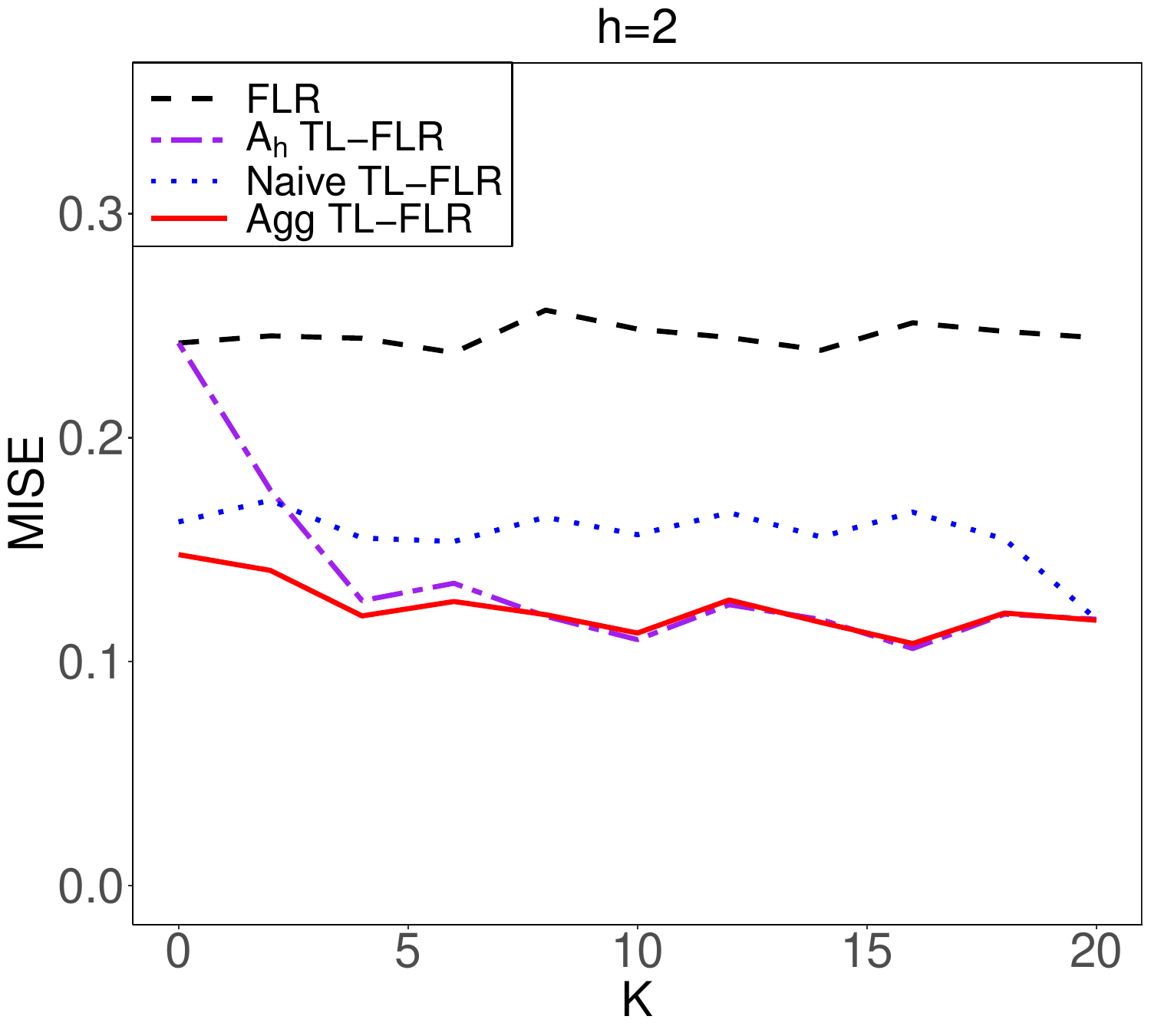} &
				\hspace{\thisgap}\includegraphics[width=\thiswidth]{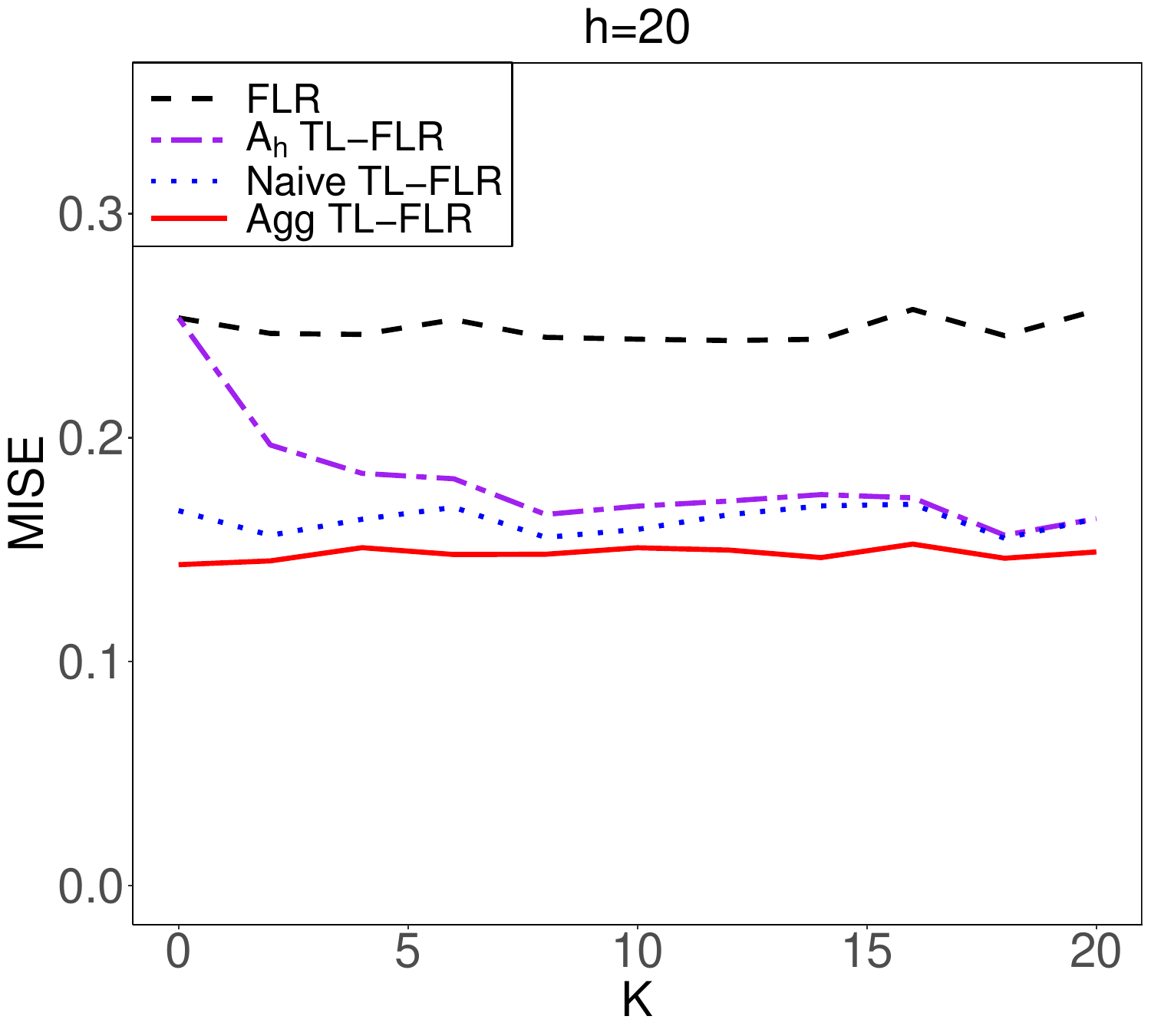} &
				\hspace{\thisgap}\includegraphics[width=\thiswidth]{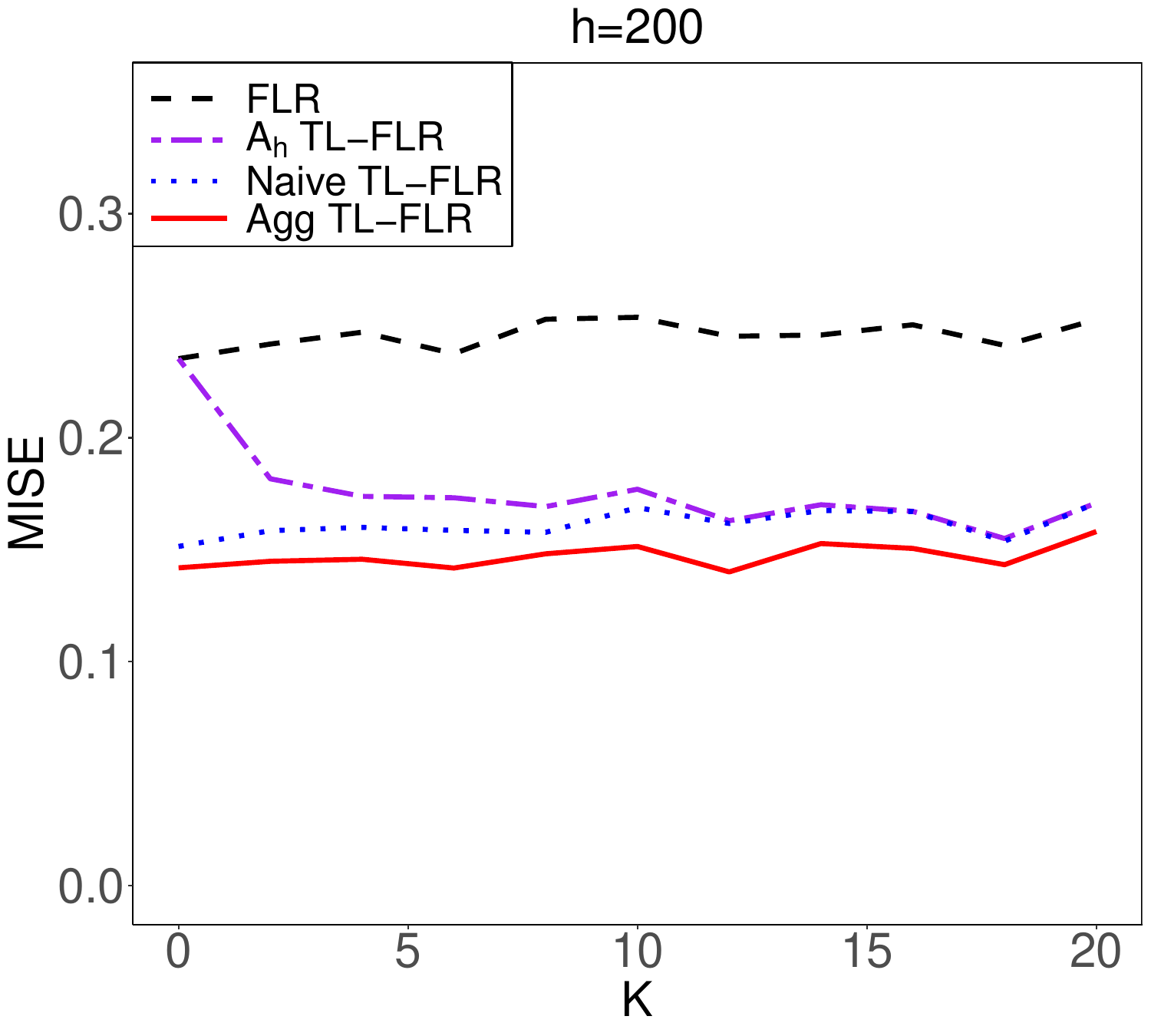} \\
				\hspace{\thisgap}\includegraphics[width=\thiswidth]{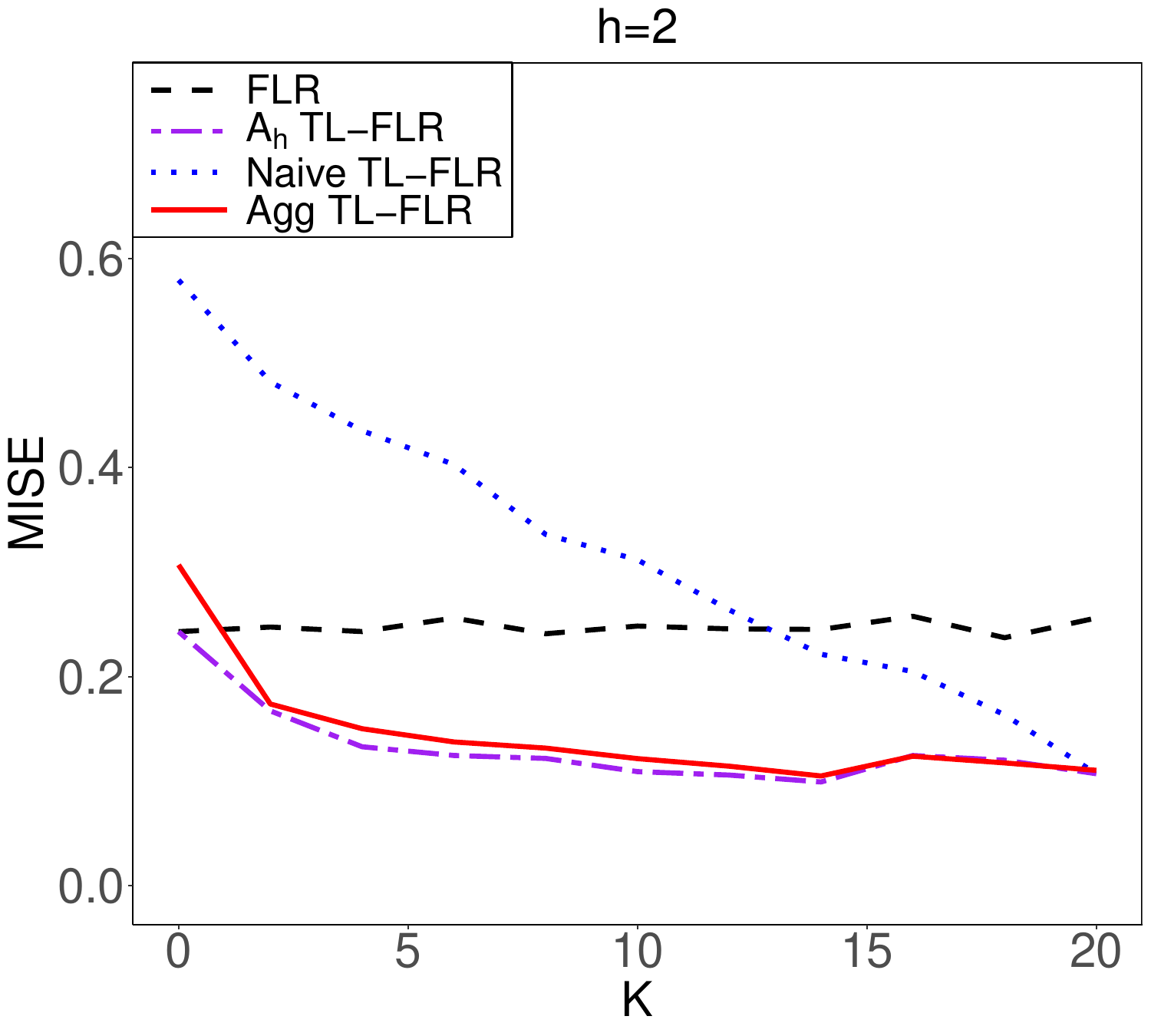} &
				\hspace{\thisgap}\includegraphics[width=\thiswidth]{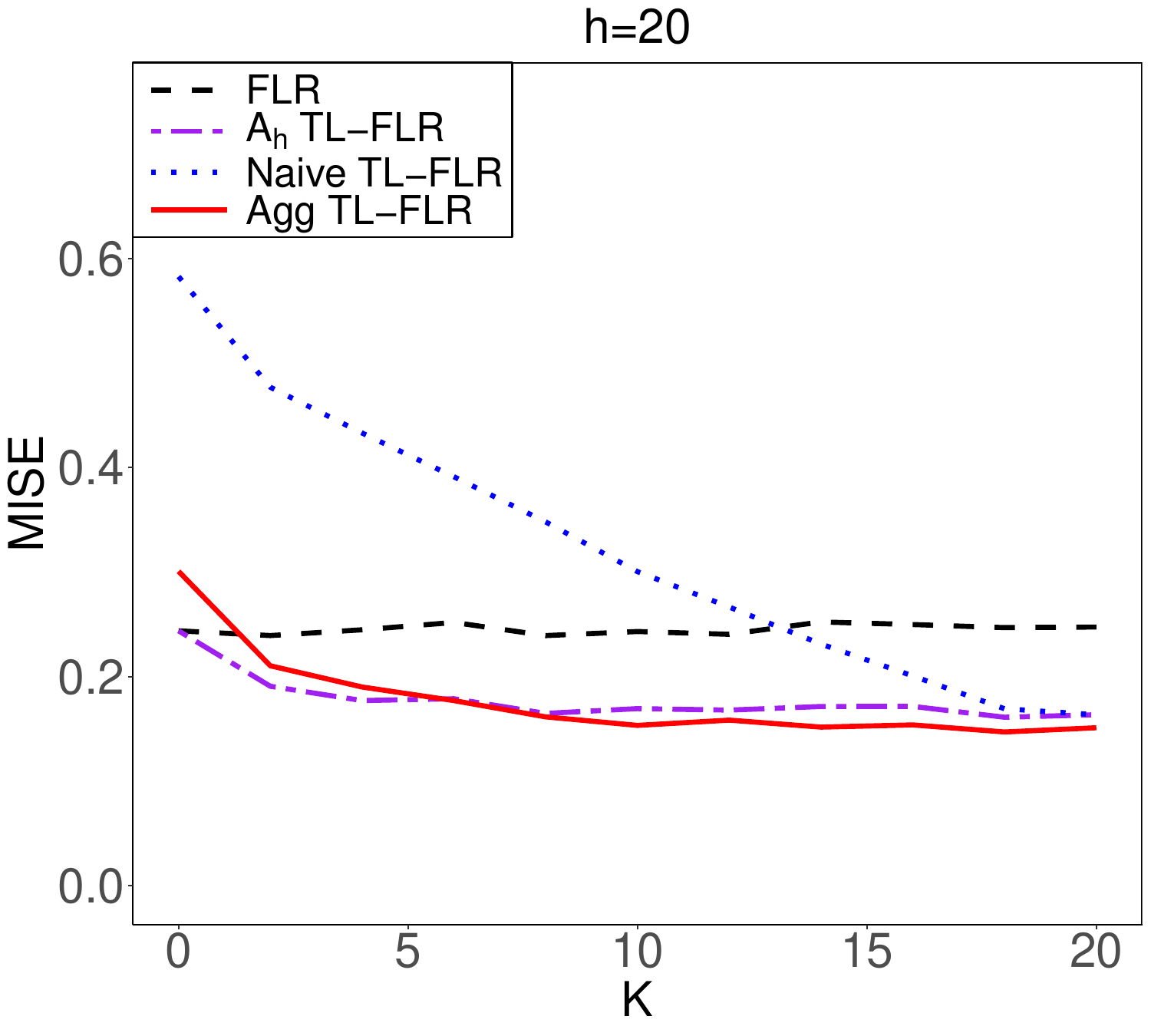} &
				\hspace{\thisgap}\includegraphics[width=\thiswidth]{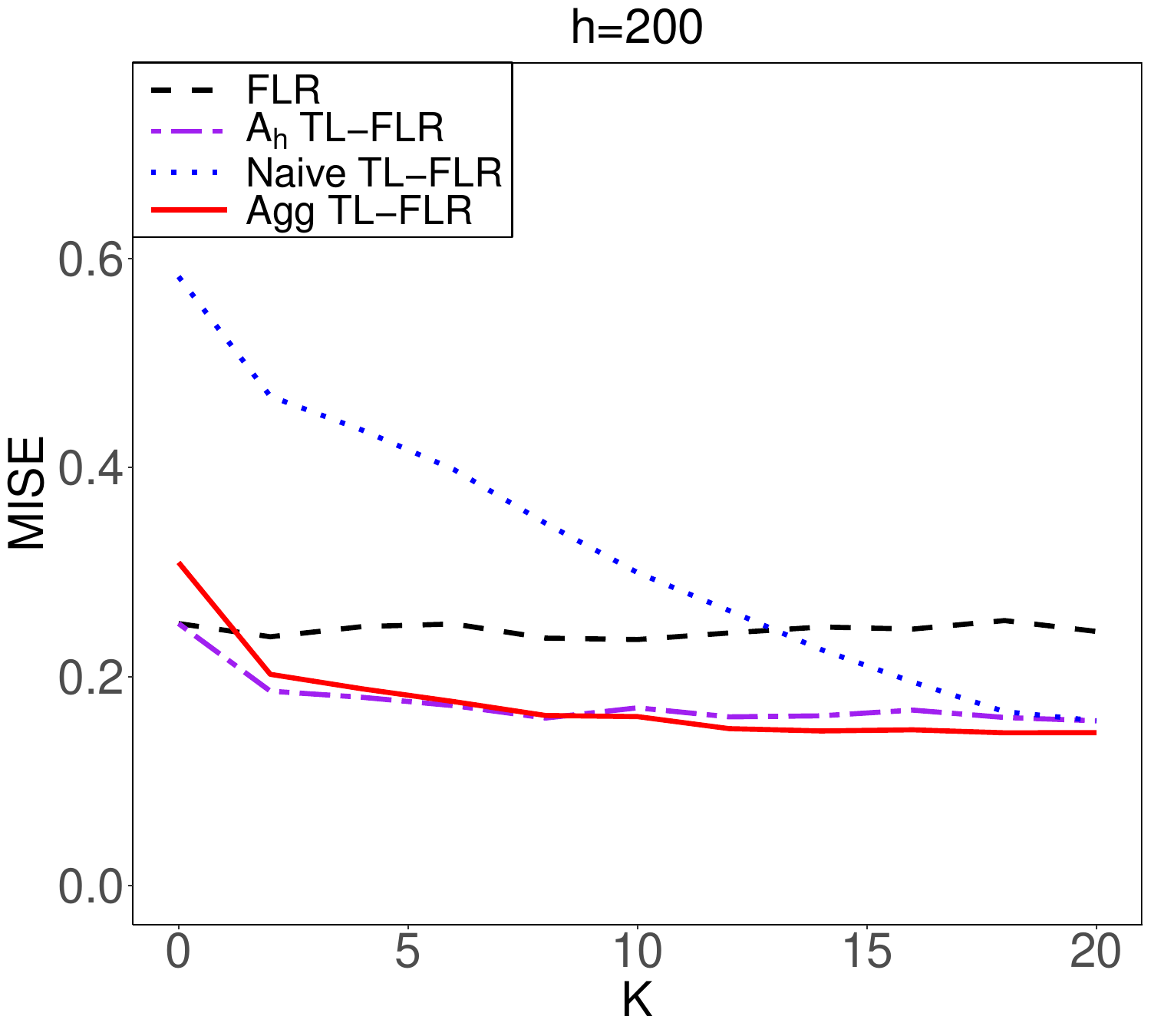} \\
				\hspace{\thisgap}\includegraphics[width=\thiswidth]{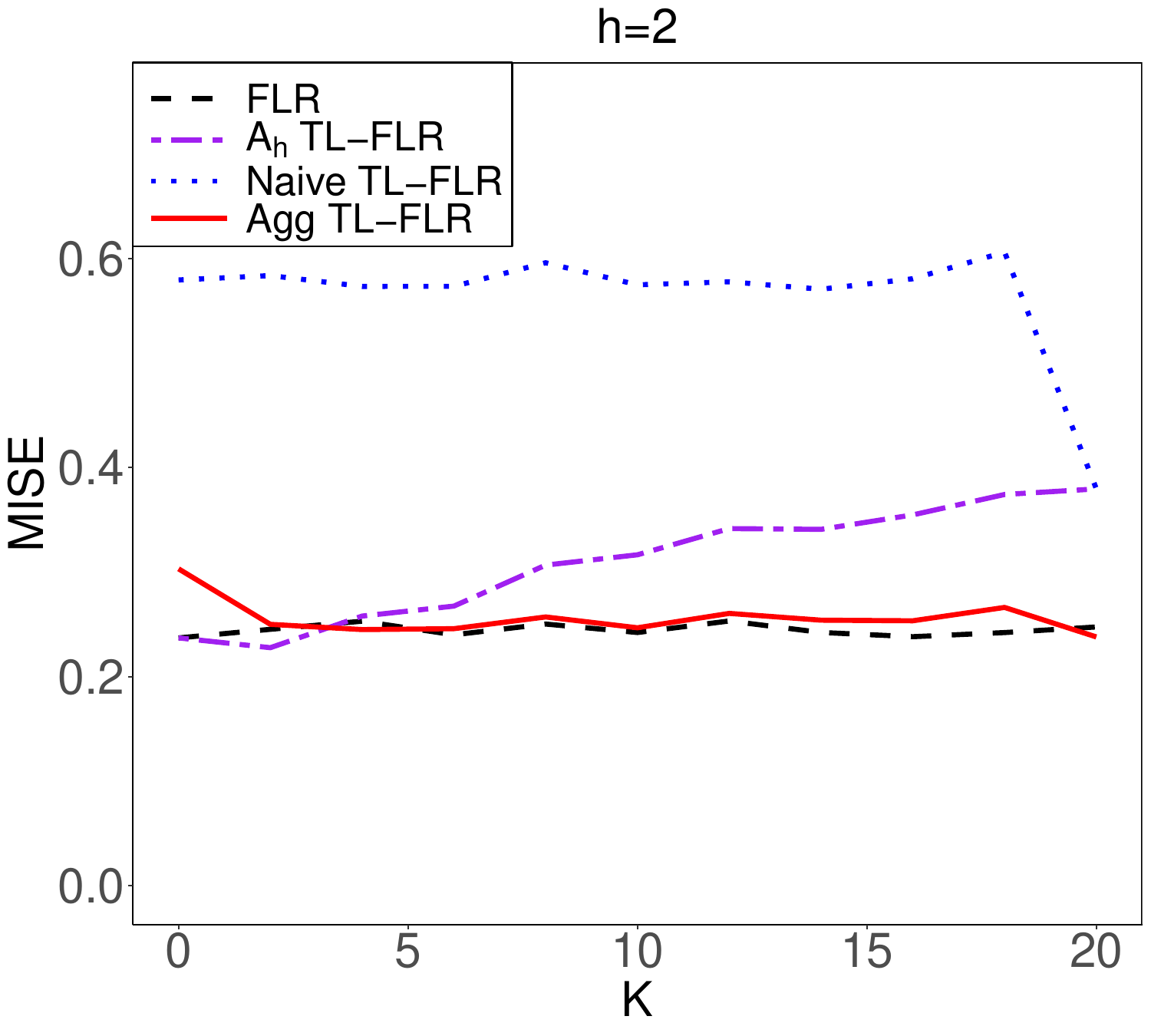} &
				\hspace{\thisgap}\includegraphics[width=\thiswidth]{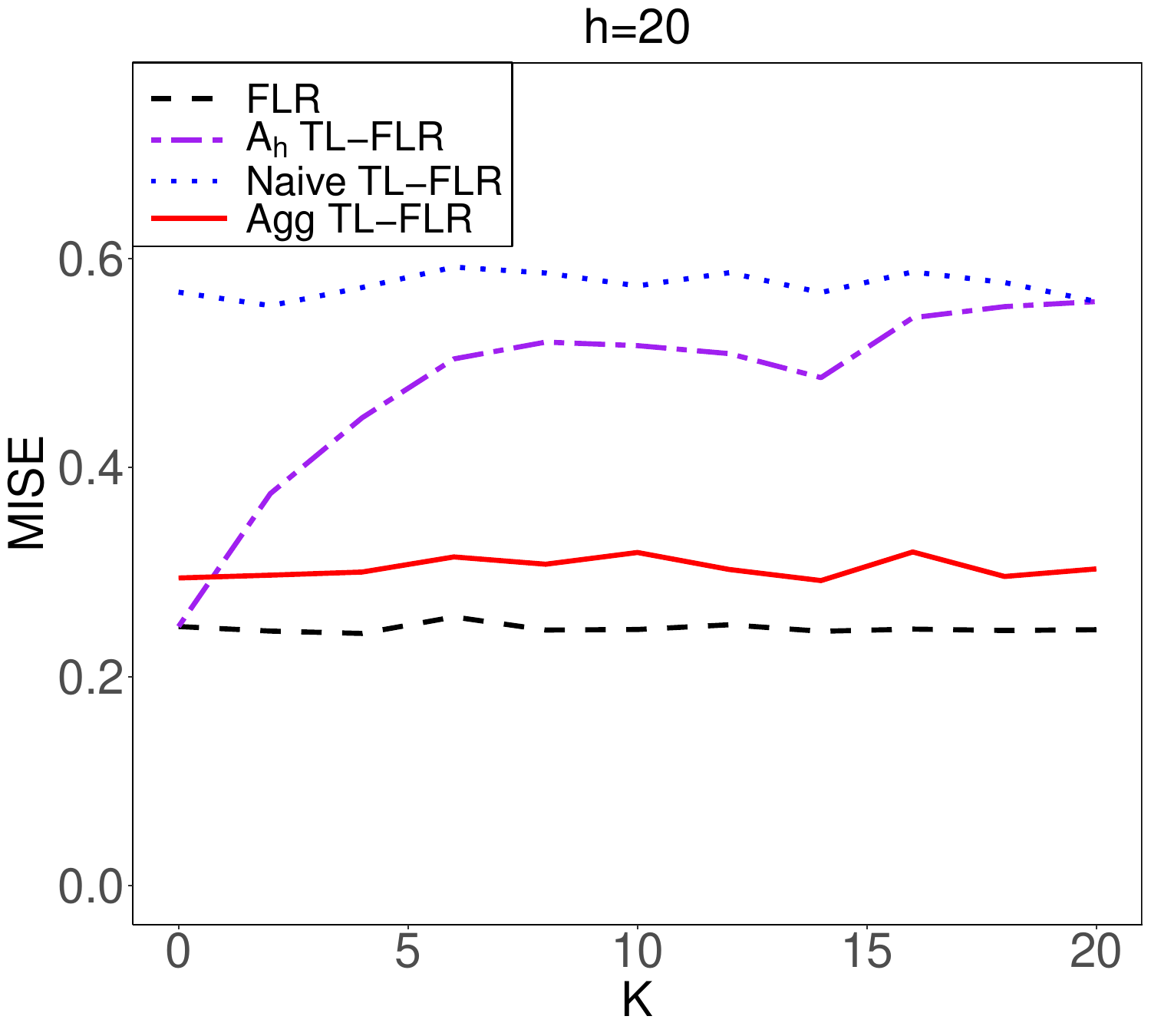} &
				\hspace{\thisgap}\includegraphics[width=\thiswidth]{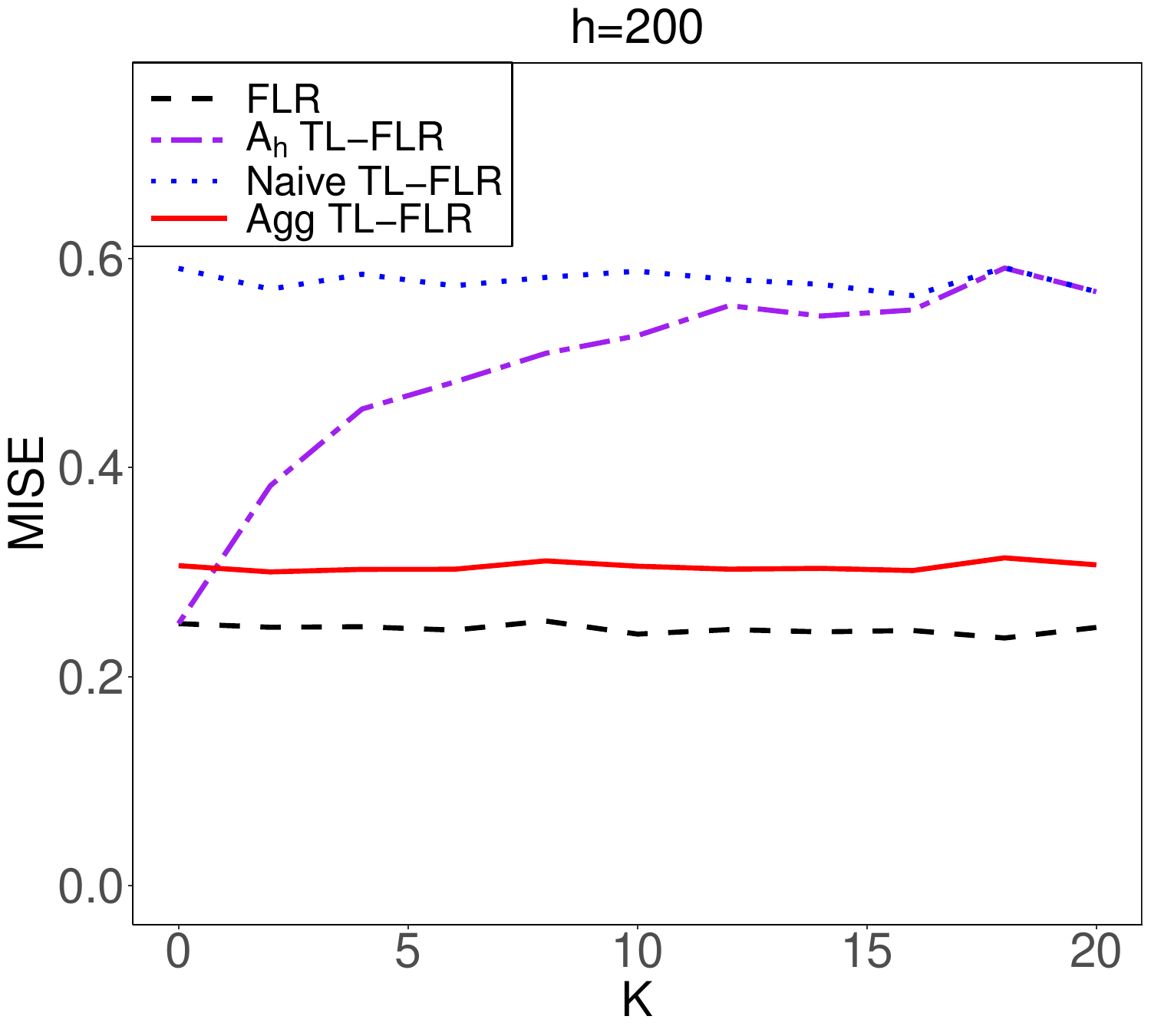} 
			\end{tabular}
			\caption{Estimation errors of different methods over 500 repetitions. Top row: Model (II); middle row: Model (III); bottom row: Model (IV).}
			\label{fig:adap}
		\end{figure}
		
		\section{Real Data} \label{sec:financial}
		
		We demonstrate the performance of the proposed method using stock price data from the period 05/01/2021-09/30/2021 \citep{lin2024hypothesis}. The dataset consists of 60, 58, 31, 30, 104, 55, 70, 68, 46, 103, 41 stocks from 11 sectors, including basic industries (BI), capital goods (CG), consumer durable (CD), consumer non-durable (CND), consumer services (CS), energy (En), finance (Fin), health care (HC), public utility (PU), technology (Tech) and transportation (Trans), respectively. 
		
		We aim to predict the monthly return of the second month using the monthly cumulative return of the first month for a specific sector,  leveraging information from other sectors via transfer learning. For a given stock, the daily prices are $ \{ v(t_0), v(t_1), \dots, v(t_T) \}$ in the first month and are $\{v'(t_0), v'(t_1), \dots, v'(t_T)\}$ in the second month. Therefore, the covariates and responses are defined as 
		\[ X(t) = \frac{v(t) - v(t_0)}{v(t_0)}, ~~ Y = \frac{v'(t_T) - v'(t_0)}{v'(t_0)}. \]
		
		We treat the data from each sector as a target sample, with all other sectors serving as source samples. The target data is randomly split into 80\% training data and 20\% testing data. The prediction errors are evaluated on the testing data over 500 independent repetitions. We compare the performance of ``Naive TL-FLR", ``Agg TL-FLR" and ``FLR" by reporting the average relative prediction errors, i.e., the prediction errors of ``Naive TL-FLR" or ``Agg TL-FLR" divided by the prediction errors of ``FLR". As illustrated in Figure \ref{fig:adap-financial}, the ``Agg TL-FLR" consistently outperforms the ``FLR" in terms of prediction accuracy, with relative prediction errors significantly below 1 in almost all cases. In contrast, the ``Naive TL-FLR" often results in much poorer performance compared to the ``FLR". These observations support our claim that the adaptive algorithm ``Agg TL-FLR" effectively alleviates negative transfer.
		
		\begin{figure}
			\centering
			\newcommand{\thiswidth}{0.4\linewidth}
			\newcommand{\thisgap}{0mm}
			\begin{tabular}{cc}
				\hspace{\thisgap}\includegraphics[width=\thiswidth]{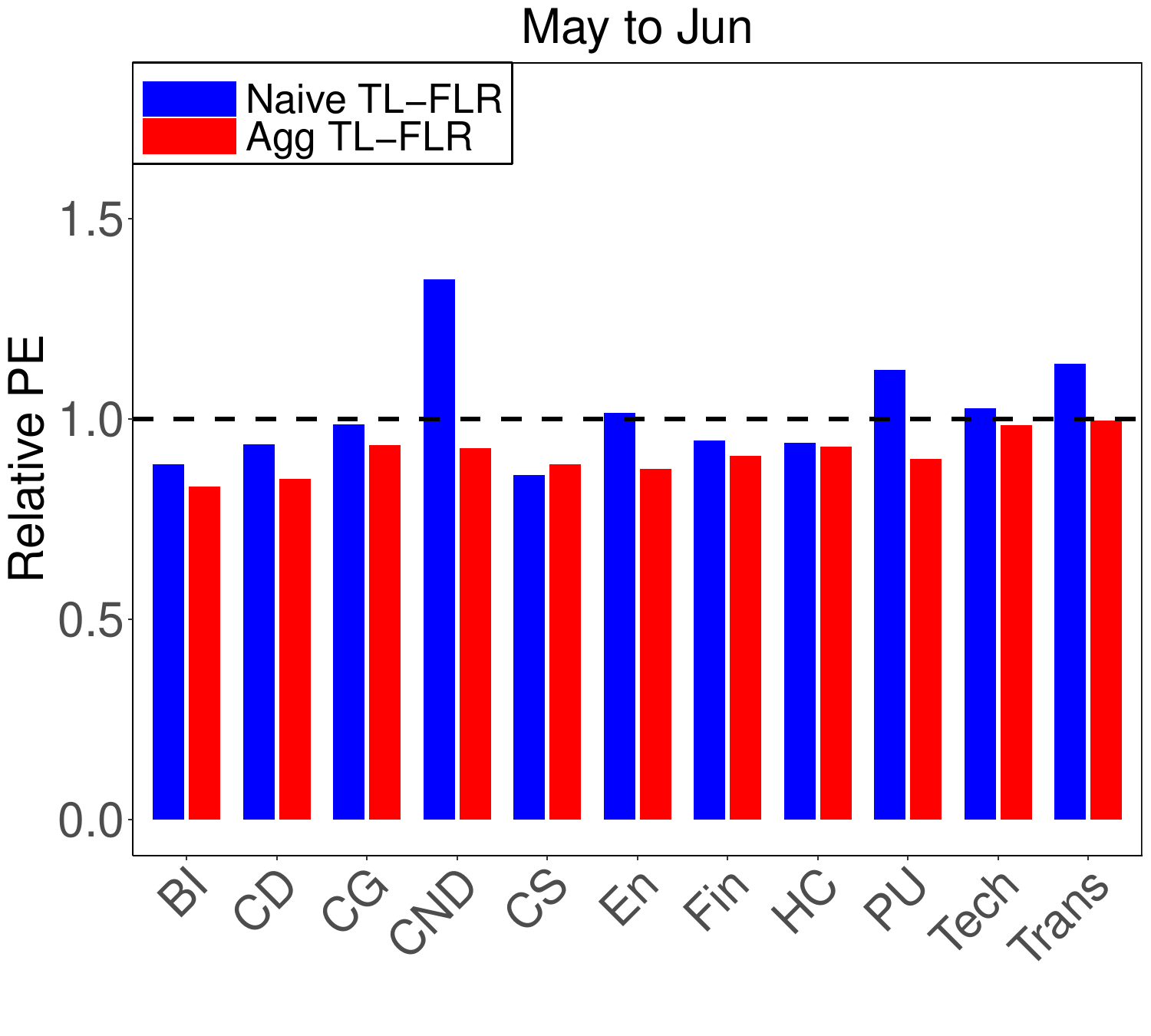} &
				\hspace{\thisgap}\includegraphics[width=\thiswidth]{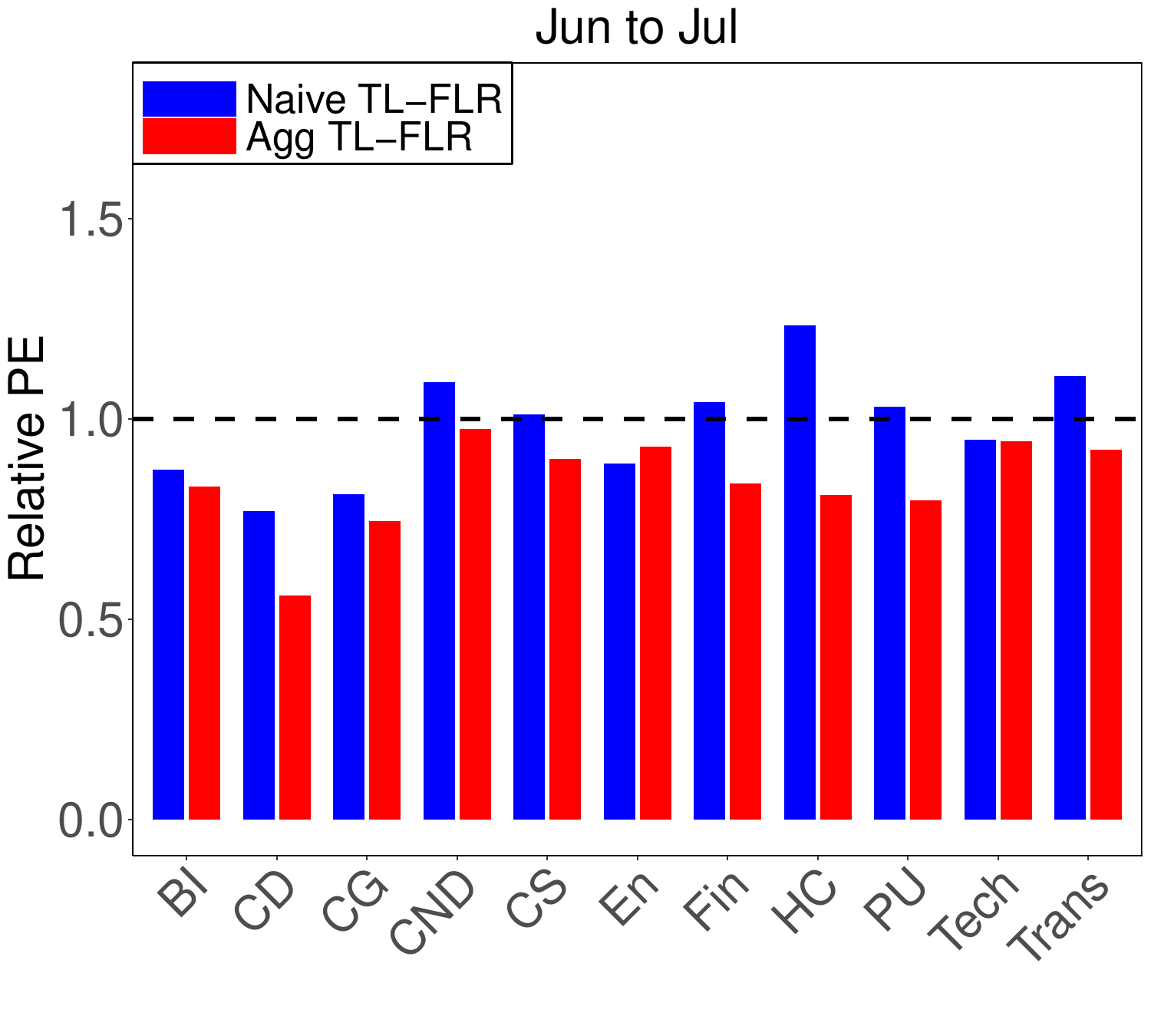} \\
				\hspace{\thisgap}\includegraphics[width=\thiswidth]{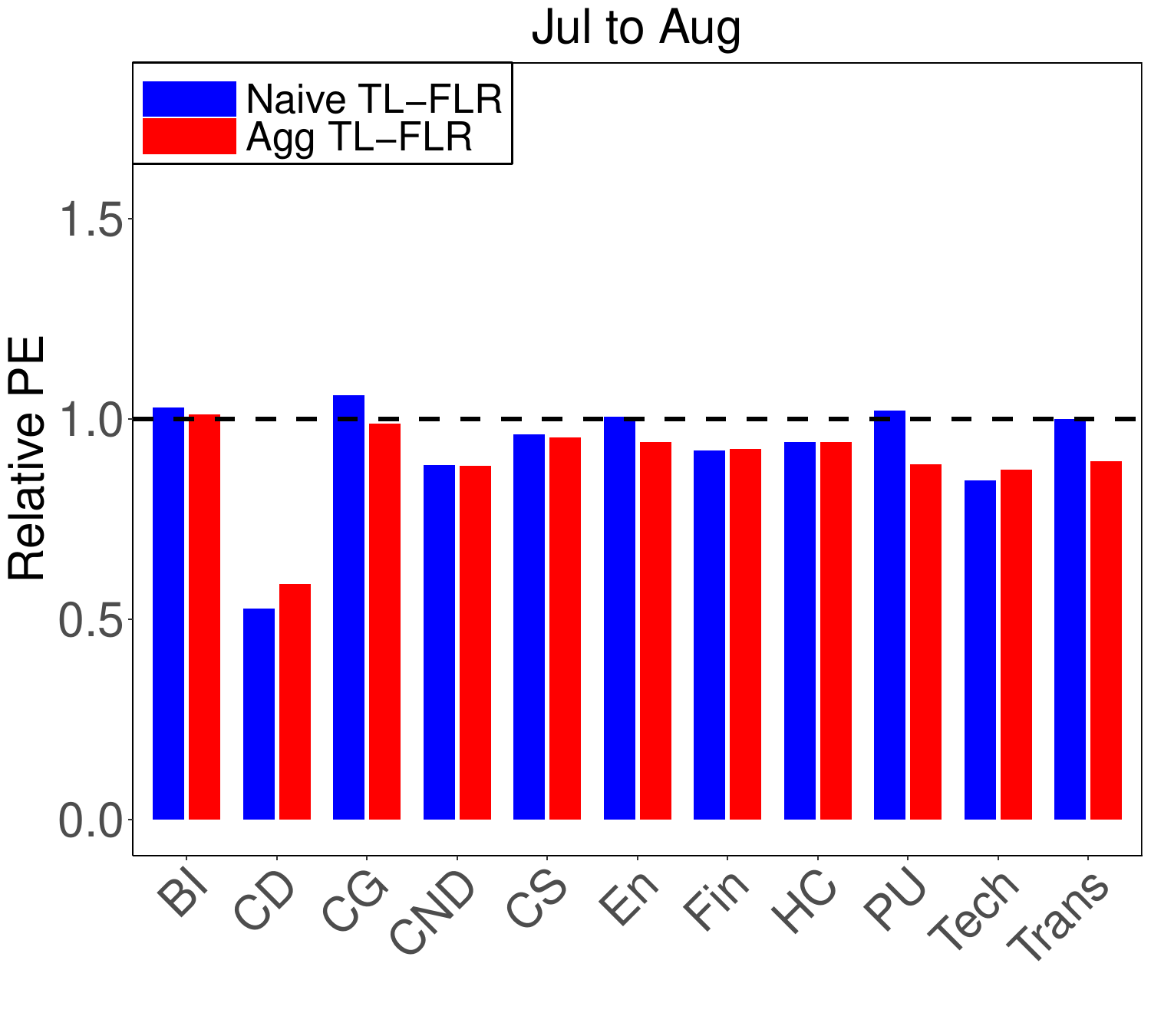} &
				\hspace{\thisgap}\includegraphics[width=\thiswidth]{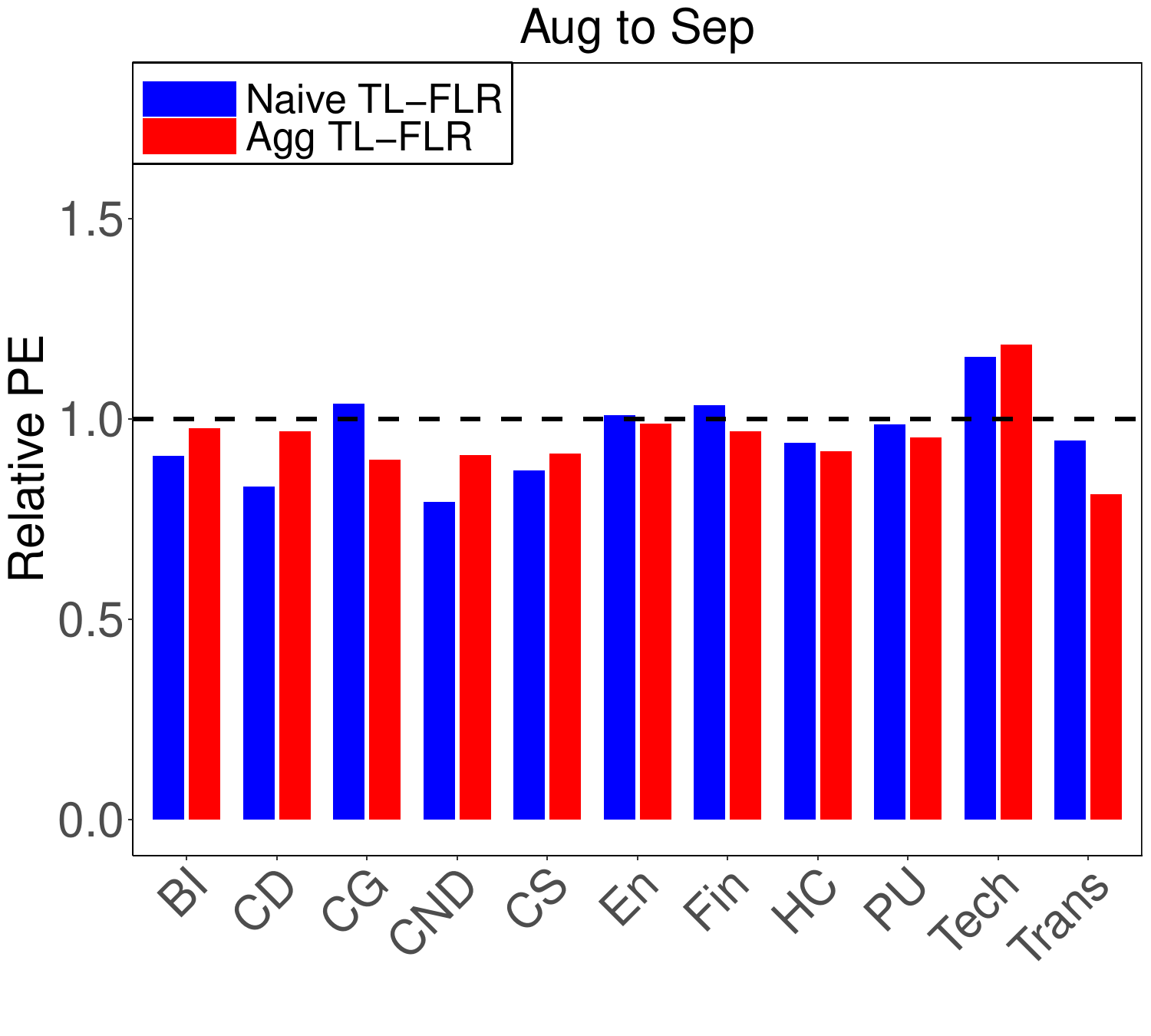} \\
			\end{tabular}
			\caption{Average relative prediction errors of different methods for each target sector over 500 repetitions.}
			\label{fig:adap-financial}
		\end{figure}
		
		\section{Conclusion}\label{sec:conclusion}
		
		In this work, we focus on slope estimation in functional linear regression under transfer learning. 
		When the covariate distributions share an aligned eigenspace, we propose an algorithm that avoids negative transfer even when the contrast is large. Our algorithm enhances the learning of the target model, provided that the sample size of the source data is sufficiently large and the contrast between the source and target models is sufficiently small.
		To mitigate performance degradation when the aligned eigenspace condition is violated, we introduce an adaptive algorithm via constructing potential candidate estimators and performing sparse aggregation. The adaptive algorithm performs well in both synthetic and real data examples.
		
		%The restricted eigenvalue property is crucial for quantifying error bounds in high-dimensional regression. We establish similar results for functional data, which may be of independent interest.

		\section*{Acknowledgments}
		Dr. Lin's research was partially supported by MOE AcRF Tier 1 grant (A-8002518-00-00).

		 \bibliography{TL-FLR}
	
	\onecolumn
	
	\renewcommand{\theequation}{S.\arabic{equation}}
	\renewcommand{\thelemma}{S.\arabic{lemma}}
	\renewcommand{\thetheorem}{S.\arabic{theorem}}
	\renewcommand{\theremark}{S.\arabic{remark}}
	\renewcommand{\thefigure}{S.\arabic{figure}}
	\renewcommand{\thecorollary}{S.\arabic{corollary}}
	
	\setcounter{secnumdepth}{0} %May be changed to 1 or 2 if section numbers are desired.
	\setcounter{figure}{0}
	\setcounter{equation}{0}
	\setcounter{theorem}{0}
	\setcounter{lemma}{0}
	\setcounter{corollary}{0}
	
	\begin{center}
		\LARGE	\textbf{Supplemtary Material for ``Transfer Learning Meets Functional Linear Regression: No Negative Transfer under Posterior Drift"}
	\end{center}
	
		\section{More Numerical Results}
		
		\subsection{Transfer Learning on $\mathcal A_h$ with Aligned Eigenspace}
		
		\begin{figure}[!h]
			\centering
			\newcommand{\thiswidth}{0.4\linewidth}
			\newcommand{\thisgap}{0.2mm}
			\begin{tabular}{cc}
				\hspace{\thisgap}\includegraphics[width=\thiswidth]{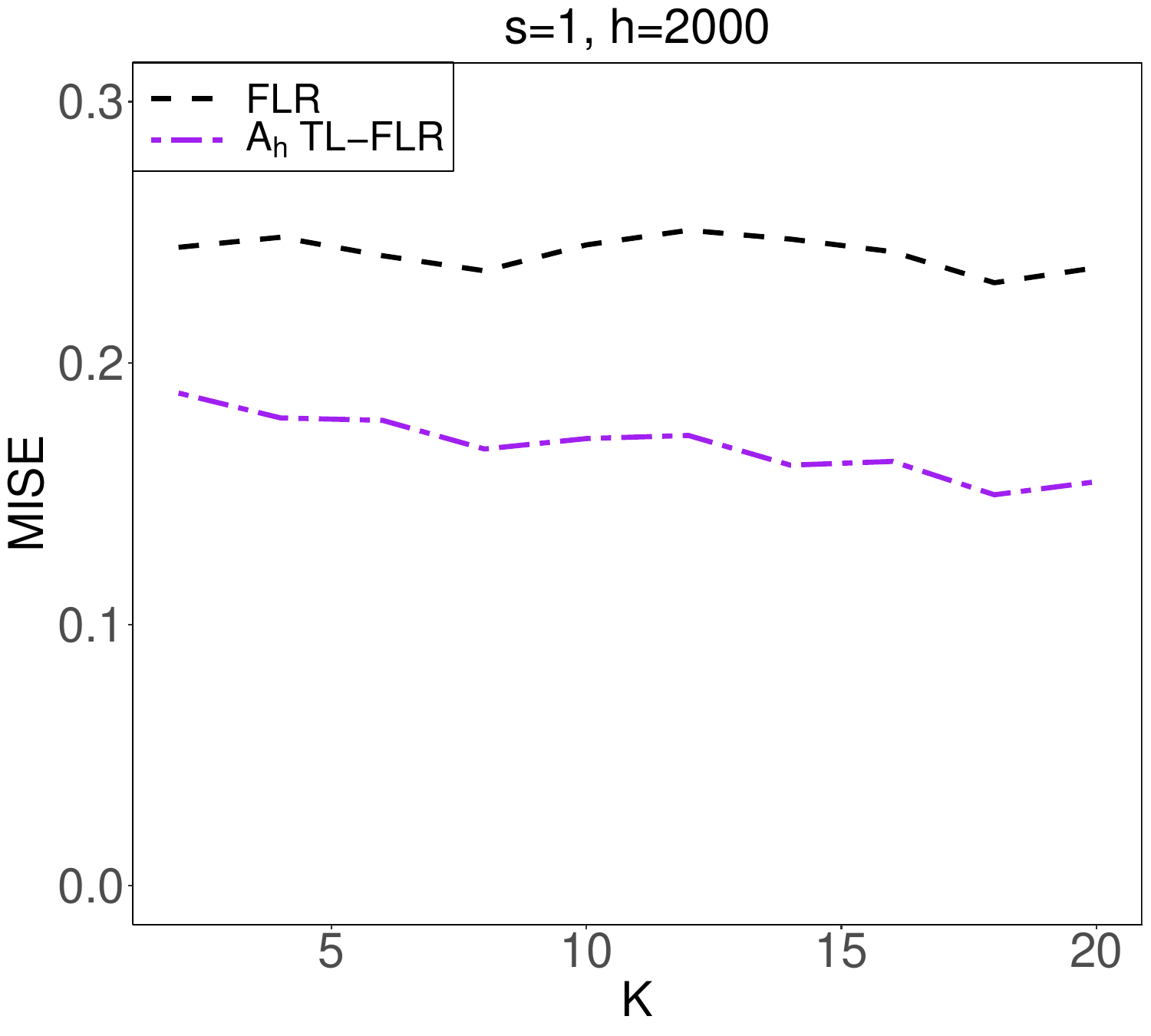} & 
				\hspace{\thisgap}\includegraphics[width=\thiswidth]{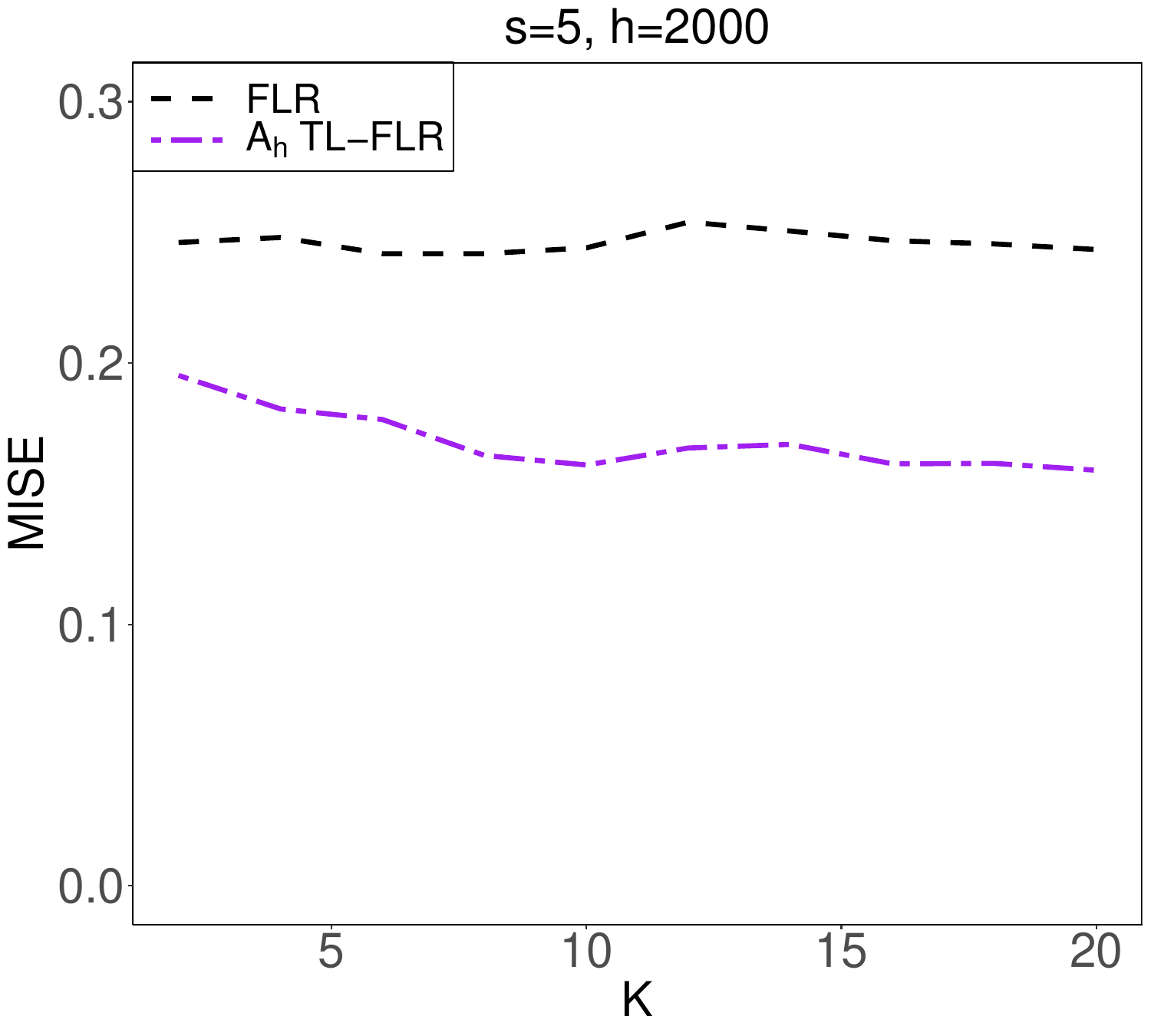} \\
				\hspace{\thisgap}\includegraphics[width=\thiswidth]{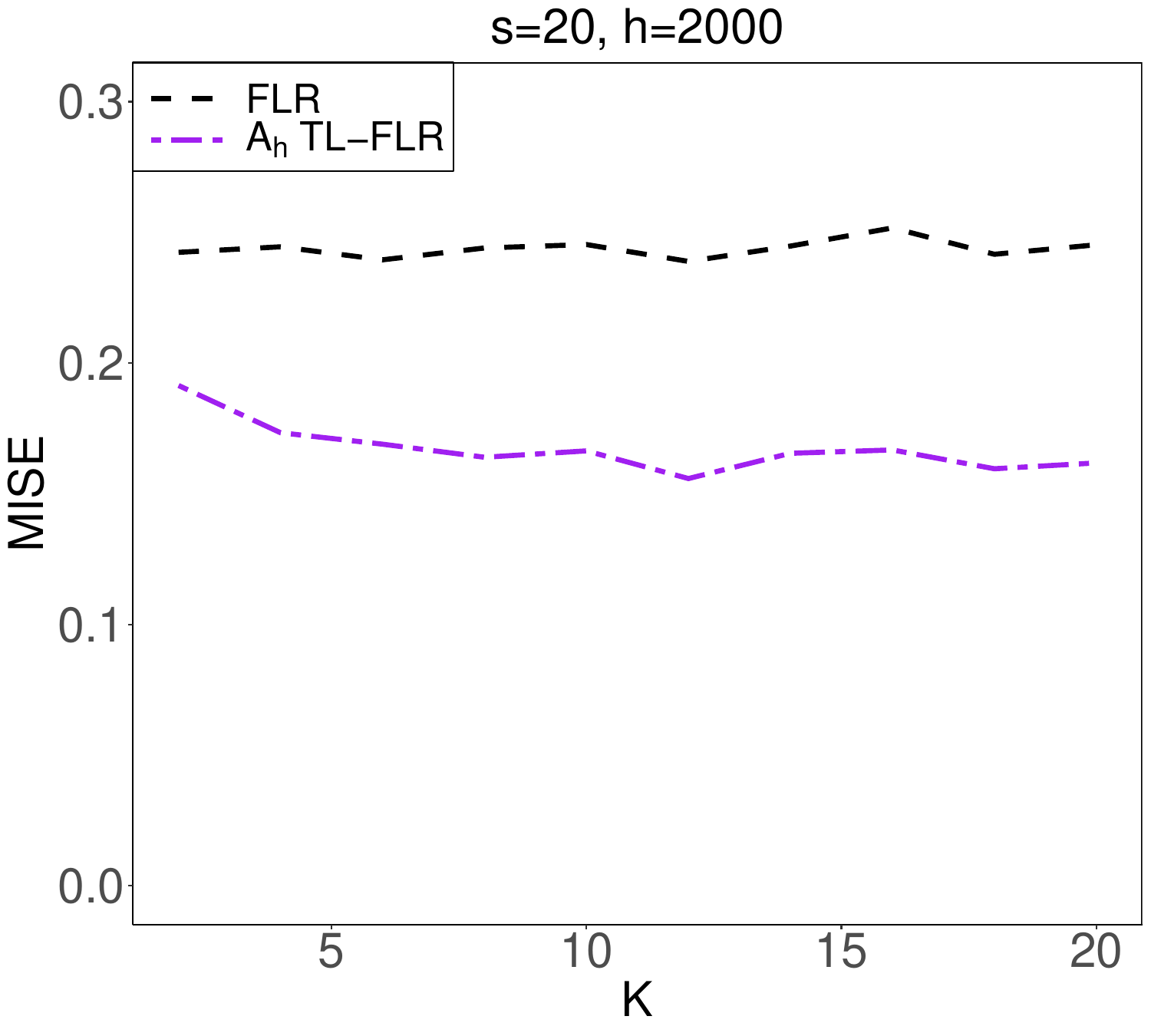} & 
				\hspace{\thisgap}\includegraphics[width=\thiswidth]{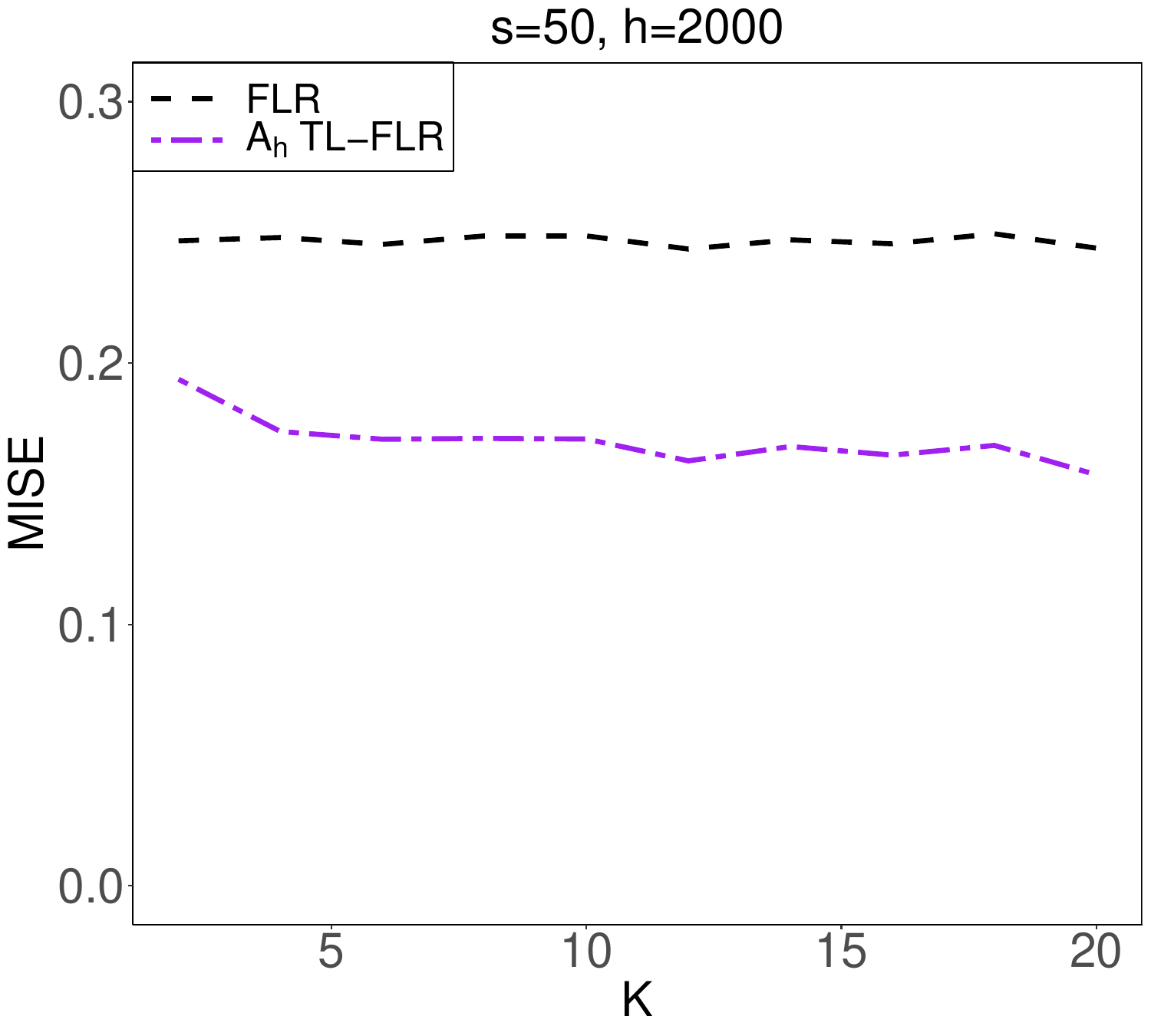}
			\end{tabular}
			\caption{Estimation errors of different methods under Model (I) over 1000 repetitions.}
			\label{fig:h2000}
		\end{figure}
		
		We present the results for $h=2000$ in Figure \ref{fig:h2000}. Despite the extremely large contrast between the source and target models, the algorithm $\mathcal A_h$ TL-FLR still performs better than FLR, demonstrating its ability to avoid negative transfer under the scenario of aligned eigenspace.
		
		\subsection{Comparisons between Q-Aggregation and Sparse Aggregation}
		
		\begin{figure}[!h]
			\centering
			\newcommand{\thiswidth}{0.3\linewidth}
			\newcommand{\thisgap}{0mm}
			\begin{tabular}{ccc}
				\hspace{\thisgap}\includegraphics[width=\thiswidth]{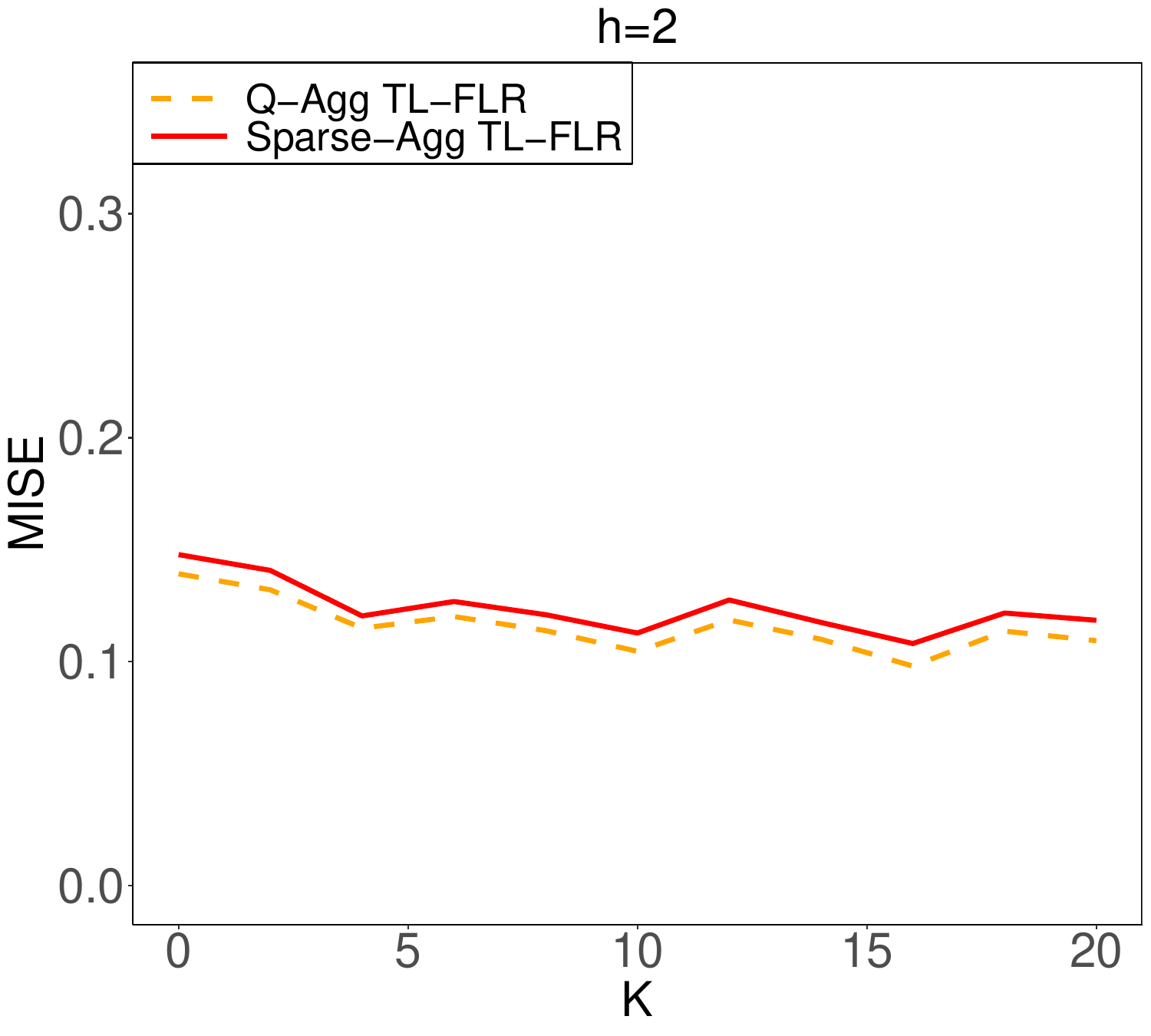} & 
				\hspace{\thisgap}\includegraphics[width=\thiswidth]{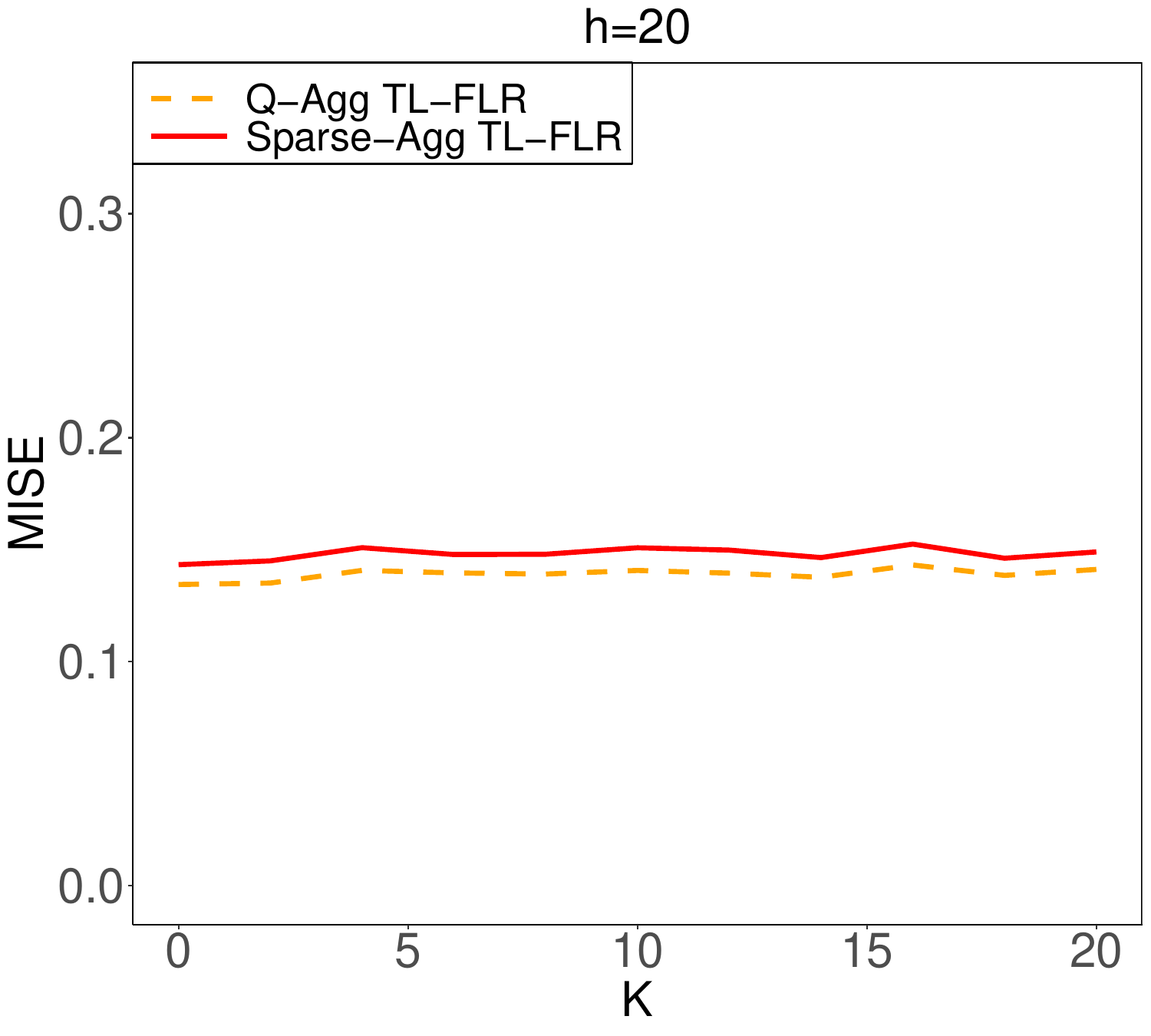} &
				\hspace{\thisgap}\includegraphics[width=\thiswidth]{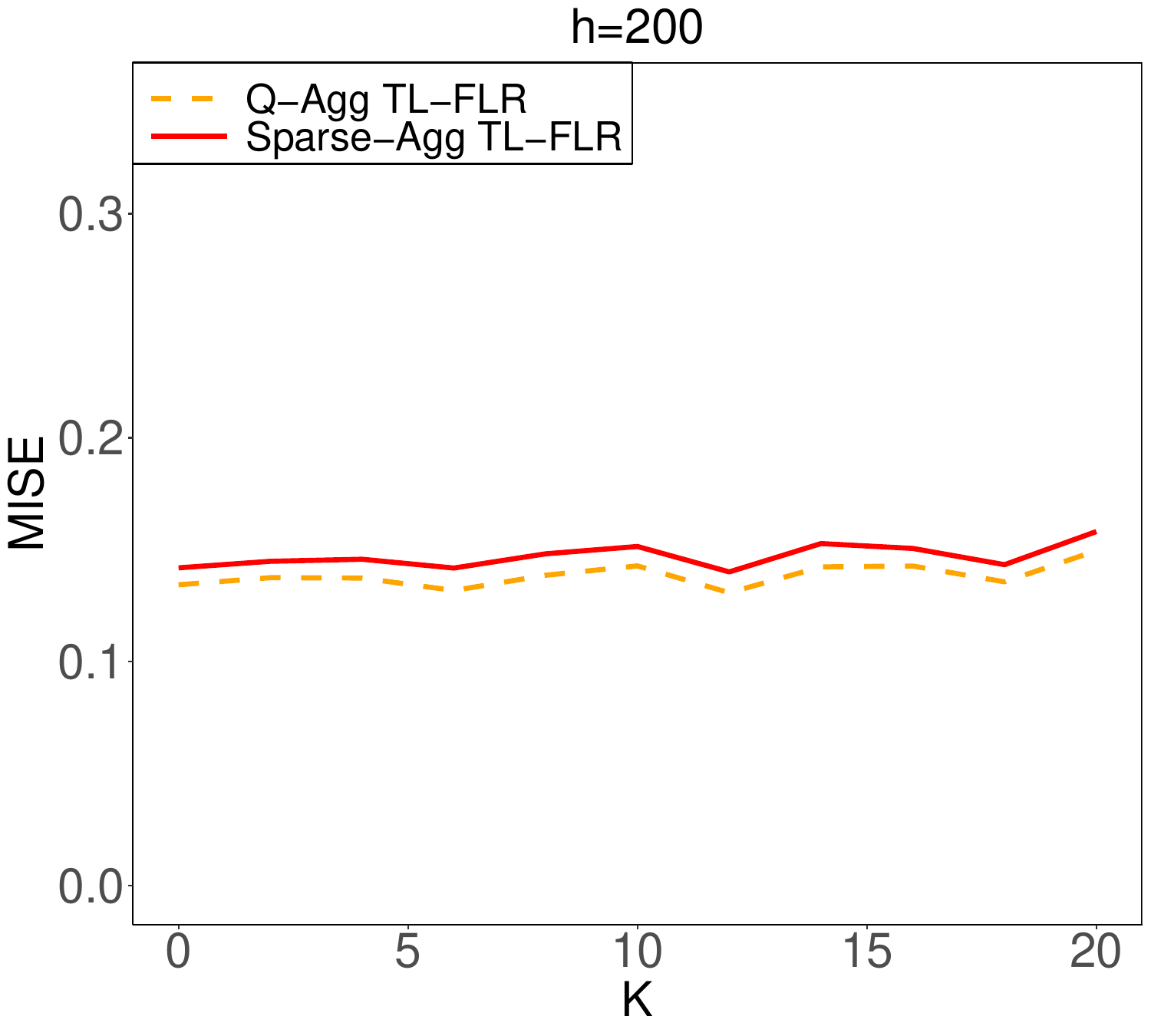} \\
				\hspace{\thisgap}\includegraphics[width=\thiswidth]{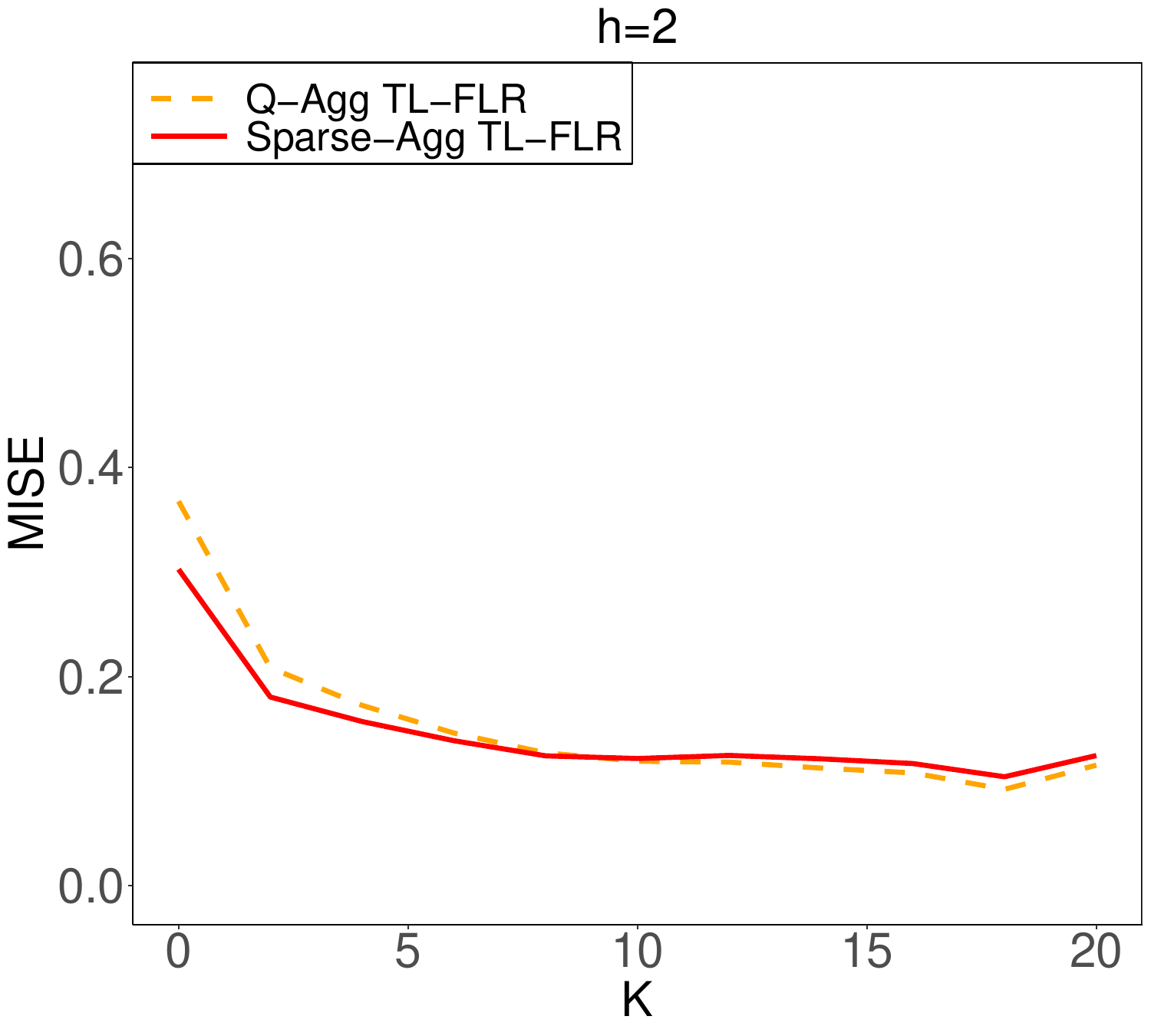} & 
				\hspace{\thisgap}\includegraphics[width=\thiswidth]{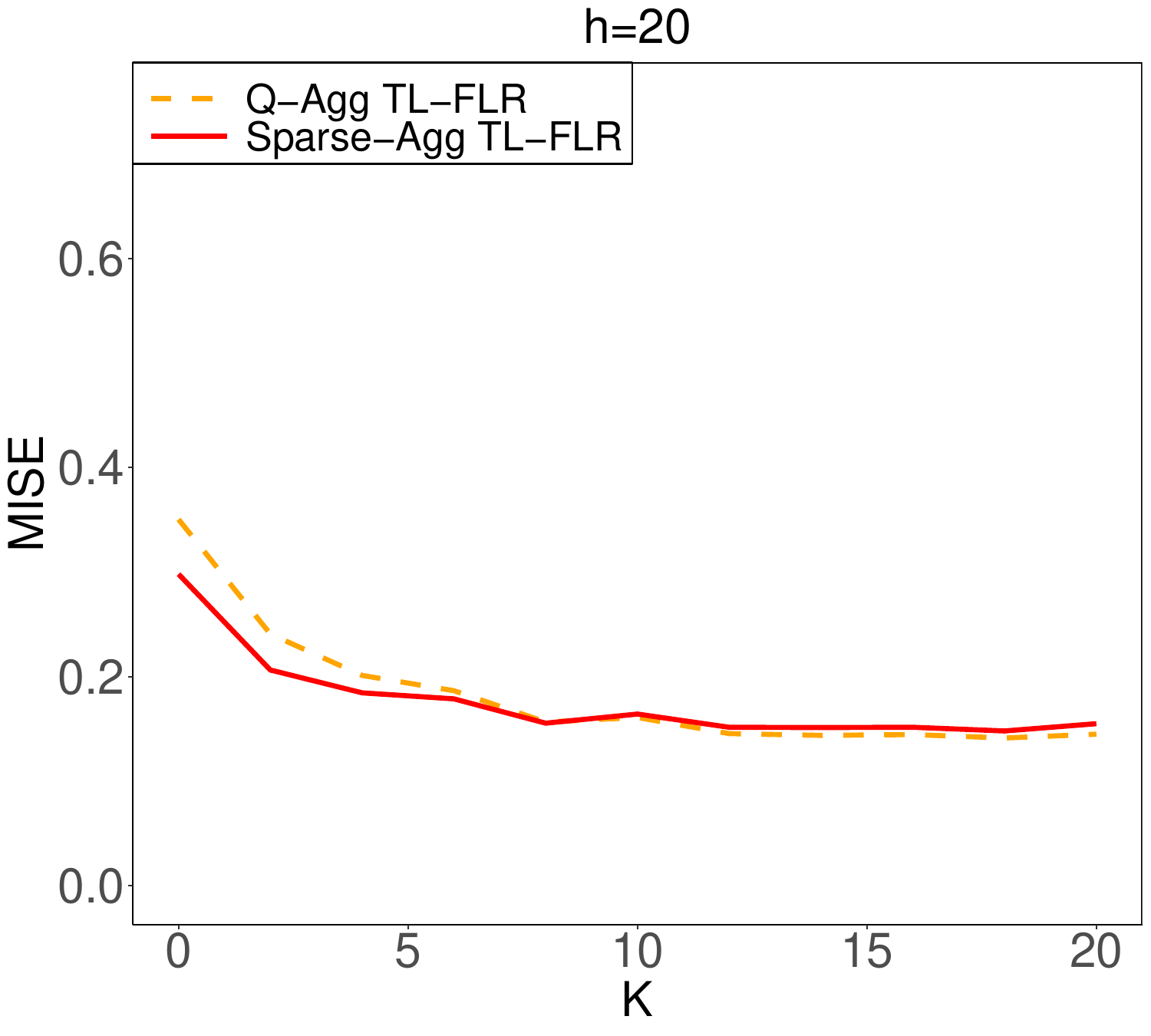} &
				\hspace{\thisgap}\includegraphics[width=\thiswidth]{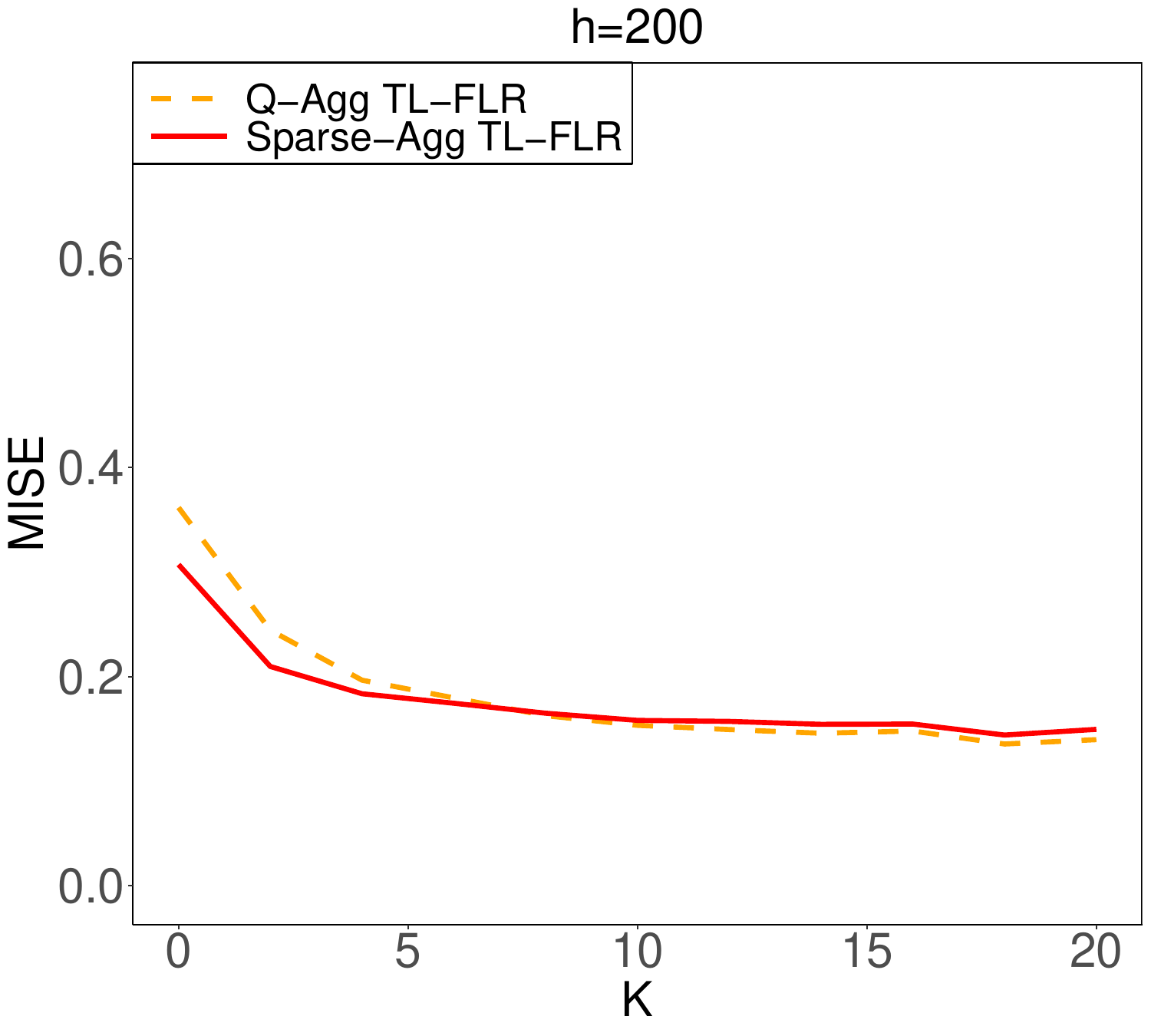} \\
				\hspace{\thisgap}\includegraphics[width=\thiswidth]{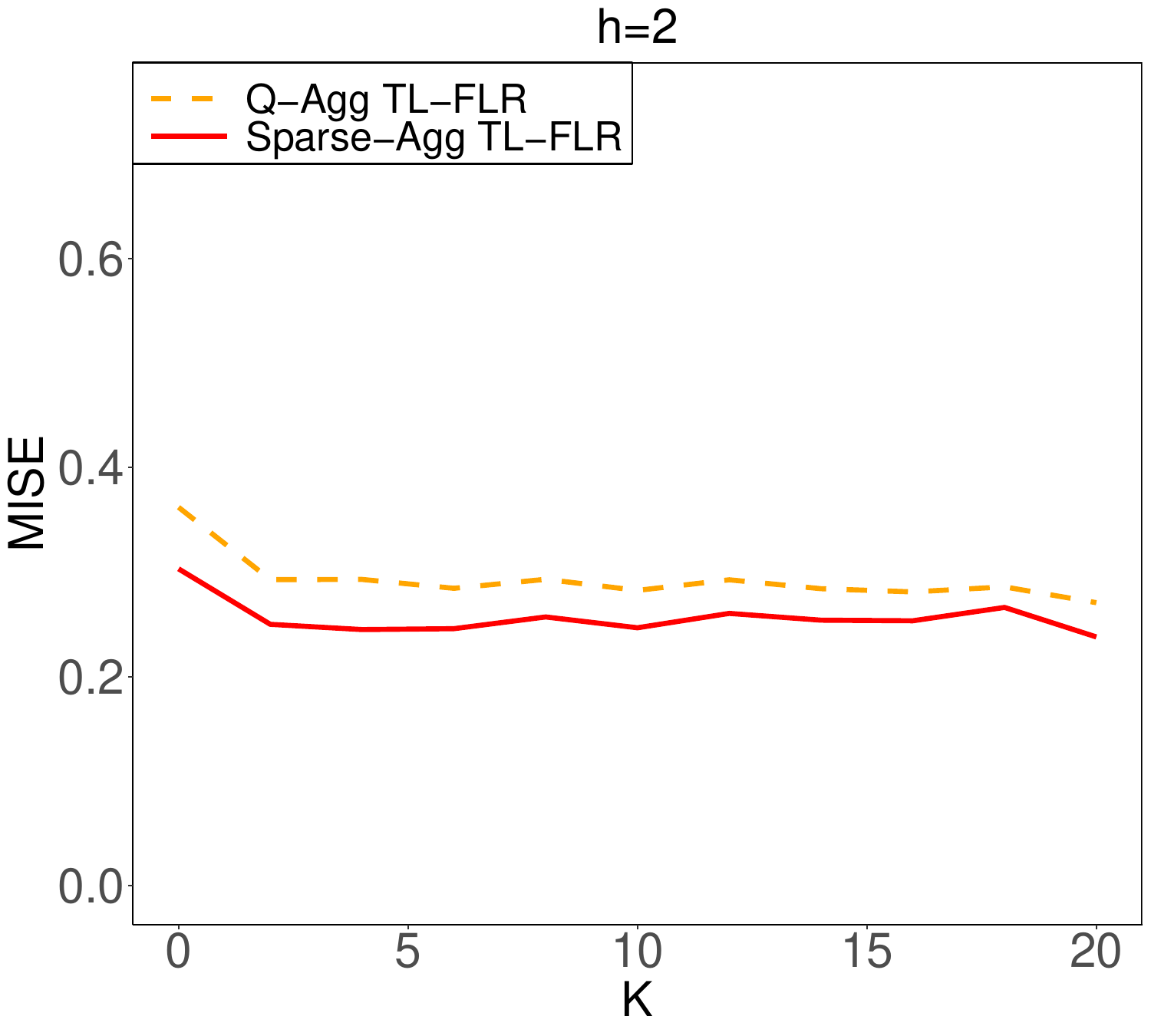} & 
				\hspace{\thisgap}\includegraphics[width=\thiswidth]{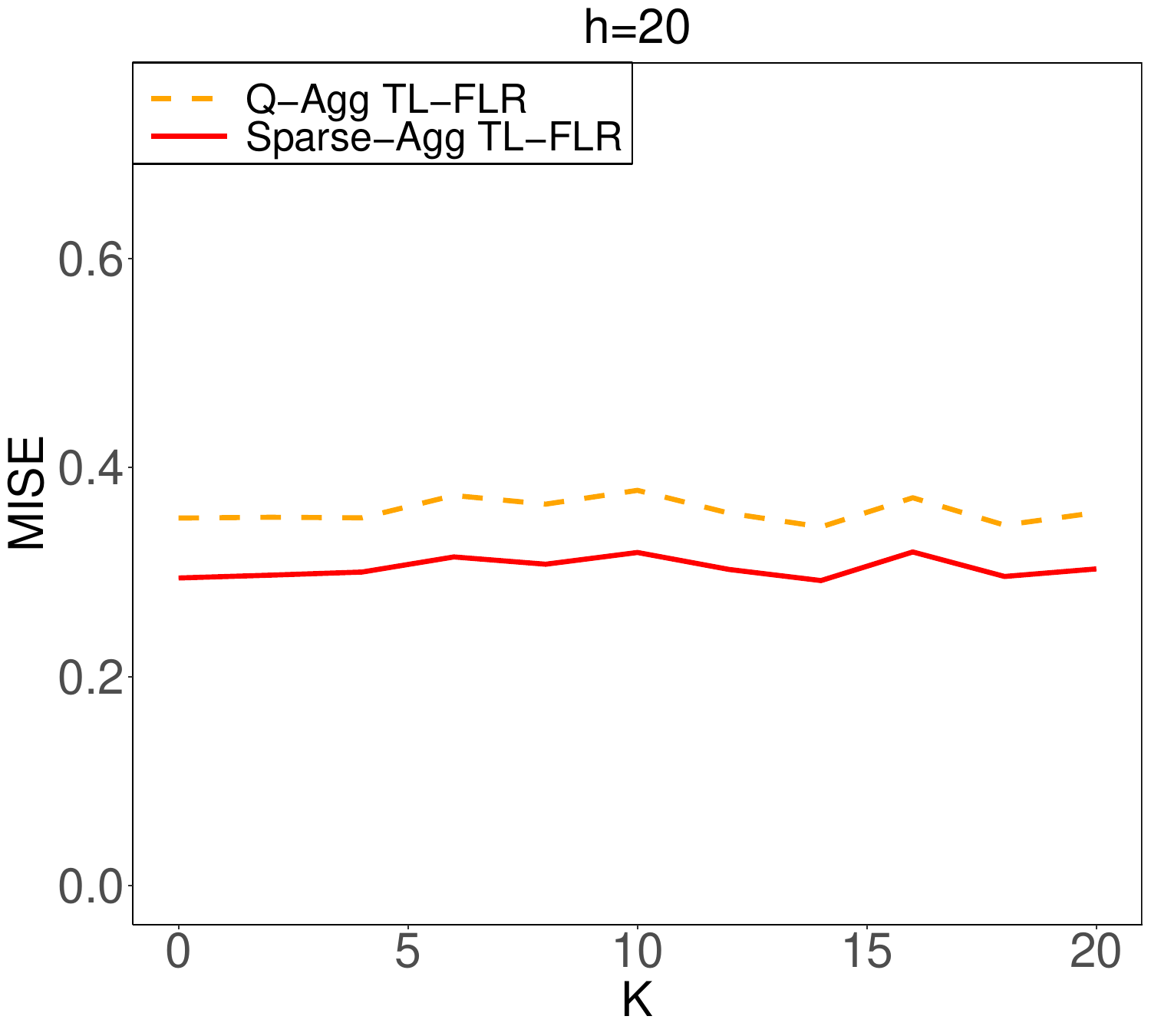} &
				\hspace{\thisgap}\includegraphics[width=\thiswidth]{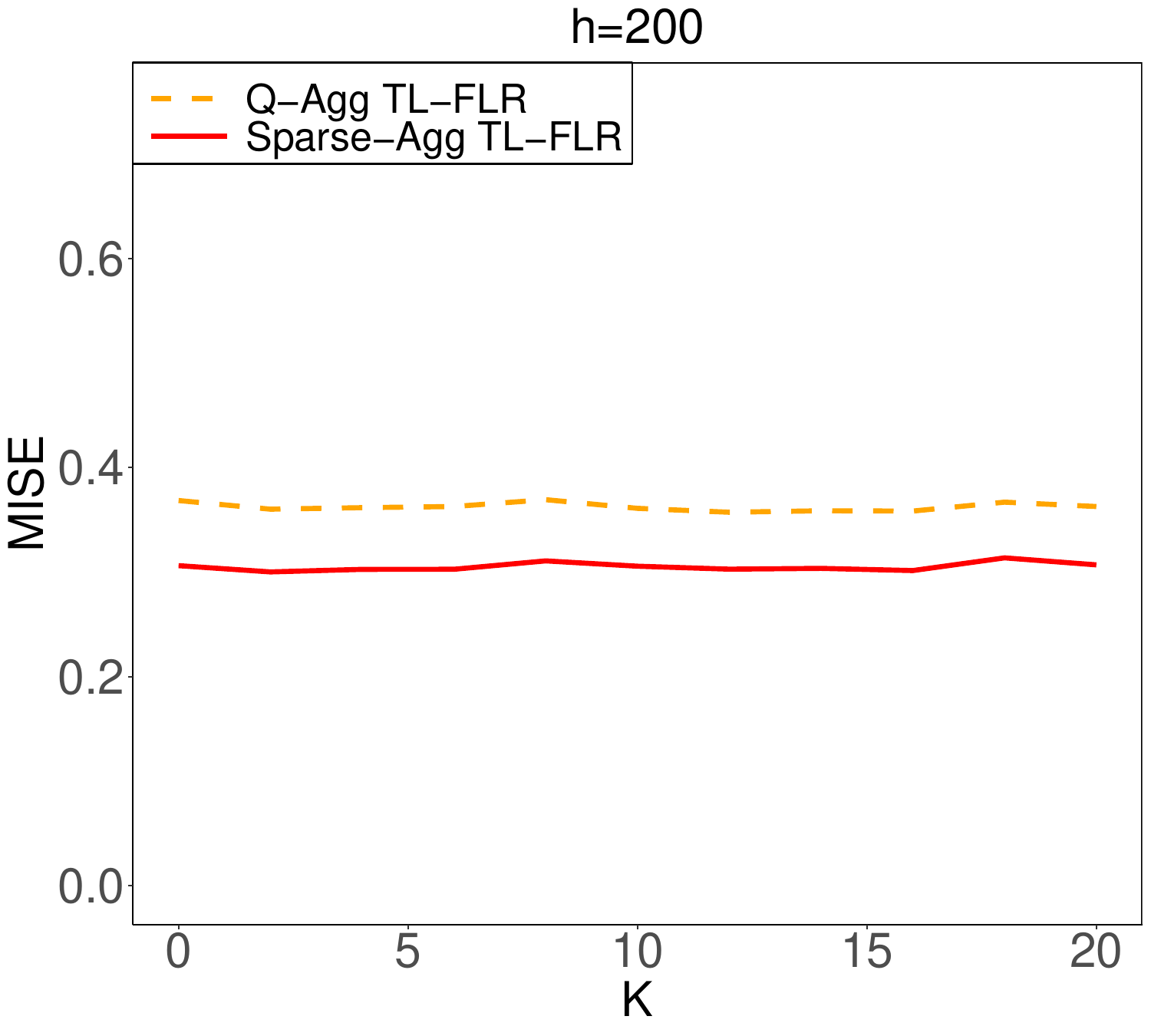} 
			\end{tabular}
			\caption{Estimation errors of the adaptive method with Q-aggregation or sparse aggregation over 500 repetitions. Top row: Model (II); middle row: Model (III); bottom row: Model (IV).}
			\label{fig:qagg-sagg}
		\end{figure}
		
		We demonstrate how to obtain the aggregate estimator by replacing sparse aggregation with Q-aggregation.
		Compute 
		\begin{align*}
			\hat \brho =  & \mathop{\arg\min}_{\brho \in \Lambda^{L+1}} \frac{1}{\bar{n}} \sum_{i \in \mathcal I_2} \left\{Y_i - \bar{Y}_2- \int_{\tdomain} \big(X_i(t) - \bar{X}_2(t)\big) \bigg(\sum_{l=0}^{L} \rho_l \hat b_l(t) \bigg) dt  \right\}^2  + \\ 
			& \quad\quad\quad\quad \frac{1}{\bar{n}} \sum_{l=0}^L \rho_l \sum_{i \in \mathcal I_2} \bigg\{Y_i - \bar{Y}_2 - \int_{\tdomain} \big(X_i(t) - \bar{X}_2(t) \big) \hat b_l(t)dt  \bigg\}^2 + \frac{2 \tau_{\brho} }{\bar{n}} \sum_{l=0}^{L} \rho_l \log(\rho_l),  
		\end{align*}
		where $\Lambda^{L+1} = \{\brho \in \real^{L+1}: \rho_l \ge 0, \sum_{l=0}^L \rho_l =1 \}$ and $\tau_{\brho} >0$.
		Denote the corresponding aggregate estimator by $\hat b_{qagg}(t) = \sum_{l=0}^L \hat \rho_l \hat b_l(t)$.
		For Q-aggregation, we tune the parameter $\tau_{\brho}$ through 5-fold cross-validation. 
		
		As depicted in Figure \ref{fig:qagg-sagg}, both aggregation methods yield comparable performance under Model (II). By contrast, the adaptive algorithm with sparse aggregation outperforms the one with Q-aggregation under Models (III) and (IV). In light of its computational advantages, we recommend using sparse aggregation.
		
		\subsection{Results under Violation of Assumptions}
		
		In this section, we present the results when the sub-Gaussianity in Assumption \ref{assump:xdist} is violated. Specifically,  We consider $X_i(t) = \sum_{k=1}^{50} \sqrt{\lambda_k} Z_{ik} \phi_k(t)$, $i=1, \dots, n$, where $Z_{ik} \stackrel{i.i.d.}{\sim} \sqrt{3/5}t_5$ and all other parameters are consistent with those in Section \ref{sec:sim}, with $t_5$ representing the $t$-distribution with 5 degrees of freedom. To avoid confusion, we refer to the corresponding models under the new data-generating mechanism of the target data as Models (I+), (II+), (III+) and (IV+), respectively.
		
		As shown in Figures \ref{fig:model-I+} and \ref{fig:adap+}, the results exhibit similar patterns to those observed in Section \ref{sec:sim}, despite the violation of the sub-Gaussian assumption for the target data. This demonstrates the robustness of the proposed method.
		
		\begin{figure}[!h]
			\centering
			\newcommand{\thiswidth}{0.23\linewidth}
			\newcommand{\thisgap}{0mm}
			\begin{tabular}{cccc}
				\hspace{\thisgap}\includegraphics[width=\thiswidth]{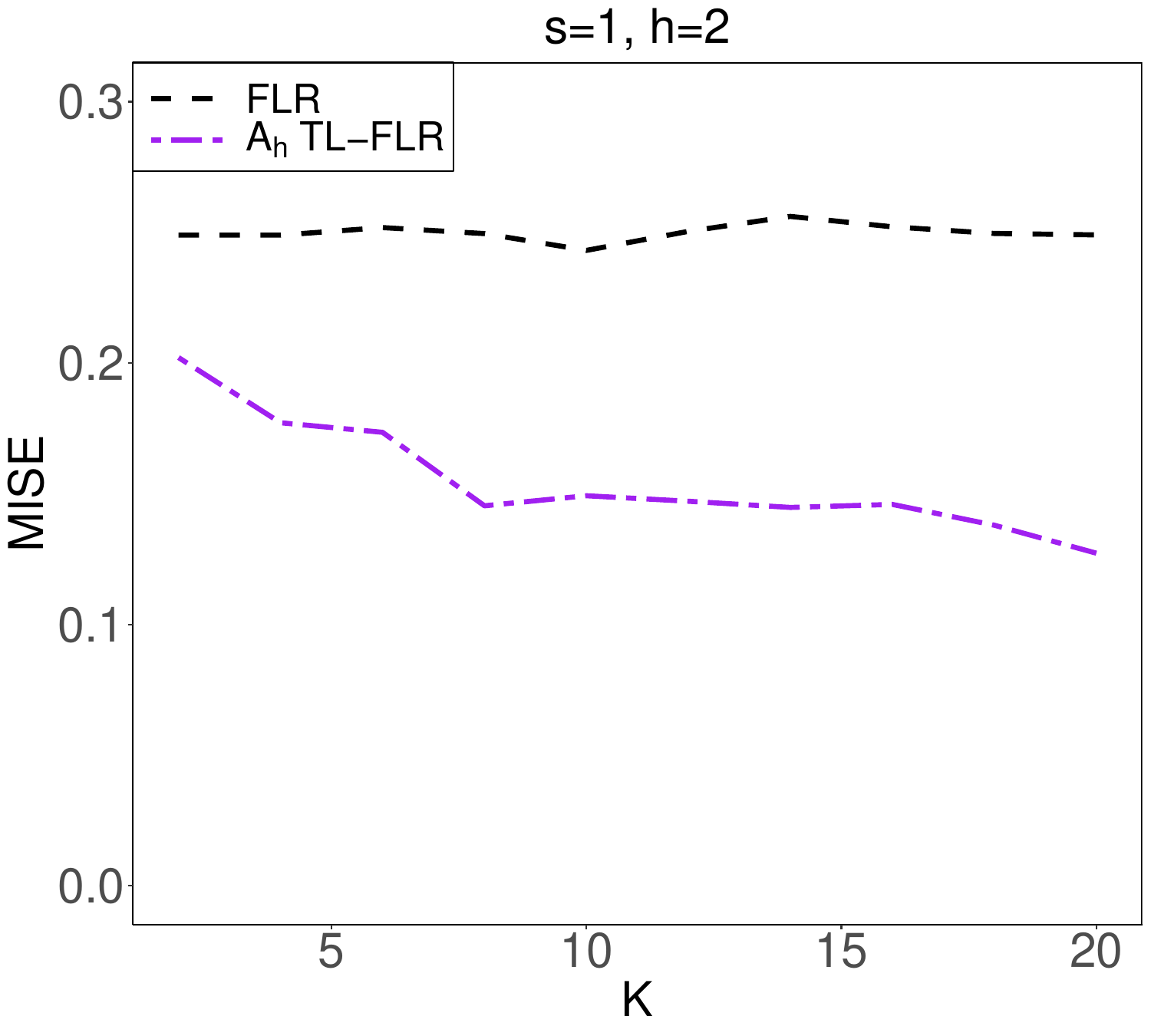} &
				\hspace{\thisgap}\includegraphics[width=\thiswidth]{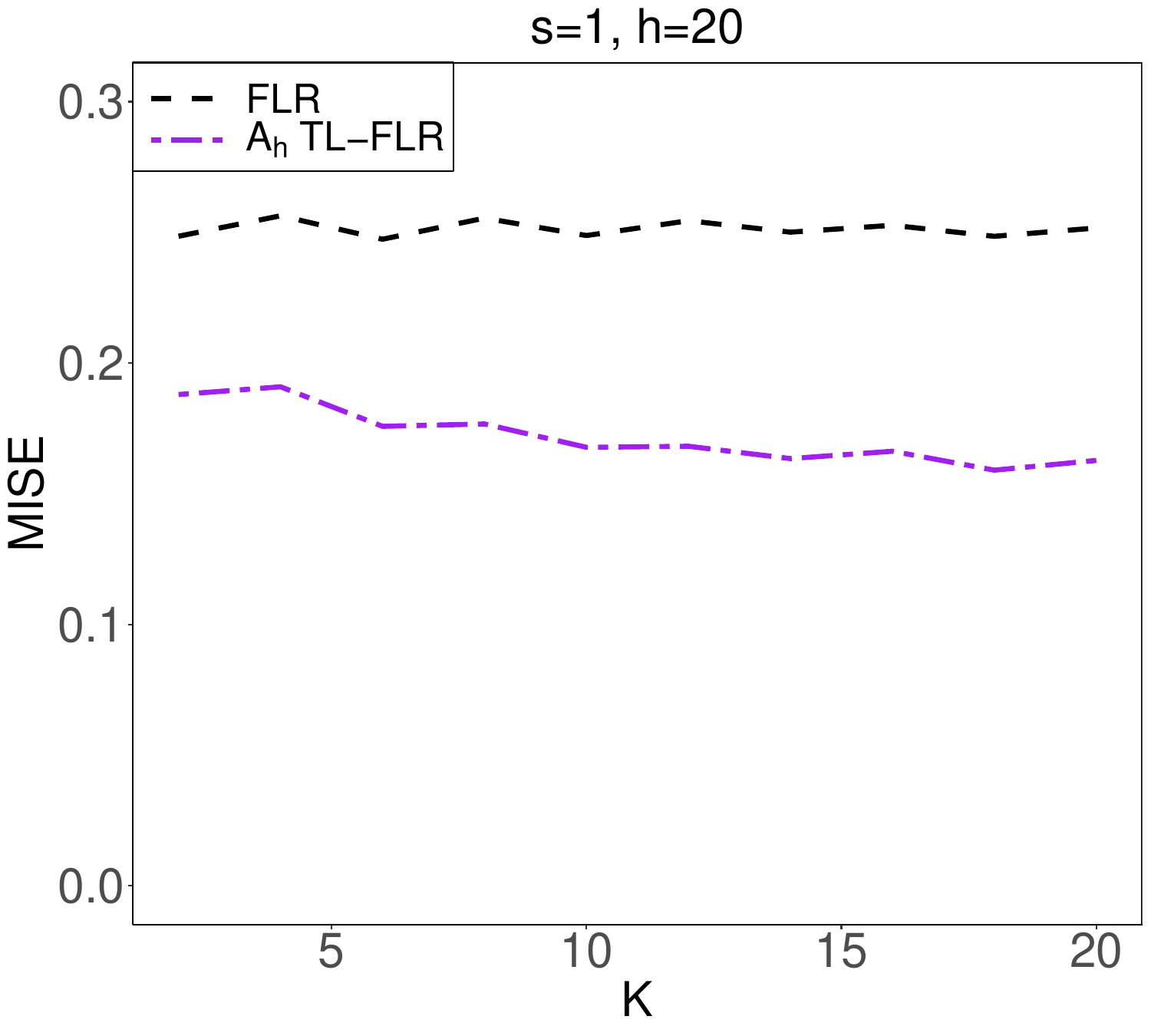} &
				\hspace{\thisgap}\includegraphics[width=\thiswidth]{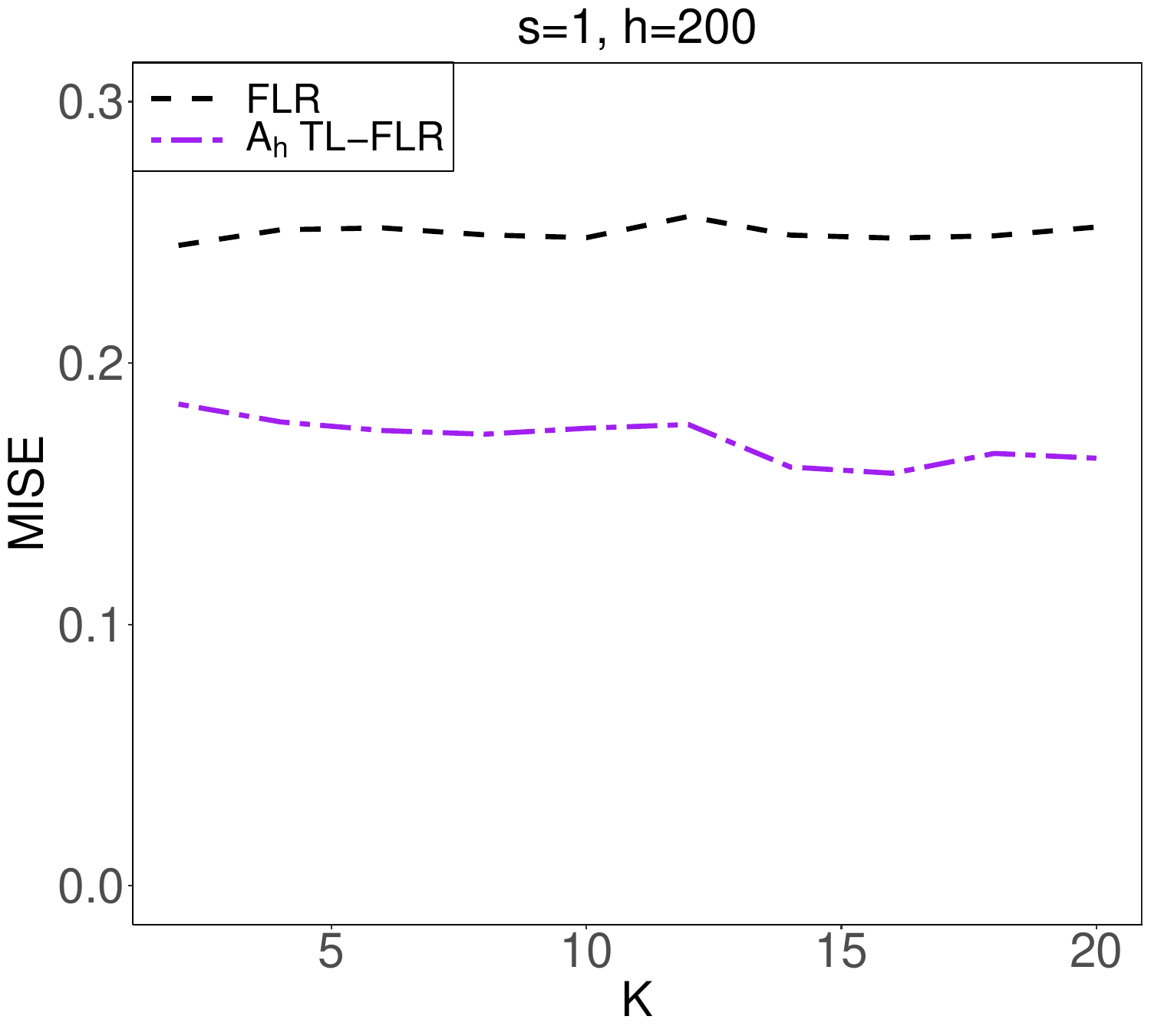}  &
				\hspace{\thisgap}\includegraphics[width=\thiswidth]{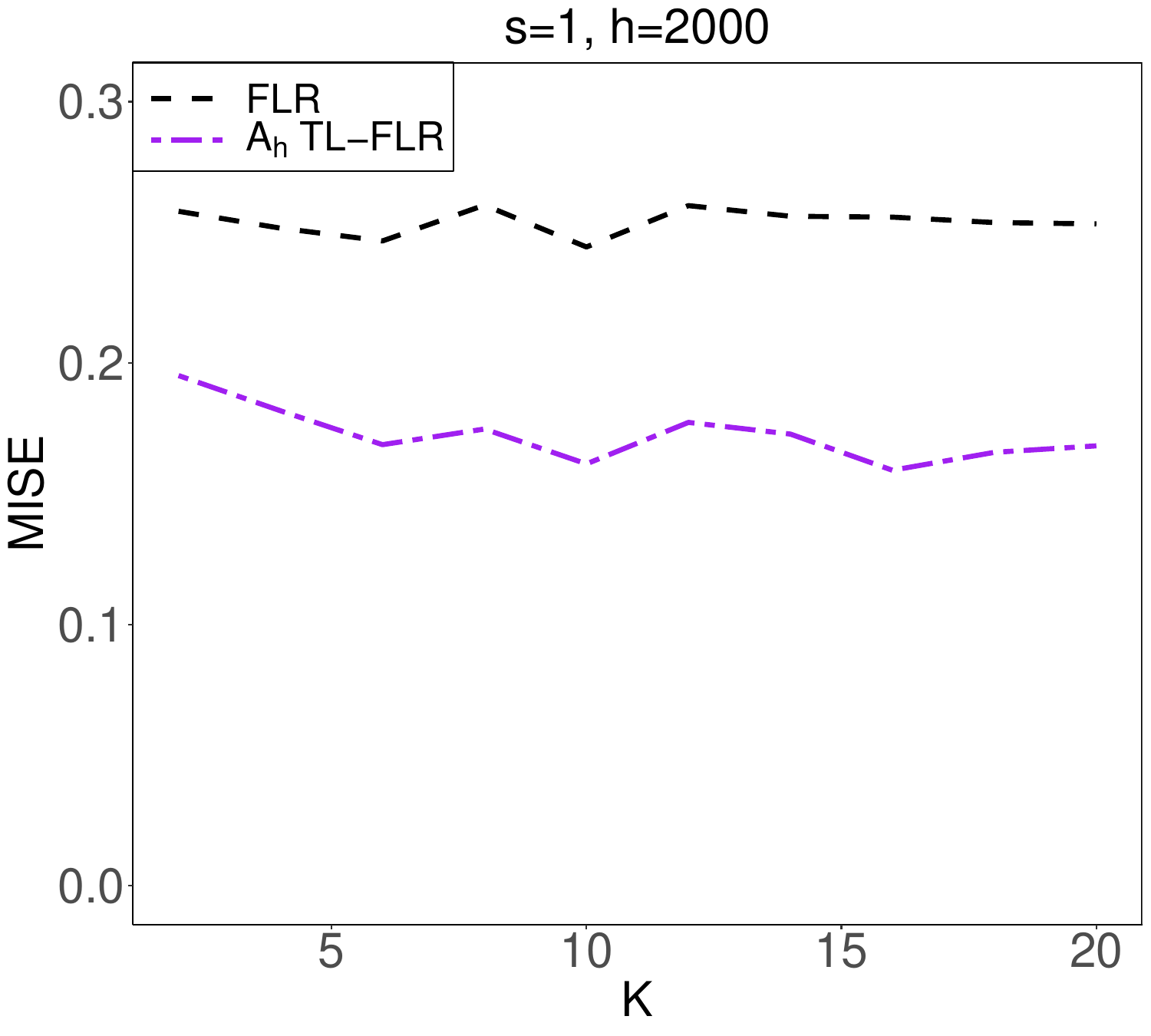}  \\
				\hspace{\thisgap}\includegraphics[width=\thiswidth]{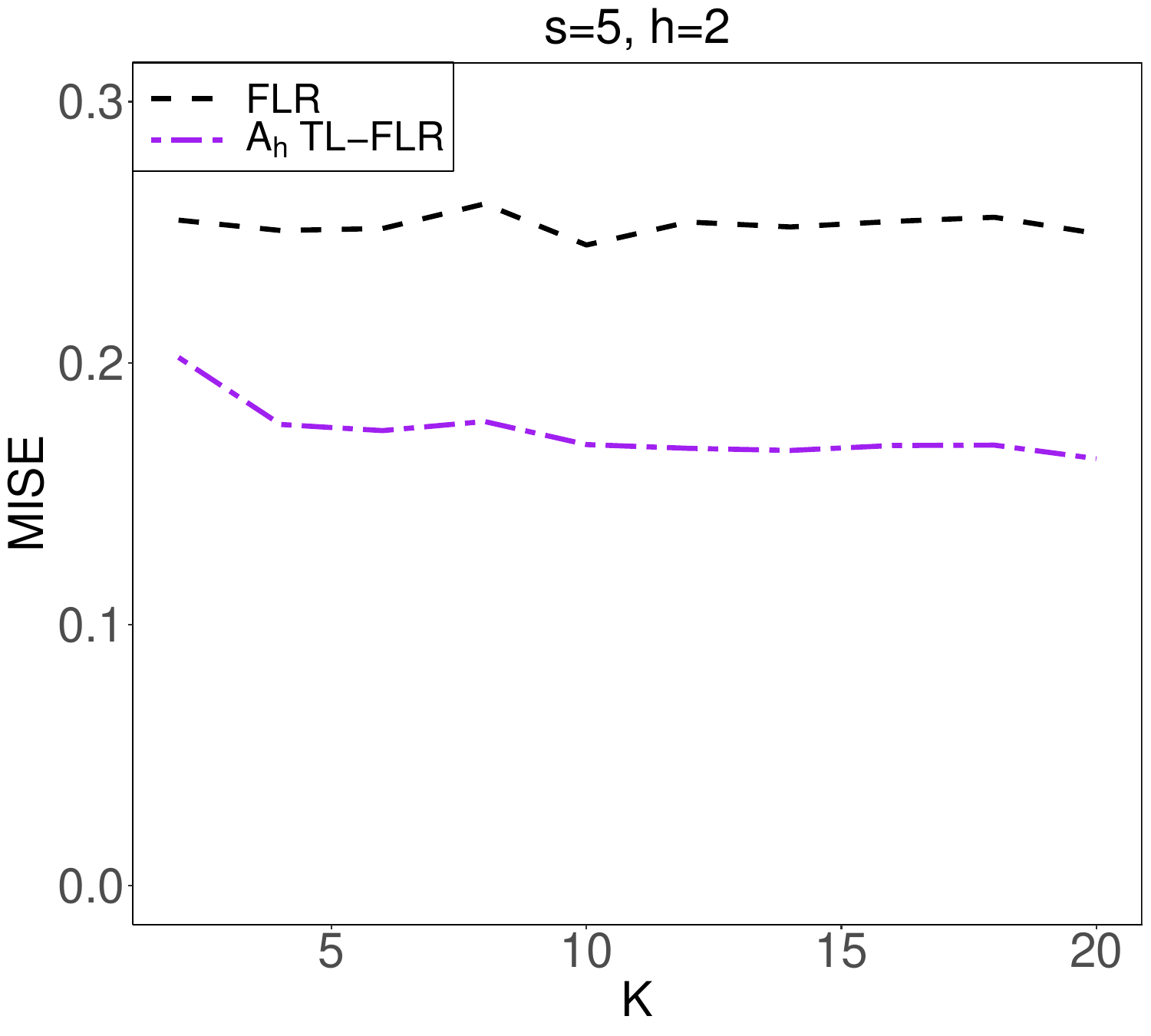} &
				\hspace{\thisgap}\includegraphics[width=\thiswidth]{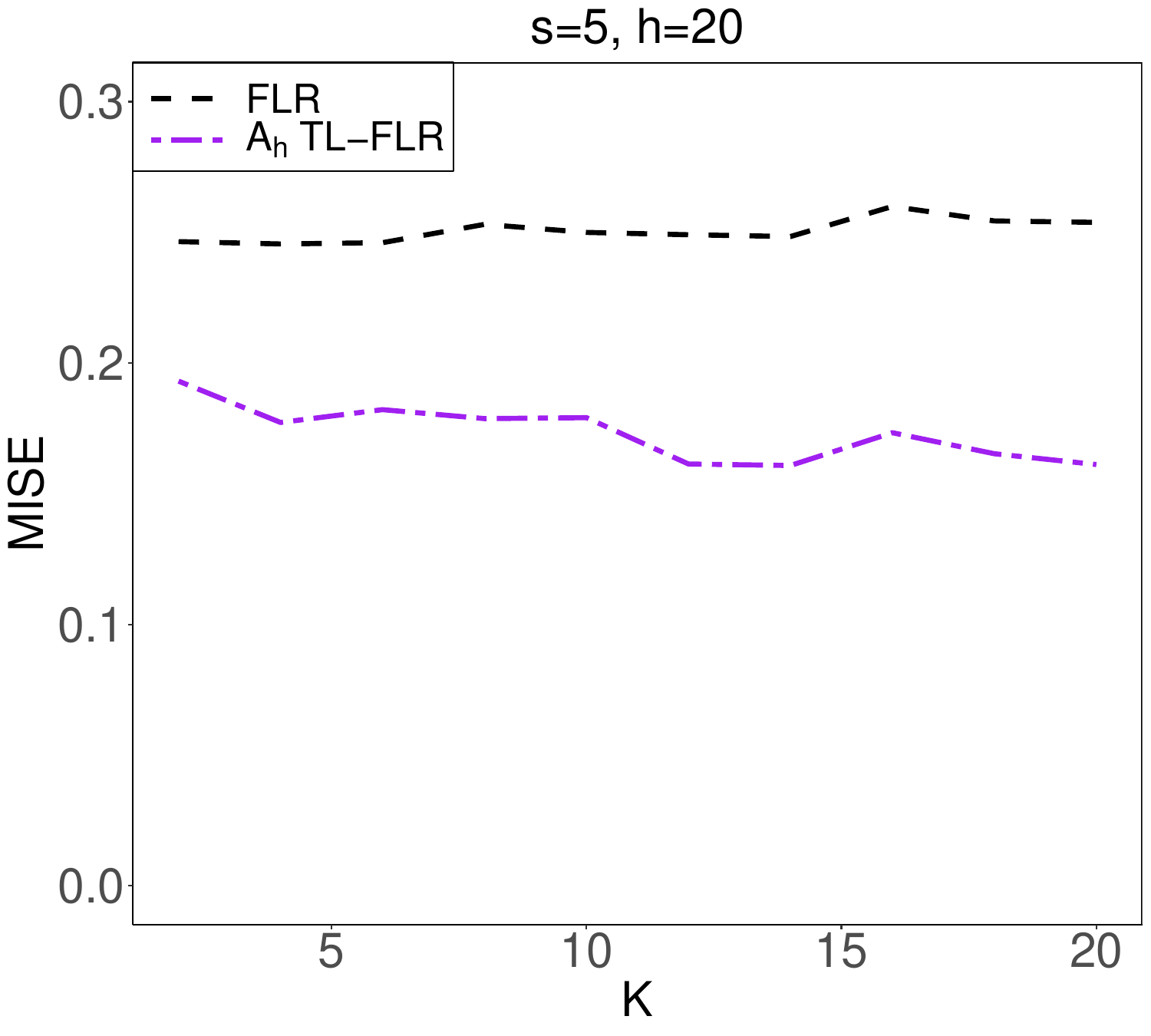} &
				\hspace{\thisgap}\includegraphics[width=\thiswidth]{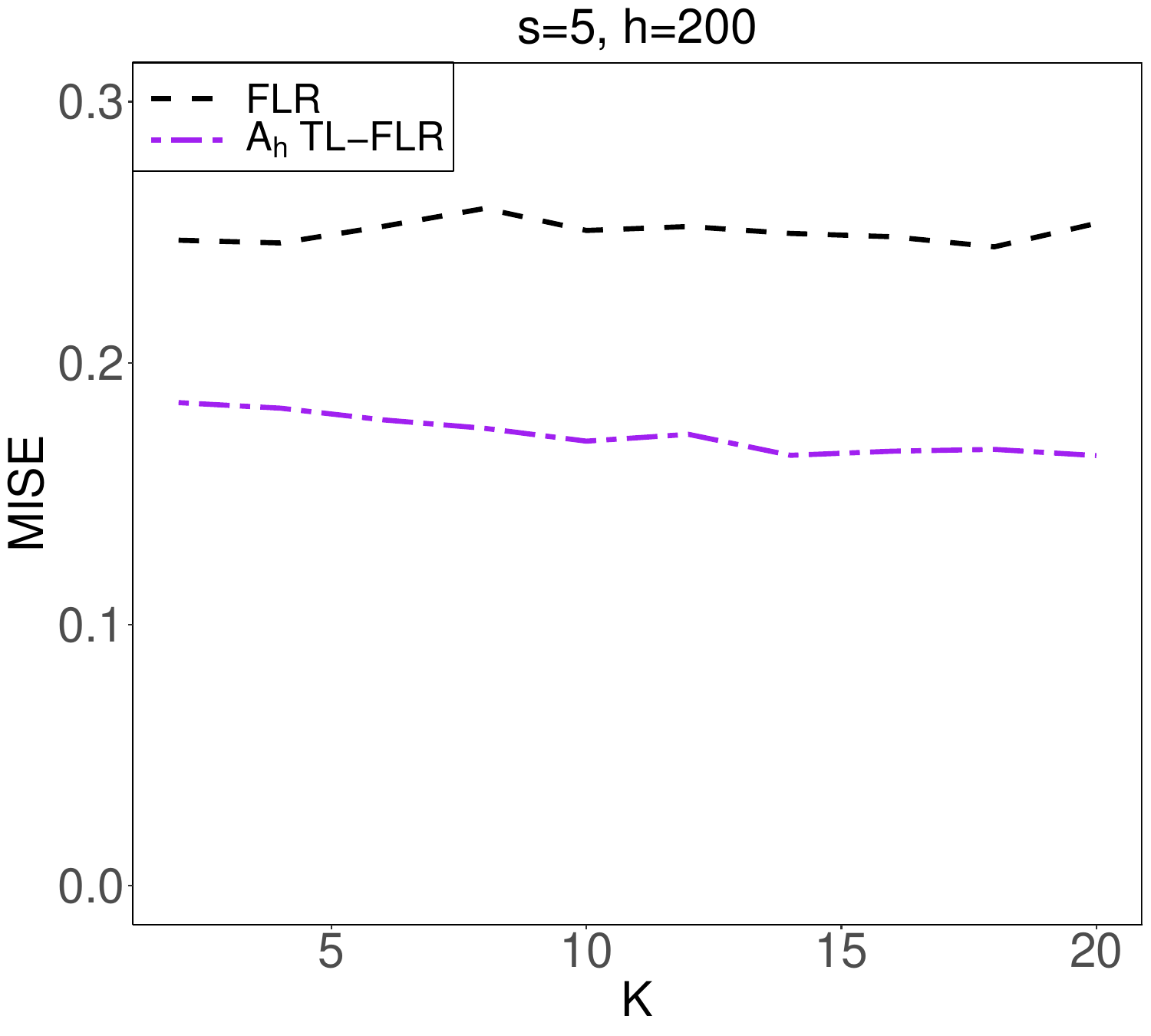} &
				\hspace{\thisgap}\includegraphics[width=\thiswidth]{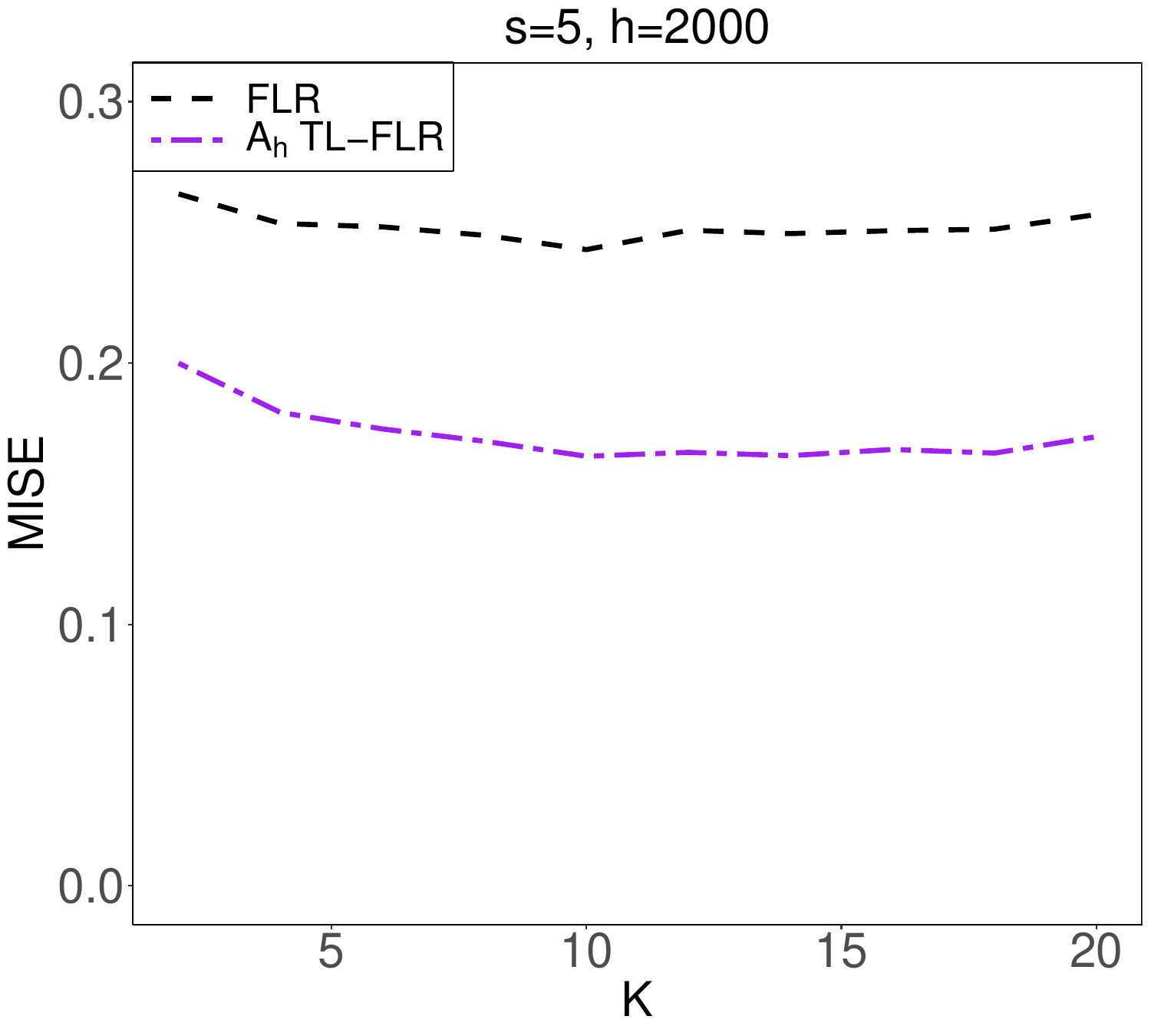}  \\
				\hspace{\thisgap}\includegraphics[width=\thiswidth]{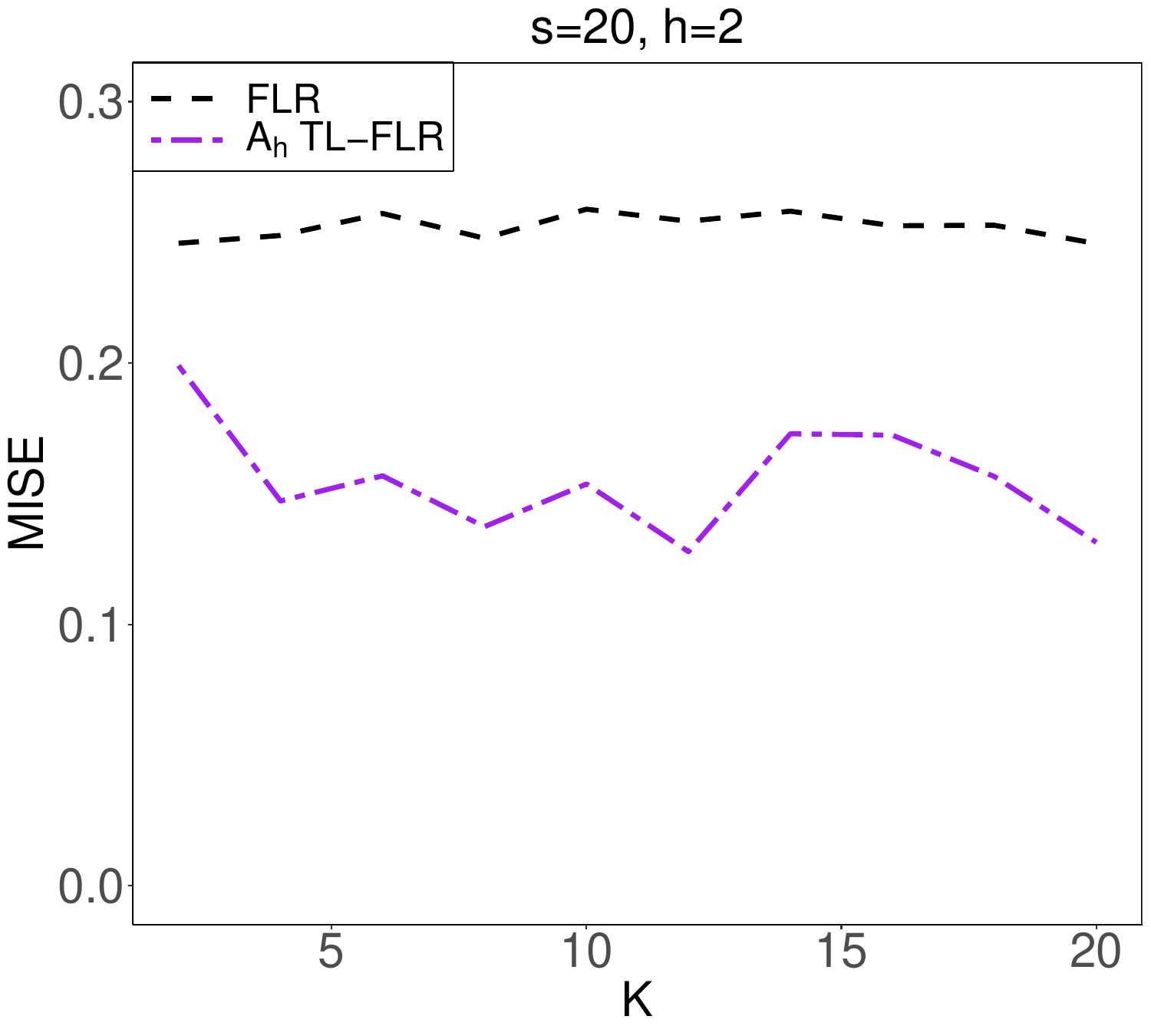} &
				\hspace{\thisgap}\includegraphics[width=\thiswidth]{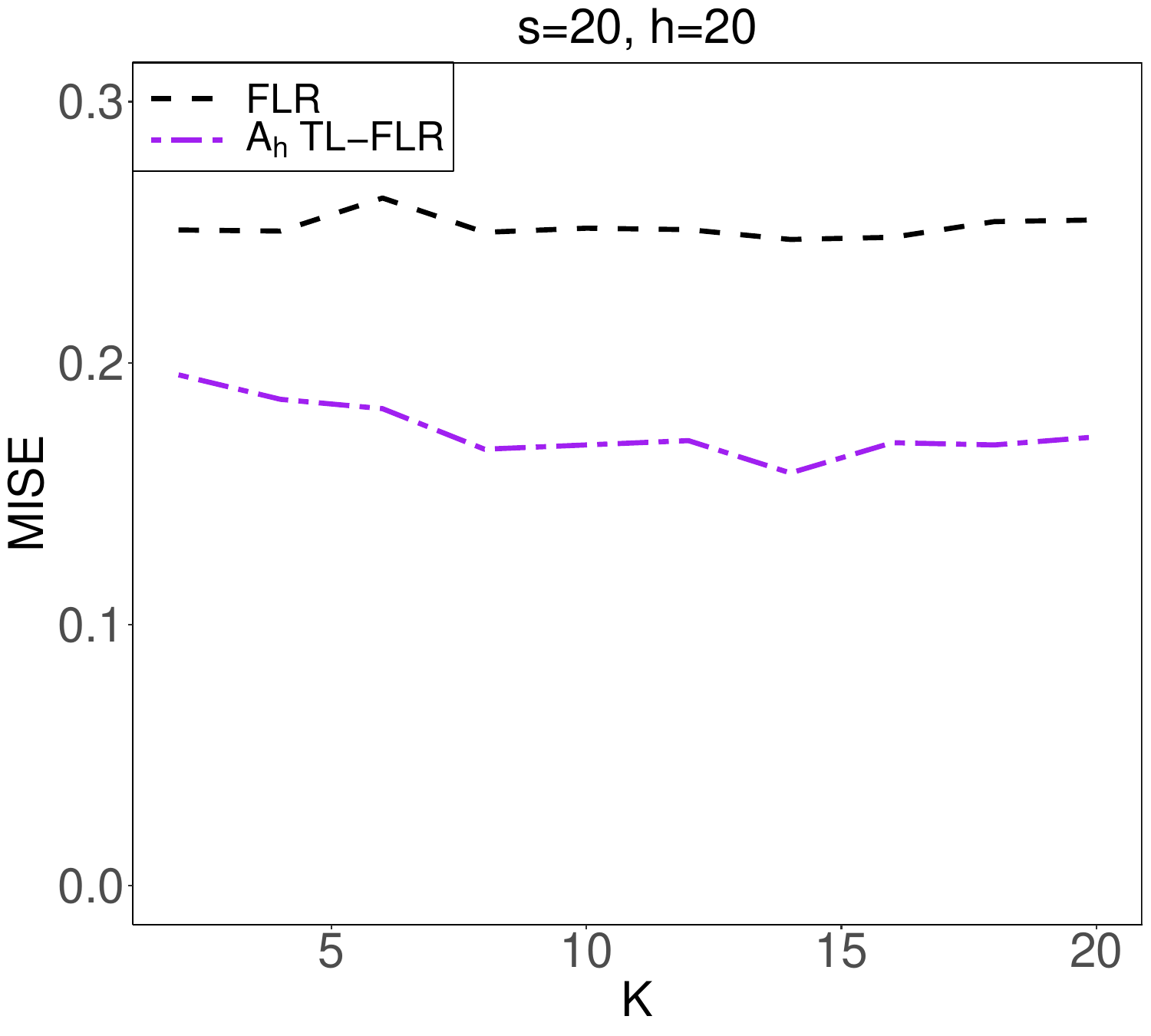} &
				\hspace{\thisgap}\includegraphics[width=\thiswidth]{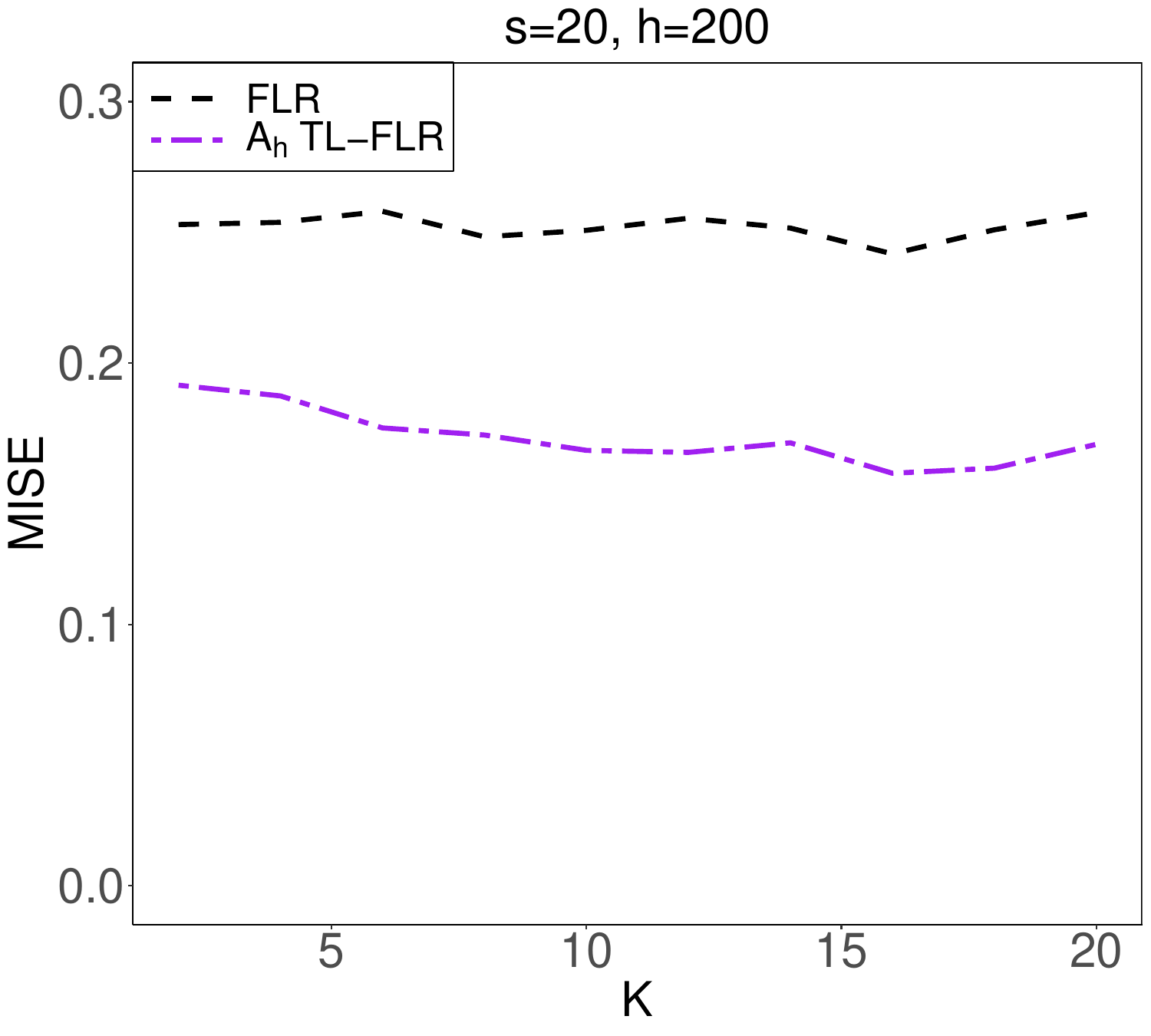} &
				\hspace{\thisgap}\includegraphics[width=\thiswidth]{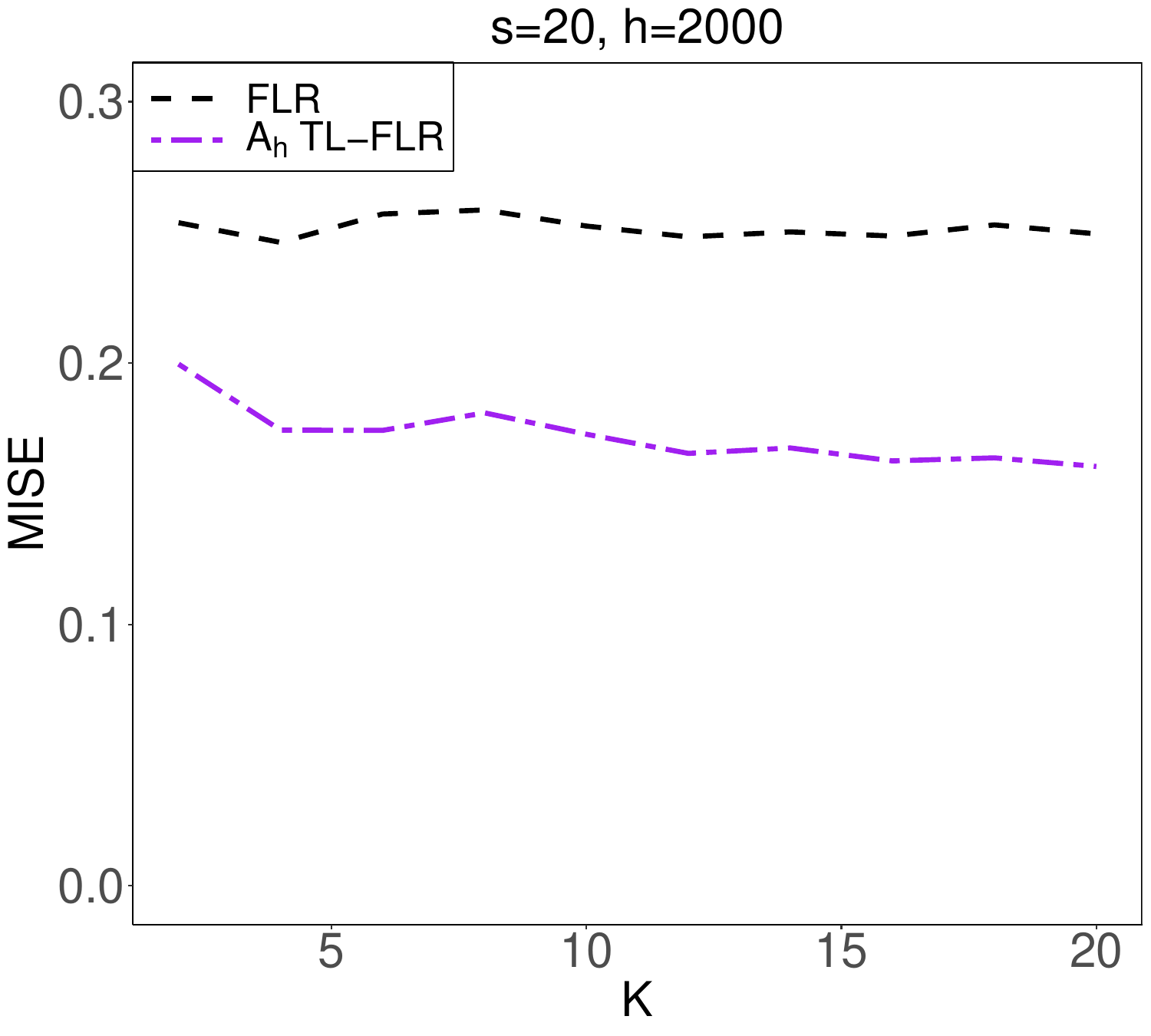}	\\	
				\hspace{\thisgap}\includegraphics[width=\thiswidth]{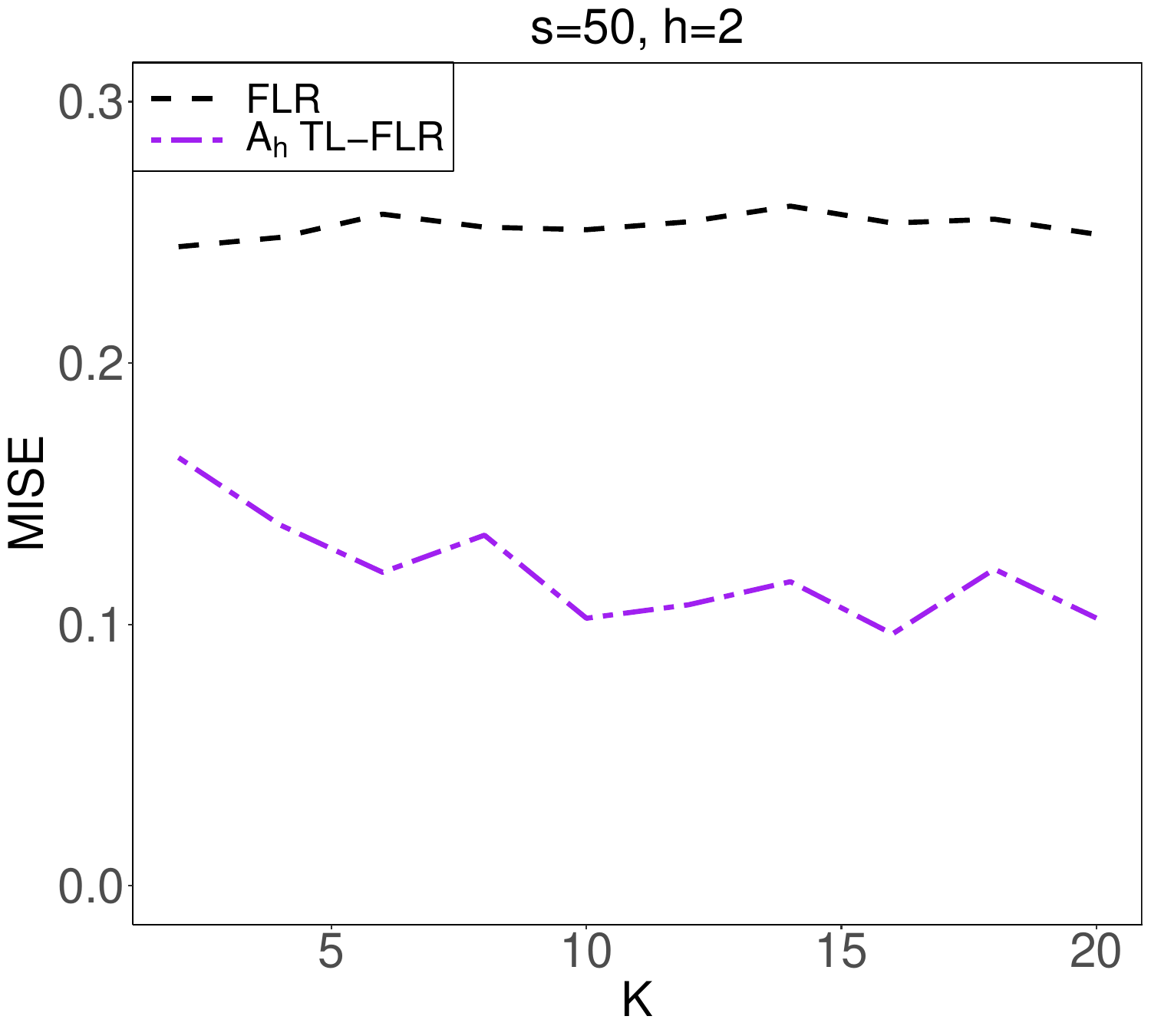} &
				\hspace{\thisgap}\includegraphics[width=\thiswidth]{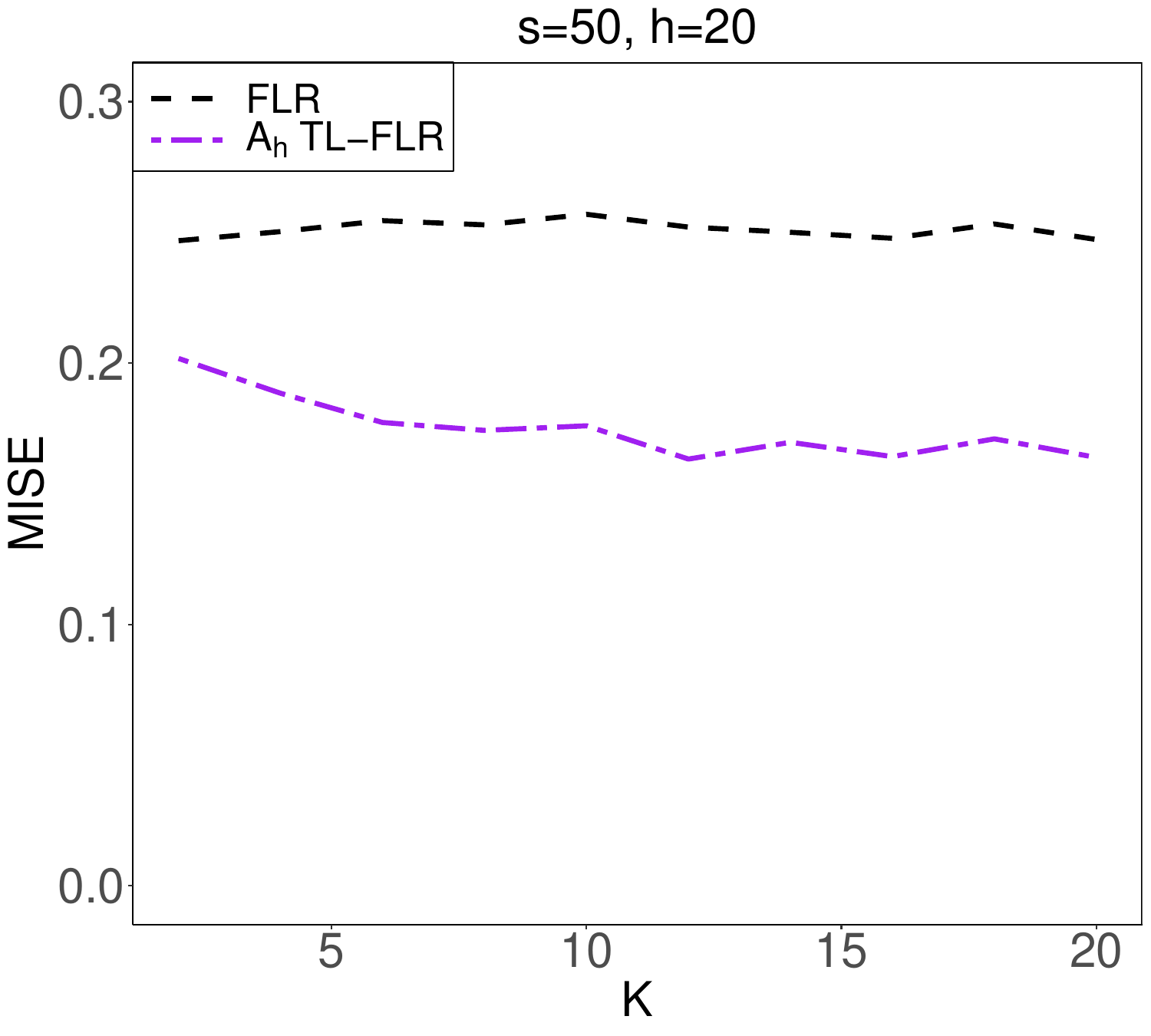} &
				\hspace{\thisgap}\includegraphics[width=\thiswidth]{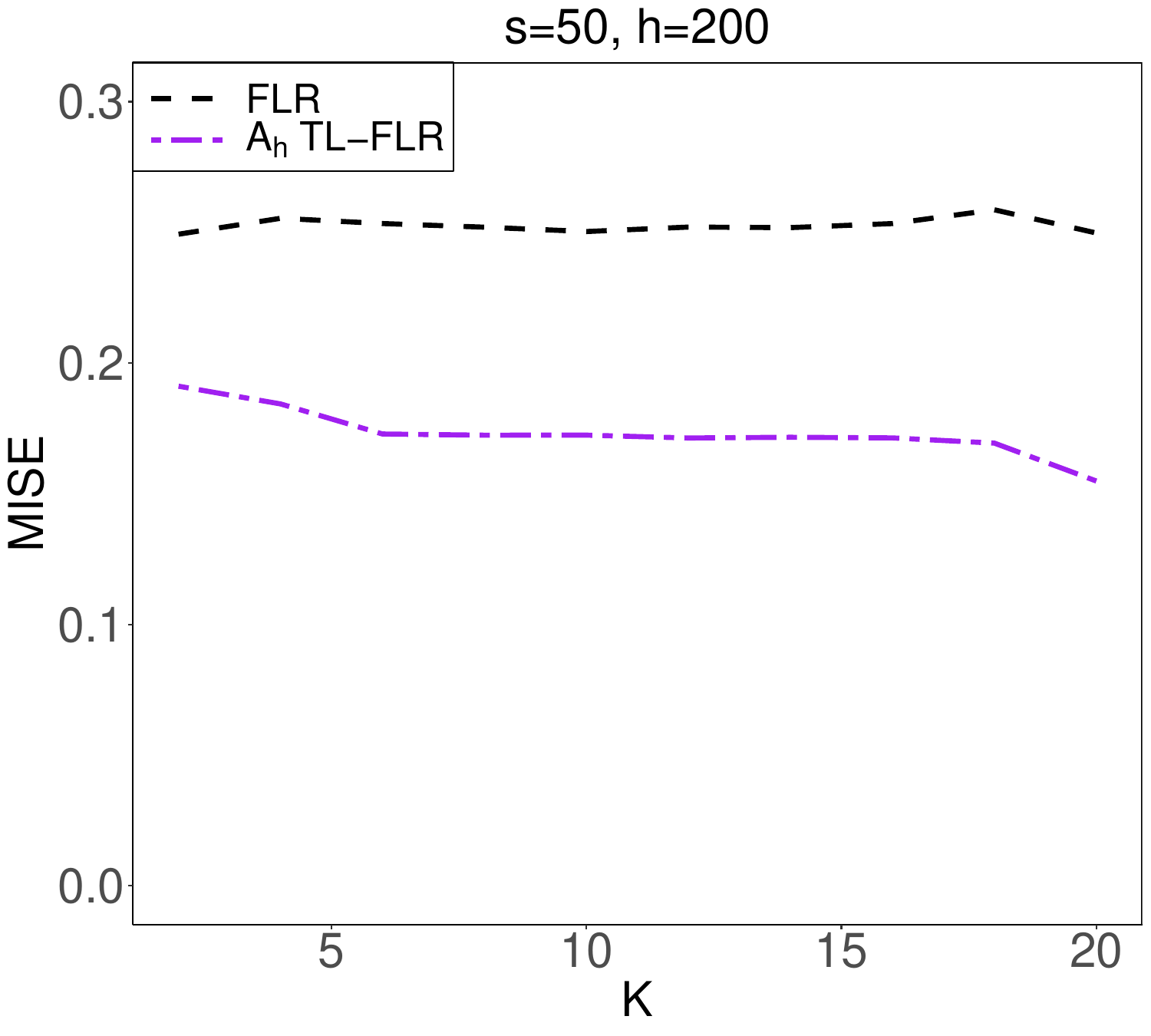} &
				\hspace{\thisgap}\includegraphics[width=\thiswidth]{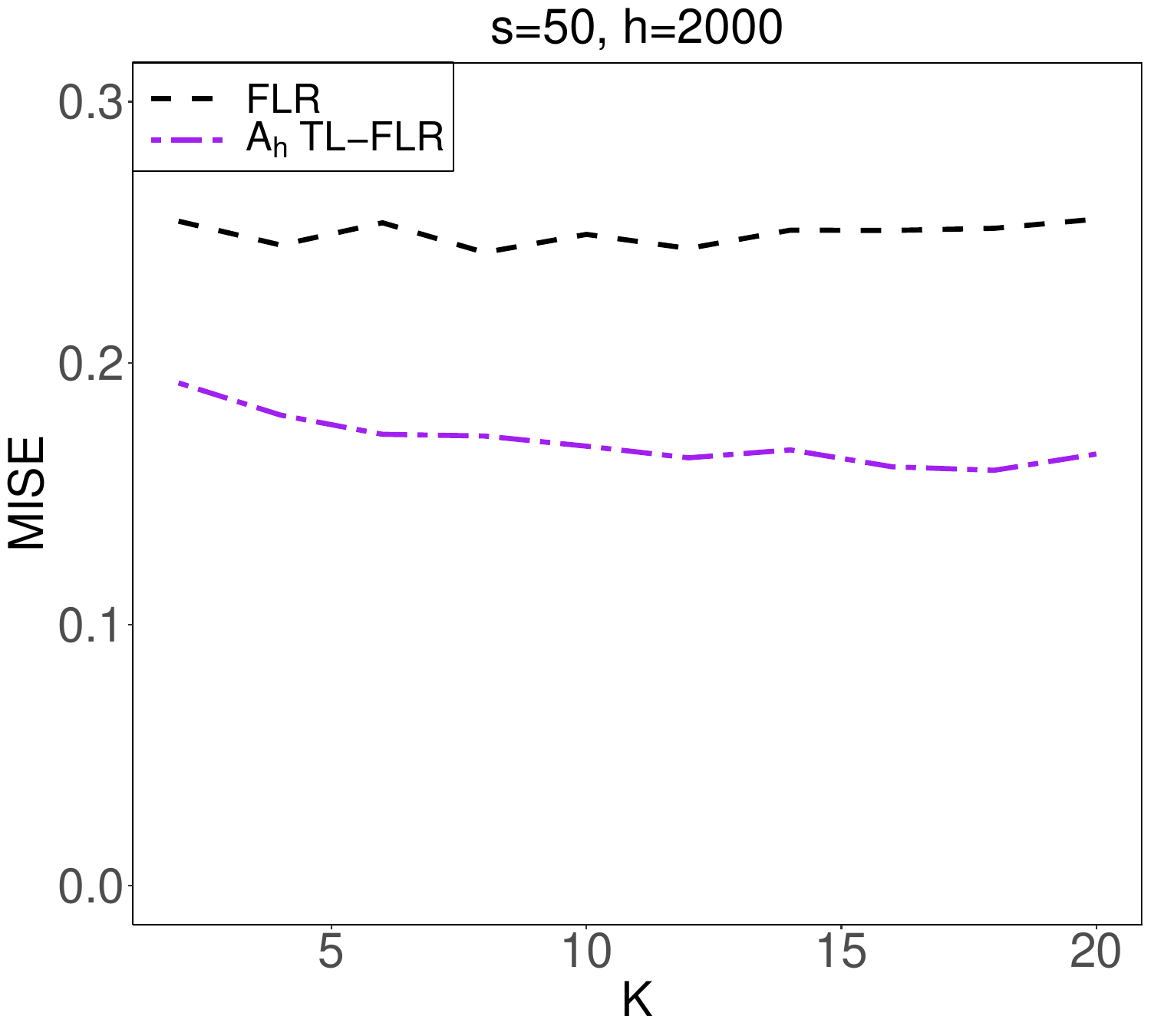}	
			\end{tabular}
			\caption{Estimation errors of different methods under Model (I+) over 1000 repetitions.}
			\label{fig:model-I+}
		\end{figure}

		\begin{figure}[!h]
			\centering
			\newcommand{\thiswidth}{0.3\linewidth}
			\newcommand{\thisgap}{0\linewidth}
			\begin{tabular}{ccc}
				\hspace{\thisgap}\includegraphics[width=\thiswidth]{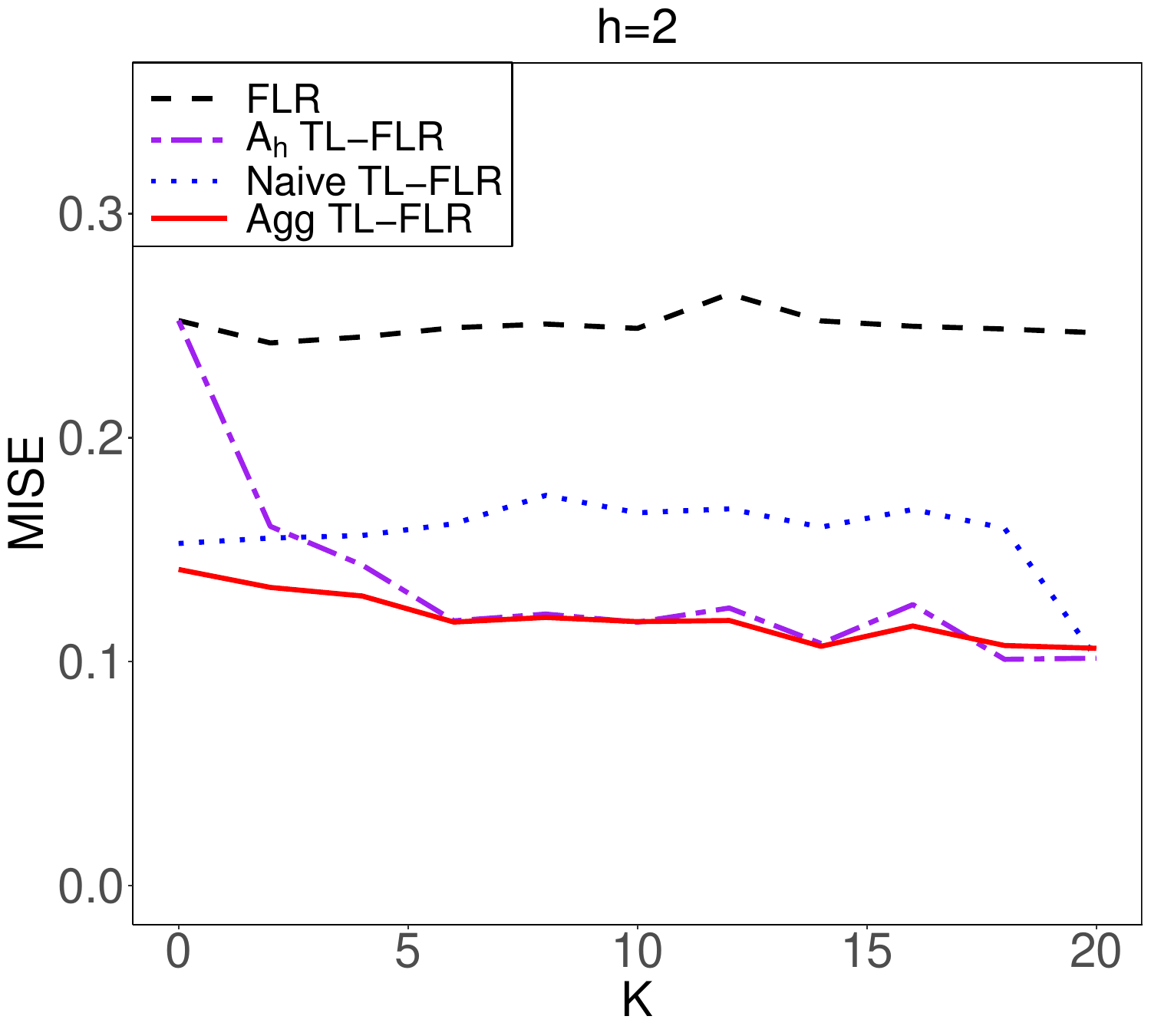} &
				\hspace{\thisgap}\includegraphics[width=\thiswidth]{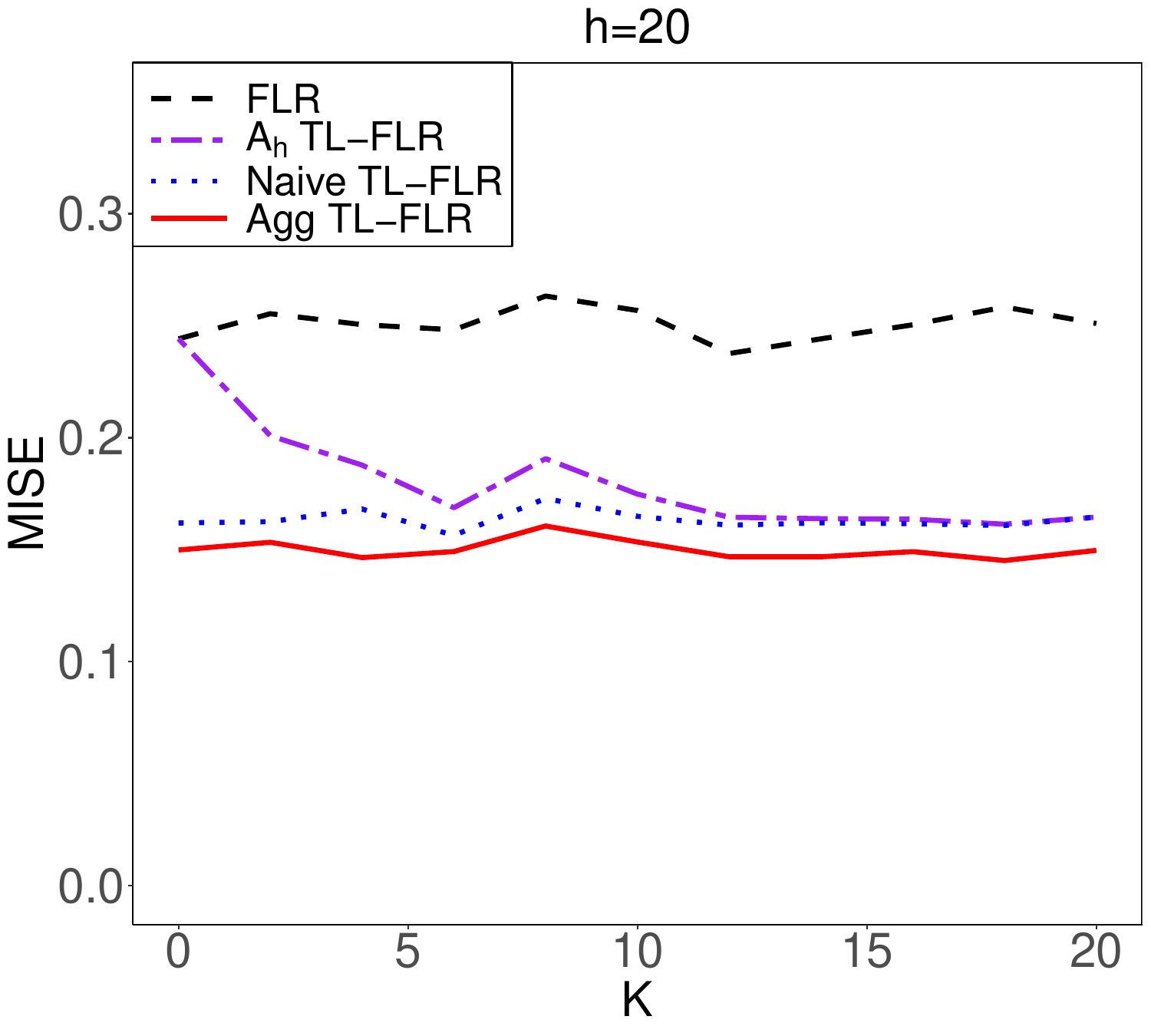} &
				\hspace{\thisgap}\includegraphics[width=\thiswidth]{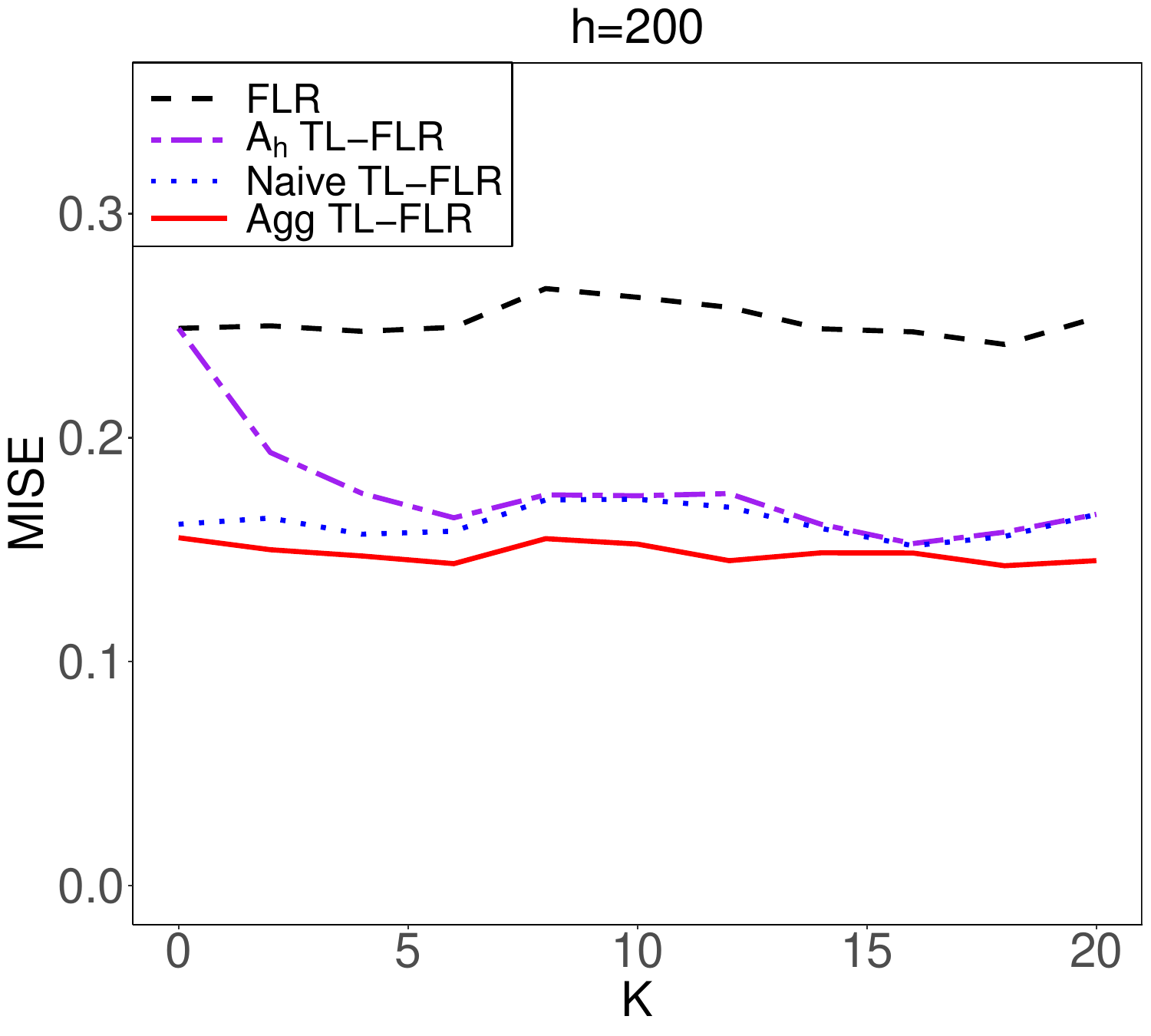} \\
				\hspace{\thisgap}\includegraphics[width=\thiswidth]{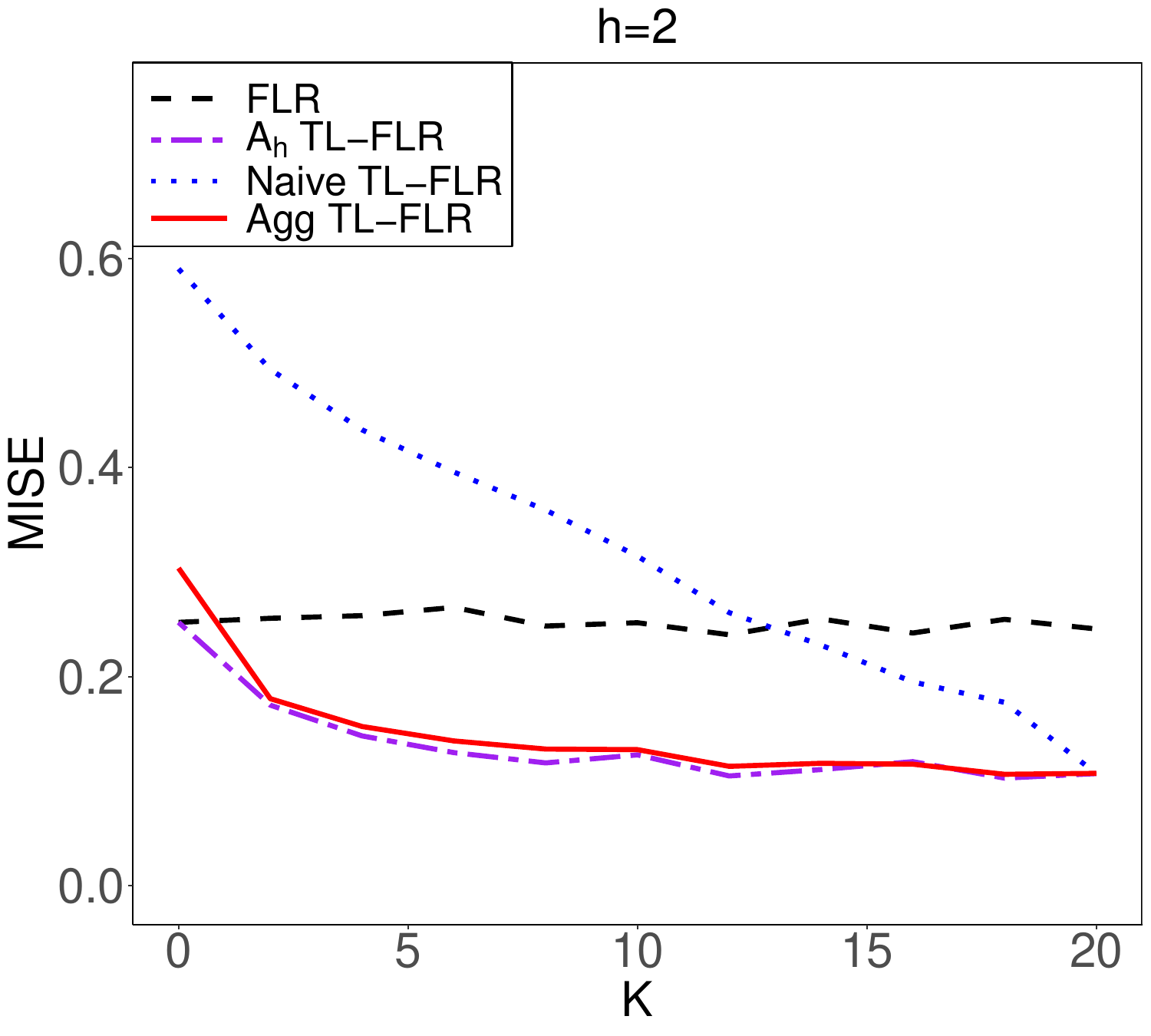} &
				\hspace{\thisgap}\includegraphics[width=\thiswidth]{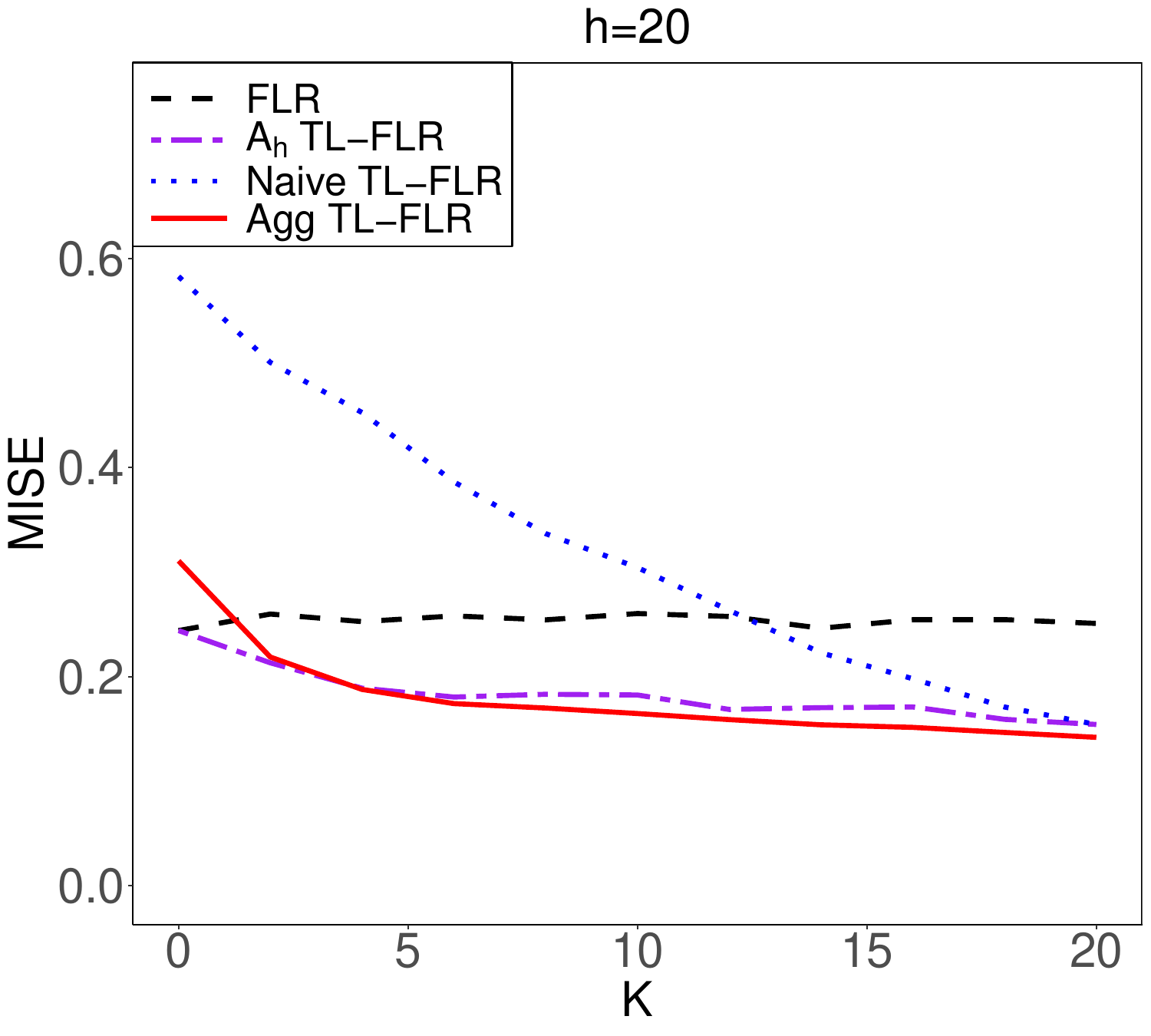} &
				\hspace{\thisgap}\includegraphics[width=\thiswidth]{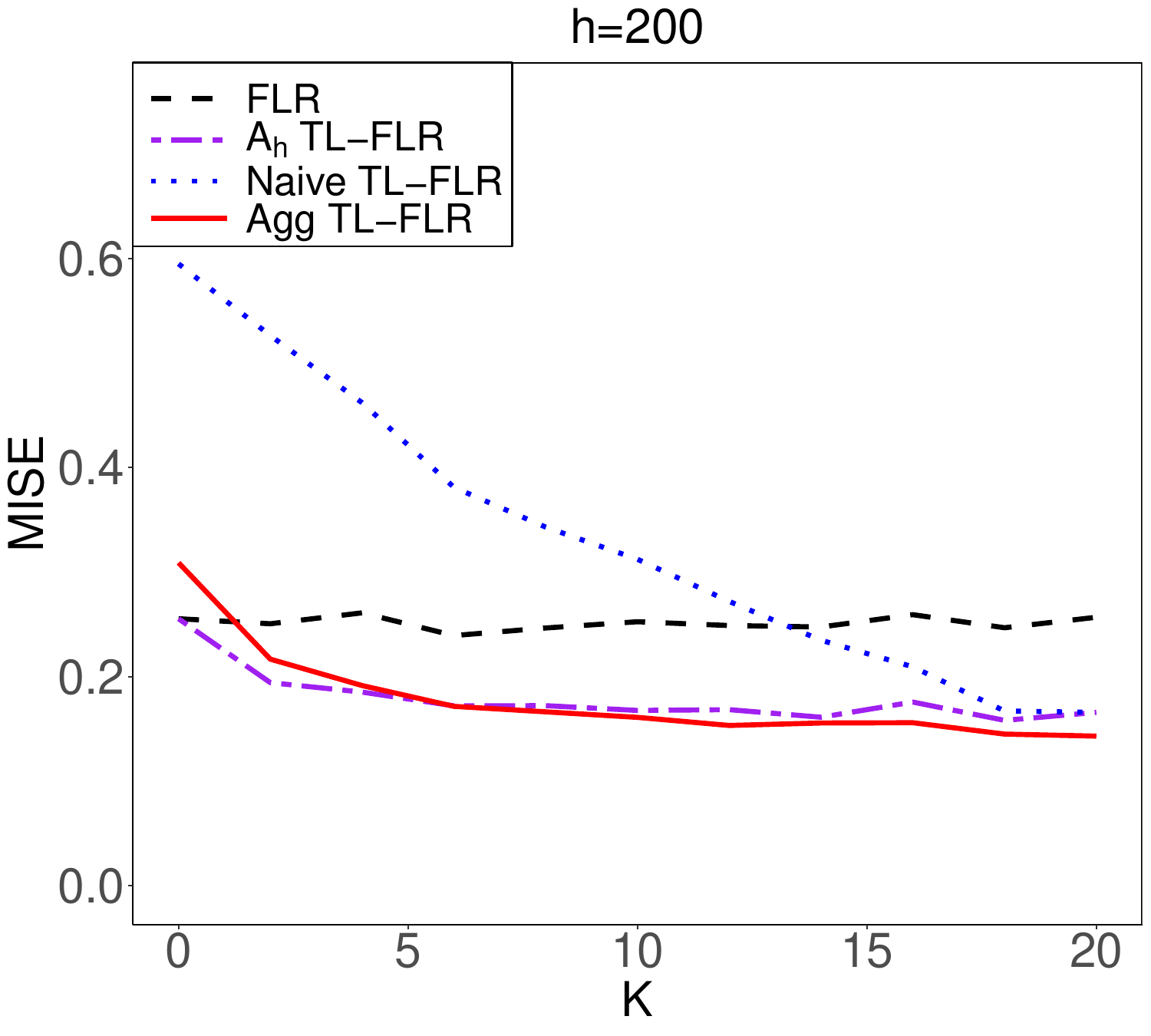} \\
				\hspace{\thisgap}\includegraphics[width=\thiswidth]{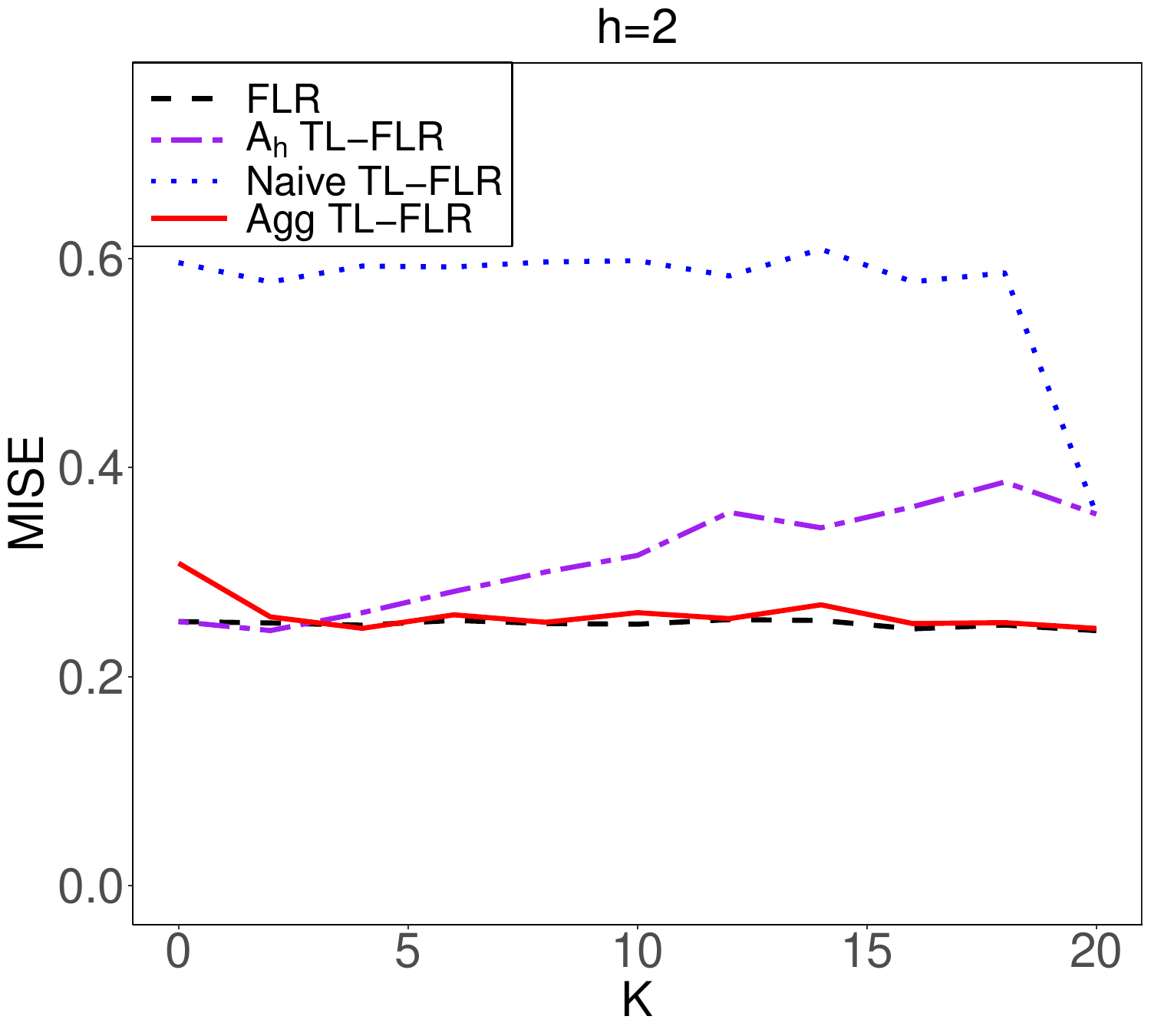} &
				\hspace{\thisgap}\includegraphics[width=\thiswidth]{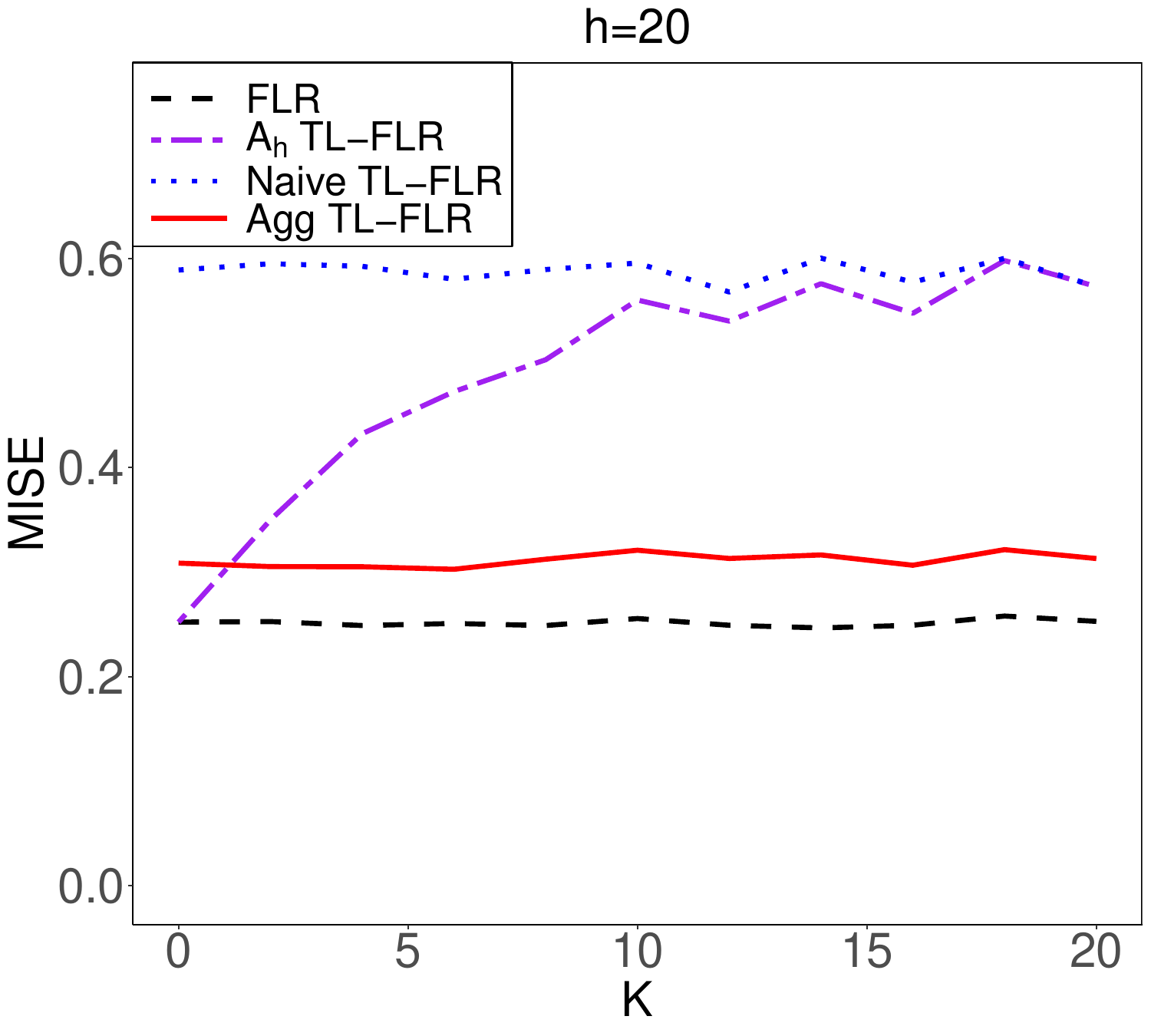} &
				\hspace{\thisgap}\includegraphics[width=\thiswidth]{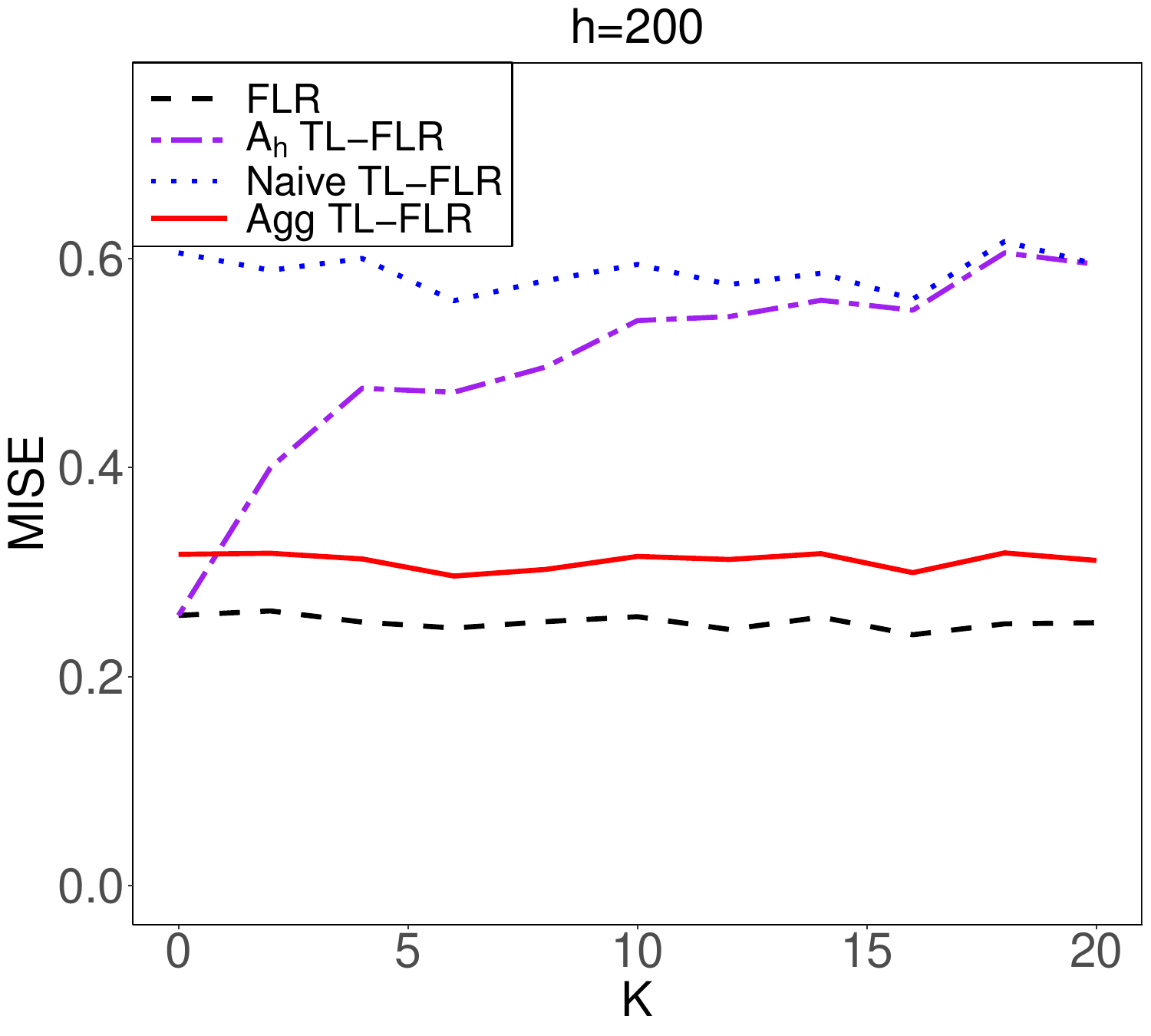} 
			\end{tabular}
			\caption{Estimation errors of different methods over 500 repetitions. Top row: Model (II+); middle row: Model (III+); bottom row: Model (IV+).}
			\label{fig:adap+}
		\end{figure}
		
		\section{Preliminaries on Proofs}
		
		For ease of presentation, we write $\hat \phi_k = \hat \phi_k^{\mathcal A}$, $\int pq = \int p(t) q(t) dt$ and $\int Mpq = \int\int M(s,t)p(s)q(t)dsdt$.
		Let $\hat \Delta = \big\{ \int \int \big(\hat K^{\mathcal A}(s,t) - K^{\mathcal A}(s,t)\big)^2 dsdt \big\}^{1/2} $, $\hat \Delta_k  = [\int  \{\int \big( \hat K^{\mathcal A}(s,t) - K^{\mathcal A}(s,t) \big) \phi_k(s) ds \}^2 dt]^{1/2} $ and $\hat \Delta_{kj} = \big|  \int \big( \hat K^{\mathcal A} - K^{\mathcal A}\big) \phi_k \phi_j \big|$, where $K^{\mathcal A} = \sum_{l=1}^L \pi_l K^{(l)}(s,t)$. Denote $\lambda_k^{\mathcal A} = \sum_{l=1}^L \pi_l \lambda_k^{(l)}. $
		
		Let $g(t) = \sum_{l=1}^L \pi_l \cov\big(Y^{(l)}, X^{(l)}(t)  \big)$ and $$\hat g(t) = \sum_{l=1}^L \pi_l (n_l-1)^{-1} \sum_{i=1}^{n_l} \left\{(X_i^{(l)}(t) - \bar{X}^{(l)}(t) )(Y_i^{(l)} - \bar{Y}^{(l)} )\right\} .$$  Then, we have $\langle g, \phi_k \rangle = \sum_{l=1}^L \pi_l \lambda_k^{(l)} w_k^{(l)} =  \sum_{l=1}^L \pi_l \lambda_k^{(l)} (b_k - \delta_k^{(l)}) $.
		Let $\bepsilon = (\epsilon_1, \dots, \epsilon_n)^{\t}$ and $\bar{\bepsilon}$ is a $n$-dimensional vector with elements $\bar{\epsilon} = n^{-1}\sum_{i=1}^n \epsilon_i$.
		Define the events 
		\begin{align}
			\mathcal E_1 &= \{ (\hat \lambda_k^{\mathcal A} - \lambda_j^{\mathcal A})^{-2} \le 2 (\lambda_j^{\mathcal A} - \lambda_k^{\mathcal A})^{-2} \le C m^{2(\alpha+1)} \text{ for } k,j=1,\dots, m  \} , \label{eq:event1}\\
			\mathcal E_2 & = \{ \lambda_m^{\mathcal A} \ge 2 \hat \Delta \} \label{eq:event2}, \\
			\mathcal E_3 &= \{ n^{-1} \| \hat \Xi^{\t} (\bepsilon - \bar{\bepsilon})\|_{\infty} \le \tau/2 \} . \label{eq:event3}
		\end{align}
		
		\section{Proof of Theorem \ref{thm:upper}}
		
		\begin{proof}[Proof of Theorem \ref{thm:upper}]
			
			Note that 
			\begin{eqnarray*}
				\| b - \hat b \|_2^2 & = & \int \left( \sum_{k=1}^{\infty} b_k \phi_k(t) - \sum_{k=1}^m \hat b_k \hat \phi_k(t)  \right)^2 dt \\
				& \le & 3\sum_{k=m+1}^{\infty} b_k^2 + 3 \sum_{k=1}^m (b_k - \hat b_k)^2 + 3 \int \left\{ \sum_{k=1}^m b_k \big( \phi_k(t) - \hat \phi_k(t) \big) \right\}^2 dt.
			\end{eqnarray*}
			
			Observe that
			\[ \sum_{k=m+1}^{\infty} b_k^2 = O(m^{1-2\beta}),  \]
			and from Lemma \ref{lem:eigfun},
			\begin{eqnarray*}
				\int \left\{ \sum_{k=1}^m b_k \big( \phi_k(t) - \hat \phi_k(t) \big) \right\}^2 dt & \le & m \sum_{k=1}^m b_k^2 \| \phi_k - \hat \phi_k \|_2^2 = O_P(m N^{-1}).
			\end{eqnarray*}
			It remains to bound $\sum_{k=1}^m (b_k - \hat b_k)^2$. If $\tau \asymp n^{-1/2}$,  it is handled by bounding $\|\bw - \hat \bw\|_2^2$ in Lemma \ref{lem:rate_w} and $\|\bdelta - \hat \bdelta\|_2^2$ in Lemma \ref{lem:rate_delta} according to the triangle inequality. If $\tau =0$, it is addressed in Lemma \ref{lem:w/o_lasso}.
			Combining the above pieces yields the claim.
		\end{proof}
		
		\begin{lemma}\label{lem:prob_events}
			Recall the events $\mathcal E_1, \mathcal E_2$ and $\mathcal E_3$ as defined in \eqref{eq:event1}-\eqref{eq:event3}. 
			\begin{itemize}
				\item If $m^{2(\alpha+1)}N^{-1} \to 0$, then $\prob(\mathcal E_1) \to 1$ and $\prob(\mathcal E_2) \to 1$.
				\item If $\tau \ge c n^{-1/2}$ for some large constant $c>0$, then $\prob(\mathcal E_3) \to 1$.
			\end{itemize}
		\end{lemma}
		
		\begin{proof}
			
			\textbf{The event $\mathcal E_1$.} If $\hat \Delta \le (1-2^{1/2}/2) |\lambda_k^{\mathcal A} -  \lambda_j^{\mathcal A}| $, then the triangle inequality implies
			\begin{align*}
				| \hat \lambda_k^{\mathcal A} - \lambda_j^{\mathcal A} | & \ge |\lambda_k^{\mathcal A} -  \lambda_j^{\mathcal A}| - |\hat \lambda_k^{\mathcal A} - \lambda_k^{\mathcal A} | \\
				& \ge |\lambda_k^{\mathcal A} -  \lambda_j^{\mathcal A}| - \hat \Delta \\
				& \ge 2^{1/2}/2 |\lambda_k^{\mathcal A} -  \lambda_j^{\mathcal A}| ,
			\end{align*}
			where the second inequality holds due to the fact that $\sup_{k \ge 1}| \hat \lambda_k^{\mathcal A} - \lambda_k^{\mathcal A}  | \le \hat \Delta$. Thus, for $k,j=1,\dots, m$,
			\begin{align*}
				(\hat \lambda_k^{\mathcal A} - \lambda_j^{\mathcal A})^{-2} & \le 2 (\lambda_k^{\mathcal A} -  \lambda_j^{\mathcal A})^{-2} \le Cm^{2(\alpha+1)}.
			\end{align*}
			Since $\hat \Delta = O_P(N^{-1/2})$ in Lemma \ref{lem:Delta}, then if $m^{2(\alpha+1)}N^{-1} \to 0$, then we have $\prob( \mathcal E_1) \to 1$. 
			
			\textbf{The event $\mathcal E_2$.} Assumption \ref{assump:cov-decay} implies $\lambda_m^{\mathcal A} \ge c m^{-\alpha}$. From Lemma \ref{lem:Delta}, $\hat \Delta = O_P(N^{-1/2})$. If $m^{2\alpha}N^{-1} \to 0$, then $\prob(\mathcal E_2) \to 1$.
			
			\textbf{The event $\mathcal E_3$.} Note that
			\begin{align*}
				(n^{-1} \| \hat \Xi^{\t} (\bepsilon - \bar{\bepsilon})\|_{\infty})^2 & = \max_{1 \le k \le m} \bigg| \frac{1}{n} \sum_{i=1}^n (\epsilon_i - \bar{\epsilon}) \int(X_i - \bar X) \hat \phi_k  \bigg|^2 \\
				& \le  \sum_{k=1}^m \left( \frac{1}{n} \sum_{i=1}^n (\epsilon_i - \bar{\epsilon}) \int(X_i - \bar X) \hat \phi_k \right)^2  \\
				& = O_P\left\{ \sum_{k=1}^m \expect \left( n^{-1} \sum_{i=1}^n \epsilon_i \int(X_i - \bar X) \hat \phi_k  \right)^2    \right\} \\
				&  = O_P(n^{-1}).
			\end{align*}
			If we take $\tau \ge c n^{-1/2}$ for some large constant $c>0$, then $\prob(\mathcal E_3) \to 1$.
			
		\end{proof}
		
		\begin{lemma}\label{lem:rate_w}
			Under Assumptions \ref{assump:xdist}-\ref{assump:x-moment}, if
			\[ N^{-1}m^{2(\alpha+1)} \to 0,  \]
			then we have	
			\begin{align*}
				& \|\hat \bw - \bw\|_2^2 = O_P\left( \frac{m^{1+\alpha}}{N} +  \frac{m^3h^2}{N} \right).
			\end{align*}
		\end{lemma}
		
		\begin{proof}
			To bound $\sum_{k=1}^m (w_k - \hat w_k)^2$, note that
			\[ \hat{w}_k = \frac{ \langle \hat g, \hat \phi_k \rangle }{\hat{\lambda}_k^{\mathcal A}}, ~~ w_k = \frac{\langle g, \phi_k\rangle}{\lambda_k^{\mathcal A}}.  \]
			%	where $\hat g(t) = \sum_{l=1}^L \pi_l (n_l-1)^{-1} \sum_{i=1}^{n_l} \left\{(X_i^{(l)}(t) - \bar{X}^{(l)}(t) )(Y_i^{(l)} - \bar{Y}^{(l)} )\right\} $ and $g(t) = \sum_{l=1}^L \pi_l \cov\big(Y^{(l)}, X^{(l)}(t)  \big)$. 
			Then,
			\begin{equation}\label{eq:hatw-w}
				\hat w_k - w_k = \frac{\langle \hat g-g, \hat \phi_k - \phi_k\rangle}{\hat \lambda_k^{\mathcal A}} + \frac{\langle \hat g-g, \phi_k\rangle}{\hat \lambda_k^{\mathcal A} } + \frac{\langle g, \hat \phi_k - \phi_k\rangle  }{\hat \lambda_k^{\mathcal A}} + \left(\frac{1}{\hat \lambda_k^{\mathcal A}} - \frac{1}{\lambda_k^{\mathcal A}} \right) \langle g, \phi_k\rangle. 
			\end{equation}
			
			We first bound the first term in \eqref{eq:hatw-w}. Under the event $\mathcal E_2$ in \ref{eq:event2}, we have $ \hat \lambda_k^{\mathcal A} \ge \lambda_k^{\mathcal A} /2 $.  By Lemma \ref{lem:g}, $\|\hat g- g\|_2^2 = O_P(N^{-1})$. Then, together with Lemma \ref{lem:eigfun}, we obtain
			\begin{eqnarray*}
				\sum_{k=1}^m \frac{\langle \hat g-g, \hat \phi_k - \phi_k\rangle^2}{(\hat \lambda_k^{\mathcal A})^2} \le 4 \sum_{k=1}^m (\lambda_k^{\mathcal A})^{-2} \|\hat g-g\|_2^2 \|\hat \phi_k - \phi_k\|_2^2 = O_P(m^{2\alpha+3}N^{-2}).
			\end{eqnarray*}
			By Lemma \ref{lem:g}, we obtain $\expect \langle \hat g-g, \phi_k\rangle^2 = O(k^{-\alpha}N^{-1})$.
			To bound the sum about the second term in \eqref{eq:hatw-w},
			\begin{eqnarray*}
				\sum_{k=1}^m \frac{\langle \hat g-g, \phi_k\rangle^2}{(\hat \lambda_k^{\mathcal A})^2 } \le 4 \sum_{k=1}^m (\lambda_k^{\mathcal A})^{-2}\langle \hat g-g, \phi_k\rangle^2 = O_P(m^{1+\alpha}N^{-1}).
			\end{eqnarray*}
			
			To handle the sum about the last term in \eqref{eq:hatw-w},
			we observe
			\begin{align*}
				\left(\frac{1}{\hat \lambda_k^{\mathcal A}} - \frac{1}{\lambda_k^{\mathcal A}} \right)^2 \langle g, \phi_k\rangle^2 & = (\hat \lambda_k^{\mathcal A} \lambda_k^{\mathcal A})^{-2}(\hat \lambda_k^{\mathcal A} - \lambda_k^{\mathcal A} )^2 \left[\sum_{l=1}^L \pi_l \expect \{ (Y^{(l)} - \expect(Y^{(l)}) ) \xi_k^{(l)} \} \right]^2 \\
				& = (\hat \lambda_k^{\mathcal A} \lambda_k^{\mathcal A})^{-2}(\hat \lambda_k^{\mathcal A} - \lambda_k^{\mathcal A} )^2 \left[\sum_{l=1}^L \pi_l \expect \left\{ \left( \sum_{j=1}^\infty \xi_j^{(l)}w_j^{(l)} + \epsilon^{(l)} \right) \xi_k^{(l)} \right\} \right]^2 \\
				&  = (\hat \lambda_k^{\mathcal A} \lambda_k^{\mathcal A})^{-2}(\hat \lambda_k^{\mathcal A} - \lambda_k^{\mathcal A} )^2 \left( \sum_{l=1}^L \pi_l \lambda_k^{(l)} (b_k - \delta_k^{(l)}) \right)^2  .
			\end{align*}
			By Lemma \ref{lem:expansions}, we obtain
			\[ (1-\| \hat \phi_k - \phi_k \|_2 ) |\hat \lambda_k^{\mathcal A} - \lambda_k^{\mathcal A} | \le \hat \Delta_{kk} + \| \hat \phi_k - \phi_k \|_2 \hat \Delta_k.  \]
			Since $\sup_{k \ge 1} |\hat \lambda_k^{\mathcal A} - \lambda_k^{\mathcal A}| \le \hat \Delta$ and $\sup_{k \ge 1} \hat \Delta_k \le \hat \Delta$, we have
			\[ |\hat \lambda_k^{\mathcal A} - \lambda_k^{\mathcal A} | \le \hat \Delta_{kk} + 2\| \hat \phi_k - \phi_k \|_2 \hat \Delta.  \]
			Consequently, using Assumptions \ref{assump:cov-decay}, \ref{assump:slope-decay}, Lemmas \ref{lem:Delta} and \ref{lem:eigfun} yields
			\begin{align*}
				& \sum_{k=1}^m \left(\frac{1}{\hat \lambda_k^{\mathcal A}} - \frac{1}{\lambda_k^{\mathcal A}} \right)^2 \langle g, \phi_k\rangle^2 \\
				\le & \sum_{k=1}^m 8(\lambda_k^{\mathcal A})^{-2}(\hat \lambda_k^{\mathcal A} - \lambda_k^{\mathcal A} )^2 b_k^2 + \sum_{k=1}^m 8(\lambda_k^{\mathcal A})^{-4}(\hat \lambda_k^{\mathcal A} - \lambda_k^{\mathcal A} )^2 \left( \sum_{l=1}^L \pi_l (\lambda_k^{(l)})^2 \right)\left( \sum_{l=1}^L \pi_l (\delta_k^{(l)})^2 \right)    \\
				\le & c \sum_{k=1}^m k^{2\alpha-2\beta}( \hat \Delta_{kk}^2 + \|\hat \phi_k - \phi_k \|_2^2\hat \Delta^2 ) + c\sum_{l=1}^L \pi_l \sum_{k=1}^m k^{2\alpha} (\delta_k^{(l)})^2  ( \hat \Delta_{kk}^2 + \|\hat \phi_k - \phi_k \|_2^2\hat \Delta^2 ) \\
				= & O_P\left(\frac{1}{N} + \frac{1+\log m + m^{3+2\alpha-2\beta}}{N^2} + \frac{h^2}{N} + \frac{m^{2\alpha+2} h^2}{N^2} \right).
			\end{align*}
			
			Next, we address the sum about the third term in \eqref{eq:hatw-w}.
			Lemma \ref{lem:expansions} implies
			\begin{eqnarray}\label{eq:g_hatpsi_k-psi_k}
				\langle g, \hat \phi_k - \phi_k\rangle &= & \sum_{j: j \neq k} (\lambda_k^{\mathcal A} - \lambda_j^{\mathcal A})^{-1}\langle g, \phi_j\rangle \left( \int(\hat K^{\mathcal A} - K^{\mathcal A})  \phi_k \phi_j \right) +  \nonumber \\
				& & \sum_{j: j \neq k} ( (\hat \lambda_k^{\mathcal A} - \lambda_j^{\mathcal A})^{-1} - (\lambda_k^{\mathcal A} - \lambda_j^{\mathcal A})^{-1} ) \langle g, \phi_j\rangle \left( \int(\hat K^{\mathcal A} - K^{\mathcal A})  \phi_k \phi_j \right) + \nonumber \\
				& & \sum_{j: j \neq k} (\hat \lambda_k^{\mathcal A} - \lambda_j^{\mathcal A})^{-1} \langle g, \phi_j\rangle\left( \int(\hat K^{\mathcal A} - K^{\mathcal A})  (\hat \phi_k - \phi_k) \phi_j \right) + \nonumber \\
				& & \langle g, \phi_k\rangle \int  (\hat \phi_k - \phi_k) \phi_k = T_{k1} + T_{k2} + T_{k3} + T_{k4}.
			\end{eqnarray}
			Then, under the event $\mathcal E_2$ in \eqref{eq:event2},
			\[ \sum_{k=1}^m (\hat \lambda_k^{\mathcal A})^{-2} \langle g, \hat \phi_k - \phi_k\rangle^2 \le c \sum_{k=1}^m (\lambda_k^{\mathcal A})^{-2} (T_{k1}^2 + T_{k2}^2 + T_{k3}^2 + T_{k4}^2).  \]
			Note that 
			\begin{align*} 
				\sum_{k=1}^m (\lambda_k^{\mathcal A})^{-2} T_{k4}^2 & \le \sum_{k=1}^m (\lambda_k^{\mathcal A})^{-2} \left(\sum_{l=1}^L \pi_l \lambda_k^{(l)} (b_k - \delta_k^{(l)})\right)^2   \|\hat\phi_k - \phi_k\|_2^2 \\
				& \le 2\sum_{k=1}^m b_k^2 \|\hat\phi_k - \phi_k\|_2^2 + 2 \sum_{k=1}^m (\lambda_k^{\mathcal A})^{-2} \left(\sum_{l=1}^L \pi_l (\lambda_k^{(l)})^2 \right)\left( \sum_{l=1}^L \pi_l (\delta_k^{(l)})^2  \right) \|\hat\phi_k - \phi_k\|_2^2  \\
				& \le c \sum_{k=1}^m b_k^2 \|\hat\phi_k - \phi_k\|_2^2 + c \sum_{l=1}^L \pi_l \sum_{k=1}^m (\delta_k^{(l)})^2  \|\hat\phi_k - \phi_k\|_2^2  \\
				&= O_P\left(\frac{1}{N} + \frac{m^2h^2}{N}\right).
			\end{align*}
			
			If the event $\mathcal E_1$ in \eqref{eq:event1} holds, then by Lemma \ref{lem:lam_k-lam_j},
			\begin{align*}
				|T_{k3}| \le & c \hat \Delta \|\hat \phi_k - \phi_k\|_2 \sum_{j: j \neq k} |\lambda_k^{\mathcal A} - \lambda_j^{\mathcal A}|^{-1}\left(j^{-(\alpha+\beta)} + j^{-\alpha}\sum_{l=1}^L \pi_l |\delta_j^{(l)}|\right) \\
				\le  & c\hat \Delta \|\hat \phi_k - \phi_k\|_2 \left\{ \sum_{j: j<k/2} \left(j^{-\beta} + \sum_{l=1}^L \pi_l |\delta_j^{(l)}|\right)  + \sum_{j:j>2k} k^\alpha \left(j^{-(\alpha+\beta)} + j^{-\alpha}\sum_{l=1}^L \pi_l |\delta_j^{(l)}|\right) + \right. \\
				& \left. \sum_{j: k/2 \le j \le 2k, j\neq k} k^{\alpha+1}|k-j|^{-1} \left(j^{-(\alpha+\beta)} + j^{-\alpha}\sum_{l=1}^L \pi_l |\delta_j^{(l)}|\right)  \right\}  \\
				\le & c\hat \Delta \|\hat \phi_k - \phi_k\|_2 \left\{ 1 + k^{1-\beta}\log k +\sum_{l=1}^L \pi_l \left(\sum_{j:j<k/2} |\delta_j^{(l)}| + \sum_{j:j>2k} |\delta_j^{(l)}| + k \sum_{j: k/2 \le j \le 2k, j\neq k} |\delta_j^{(l)}|   \right)  \right\}  \\
				\le & c \hat \Delta \|\hat \phi_k - \phi_k\|_2(1+k^{1-\beta}\log k+kh).
			\end{align*}
			Thus,
			\[
			\sum_{k=1}^m (\lambda_k^{\mathcal A})^{-2} T_{k3}^2 = O_P\left(\frac{m^{2\alpha+3}}{N^2} + \frac{m^{2\alpha+5} h^2}{N^2} \right).
			\]
			
			Under the event $\mathcal E_1$ in \eqref{eq:event1} and by Lemma \ref{lem:lam_k-lam_j},
			\begin{align*}
				T_{k2}^2 & \le c\left\{\sum_{j: j \neq k} \frac{|\lambda_k^{\mathcal A} - \hat \lambda_k^{\mathcal A}|}{(\lambda_k^{\mathcal A} - \lambda_j^{\mathcal A})^2} \bigg(j^{-(\alpha+\beta)} + j^{-\alpha}\sum_{l=1}^L \pi_l |\delta_j^{(l)}|\bigg) \left( \int(\hat K^{\mathcal A} - K^{\mathcal A})  \phi_k \phi_j \right) \right\}^2 \\
				& \le c\left\{\sum_{j: j \neq k} \frac{|\lambda_k^{\mathcal A} - \hat \lambda_k^{\mathcal A}|^2}{(\lambda_k^{\mathcal A} - \lambda_j^{\mathcal A})^4} \bigg(j^{-(\alpha+\beta)} + j^{-\alpha}\sum_{l=1}^L \pi_l |\delta_j^{(l)}|\bigg)^2 \right\} \left\{\sum_{j: j \neq k} \left( \int(\hat K^{\mathcal A} - K^{\mathcal A})  \phi_k \phi_j \right)^2\right\} \\
				& \le c|\lambda_k^{\mathcal A} - \hat \lambda_k^{\mathcal A}|^2 \hat \Delta_k^2\left\{ \sum_{j: j \neq k} j^{-2(\alpha+\beta)} (\lambda_k^{\mathcal A} - \lambda_j^{\mathcal A})^{-4} + \sum_{j: j \neq k} j^{-2\alpha}\bigg( \sum_{l=1}^L \pi_l |\delta_j^{(l)}| \bigg)^2(\lambda_k^{\mathcal A}-\lambda_j^{\mathcal A})^{-4} \right\} \\
				& \le c |\lambda_k^{\mathcal A} - \hat \lambda_k^{\mathcal A}|^2 \hat \Delta_k^2 ( 1 + \log k + k^{2\alpha-2\beta+4} + k^{2\alpha+4}h^2 ).
			\end{align*}
			Consequently,
			\begin{align*}
				\sum_{k=1}^m (\lambda_k^{\mathcal A})^{-2} T_{k2}^2  & \le c\sum_{k=1}^m k^{2\alpha} |\lambda_k^{\mathcal A} - \hat \lambda_k^{\mathcal A}|^2 \hat \Delta_k^2 ( 1 + \log k + k^{2\alpha-2\beta+4} + k^{2\alpha+4}h^2 ) \\
				& = O_P\left( N^{-2}\sum_{k=1}^m (k^{2\alpha}\log m + k^{4\alpha-2\beta+4} ) + N^{-1}\sum_{k=1}^m (\hat \Delta_{kk}^2 + \|\hat \phi_k - \phi_k\|_2^2 \hat \Delta^2 ) k^{4\alpha+4}h^2  \right) \\
				& = O_P\left( \frac{m^{4\alpha-2\beta+5}}{N^2} + \frac{m^{1+2\alpha}\log m}{N^2} + \frac{m^{2\alpha+5}h^2}{N^2}+  \frac{m^{4\alpha+7}h^2}{N^3} \right).
			\end{align*}
			
			Applying Lemma \ref{lem:T1} yields 
			\begin{eqnarray*}
				E  \sum_{k=1}^m (\lambda_k^{\mathcal A})^{-2} T_{k1}^2 & = & E \left\{ \sum_{k=1}^m (\lambda_k^{\mathcal A})^{-2} \left( \sum_{j: j \neq k} (\lambda_k^{\mathcal A} - \lambda_j^{\mathcal A})^{-1} \langle g, \phi_j\rangle \left( \int(\hat K^{\mathcal A} - K)  \phi_k \phi_j \right)   \right)^2  \right\} \\
				& = & O\left( \frac{m^{1+\alpha}}{N} + \frac{m^3h^2}{N} + \frac{m^{1+\alpha}h^2}{N} \right).
			\end{eqnarray*}
			Combining all the above pieces completes the proof.
		\end{proof}
		
		\begin{lemma}\label{lem:rate_delta}
			Under Assumptions \ref{assump:xdist}-\ref{assump:x-moment}, if $m^{2(\alpha+1)}N^{-1} \to 0$ and we take $\tau \asymp n^{-1/2}$, then
			\begin{align*}
				\|\bdelta - \hat \bdelta\|_2^2 = O_P\left( \frac{m^{\alpha} h}{n^{1/2}}  \wedge h^2 + m^{1-2\beta} + \frac{m^{1+\alpha}}{N} + \frac{m^3h^2}{N} \right).
			\end{align*}
		\end{lemma}
		
		\begin{proof}
			%	Let $\bepsilon = (\epsilon_1, \dots, \epsilon_n)^{\t}$ and $\bar{\bepsilon}$ is a $n$-dimensional vector with elements $\bar{\epsilon} = n^{-1}\sum_{i=1}^n \epsilon_i$.
			Notice that
			\begin{align*}
				\frac{1}{2n} \|\bY - \bar{\bY} - \hat \Xi \hat \bw - \hat \Xi \hat \bdelta\|_2^2 + \tau \|\hat \bdelta\|_1 \le \frac{1}{2n} \|\bY - \bar{\bY} - \hat \Xi \hat \bw - \hat \Xi \bdelta\|_2^2 + \tau \| \bdelta\|_1.
			\end{align*}
			Thus, under the event $\mathcal E_3$ in \eqref{eq:event3}, 
			\begin{eqnarray}\label{eq:basic_ineq}
				& & \frac{1}{2n}\|\hat{\Xi} (\bdelta - \hat{\bdelta}) \|_2^2 \nonumber\\
				& \le & \tau \|\bdelta\|_1 - \tau \|\hat{\bdelta}\|_1 + \frac{1}{n} \langle \bY - \bar{\bY} - \hat{\Xi} \hat{\bw} - \hat{\Xi}\bdelta, \hat{\Xi}(\hat{ \bdelta}-\bdelta) \rangle \nonumber\\
				& = & \tau \|\bdelta\|_1 - \tau \|\hat{\bdelta}\|_1 + \frac{1}{n}|\langle \bepsilon - \bar{\bepsilon} + \tilde{\Xi}_\circ \tilde{\bb} + (\Xi_\circ - \hat{\Xi})\bb + \hat{\Xi}(\bw-\hat{\bw}), \hat{\Xi}(\hat{\bdelta}-\bdelta)  \rangle| \nonumber\\
				& \le & \tau \|\bdelta\|_1 - \tau \|\hat{\bdelta}\|_1 + \frac{1}{n}\|\hat \Xi^\t (\bepsilon - \bar{\bepsilon})\|_\infty \|\bdelta - \hat \bdelta\|_1 + \frac{1}{n}|\langle \tilde{\Xi}_\circ \tilde{\bb} + (\Xi_\circ - \hat{\Xi})\bb + \hat{\Xi}(\bw-\hat{\bw}), \hat{\Xi}(\hat{\bdelta}-\bdelta)  \rangle|  \nonumber\\
				%		& \le & \frac{3\tau}{2} \|\bdelta\|_1 - \frac{\tau}{2} \|\hat{\bdelta}\|_1 + \frac{1}{4n} \|\hat{\Xi} (\bdelta - \hat{\bdelta}) \|_2^2 + \frac{1}{n}\| \tilde{\Xi}_\circ \tilde{\bb} + (\Xi_\circ - \hat{\Xi})\bb + \hat{\Xi}(\bw-\hat{\bw}) \|_2^2, 
				& \le & 2\tau \|\bdelta\|_1 - \frac{\tau}{2} \|\hat{\bdelta} - \bdelta\|_1 + \frac{1}{4n} \|\hat{\Xi} (\bdelta - \hat{\bdelta}) \|_2^2 + \frac{1}{n}\| \tilde{\Xi}_\circ \tilde{\bb} + (\Xi_\circ - \hat{\Xi})\bb + \hat{\Xi}(\bw-\hat{\bw}) \|_2^2,
			\end{eqnarray}
			where $\bb = (b_{1}, \dots, b_{m})^{\t}$, $\tilde{\bb} = (b_{(m+1)}, \dots)^{\t}$, $\xi_{i,k} = \int_\tdomain \big(X_i(t) - \mu(t)\big) \phi_k(t)dt $, $\bar \xi_k = n^{-1}\sum_{i=1}^n \xi_{i,k}$,
			\[ \Xi_\circ = \left( \begin{array}{ccc}
				\xi_{1,1}-\bar \xi_1 & \cdots & \xi_{1,m} - \bar \xi_m \\
				\xi_{2,1} -\bar \xi_1 & \cdots & \xi_{2,m} - \bar \xi_m \\
				\vdots & \ddots & \vdots \\
				\xi_{n,1} - \bar \xi_1 & \cdots & \xi_{n,m} - \bar \xi_m
			\end{array} \right), ~~~~  
			\tilde{\Xi}_\circ = \left( \begin{array}{ccc}
				\xi_{1,m+1} - \bar \xi_{m+1} & \cdots  \\
				\xi_{2,m+1} - \bar \xi_{m+1} & \cdots  \\
				\vdots & \ddots \\
				\xi_{n,m+1}-\bar \xi_{m+1} & \cdots 
			\end{array} \right). \]
			This implies,
			\[ \frac{1}{4n}\|\hat{\Xi} (\bdelta - \hat{\bdelta}) \|_2^2 \le 2\tau \|\bdelta\|_1 - \frac{\tau}{2} \|\hat{\bdelta} - \bdelta\|_1 + \frac{3}{n} \| \tilde{\Xi}_\circ \tilde{\bb}\|_2^2 + \frac{3}{n}\|(\Xi_\circ - \hat{\Xi})\bb\|_2^2 + \frac{3}{n} \|\hat{\Xi}(\bw-\hat{\bw}) \|_2^2.   \]
			By Assumptions \ref{assump:cov-decay} and \ref{assump:slope-decay}, we have
			\[ \frac{1}{n} \expect \|\tilde \Xi_\circ \tilde \bb\|_2^2 = O(m^{1-\alpha-2\beta}). \]
			Note that 
			\begin{eqnarray}\label{eq:error1}
				\frac{1}{n} \|(\Xi_\circ - \hat{\Xi})\bb\|_2^2 & = & \frac{1}{n} \sum_{i=1}^n \left\{ \sum_{k=1}^m  b_k \int (X_i - \bar{X}) (\phi_k - \hat \phi_k) \right\}^2 \nonumber\\
				& \le & \frac{1}{n} \sum_{i=1}^n m \sum_{k=1}^m b_k^2 \|X_i - \bar{X}\|_2^2 \| \phi_k - \hat \phi_k\|_2^2 = O_P\left(\frac{m}{N} \right).
			\end{eqnarray}
			
			Moreover,
			\begin{eqnarray*}
				& & \frac{1}{n} \|\hat{\Xi}(\bw-\hat{\bw}) \|_2^2  =  \frac{1}{n} \sum_{i=1}^n \left\{ \sum_{k=1}^m  (w_k - \hat w_k) \int (X_i - \bar{X})\hat \phi_k \right\}^2 \\
				& \le & \frac{3}{n} \sum_{i=1}^n \left\{ \sum_{k=1}^m (w_k - \hat w_k) \int (X_i - \mu) \phi_k \right\}^2 + \frac{3}{n}\sum_{i=1}^n \left\{ \sum_{k=1}^m (w_k - \hat w_k) \int (X_i - \mu) (\hat \phi_k - \phi_k) \right\}^2 + \\
				& & \frac{3}{n} \sum_{i=1}^n \left\{ \sum_{k=1}^m (w_k - \hat w_k) \int ( \mu - \bar{X}) \hat \phi_k \right\}^2.
			\end{eqnarray*}
			It is obvious that the last term is negligible. 
			Using the independence between $\hat \bw$ and $X_i$ yields 
			\begin{align*}
				\frac{1}{n}\sum_{i=1}^n E \left[\left\{ \sum_{k=1}^m (w_k - \hat w_k) \int (X_i - \mu) \phi_k \right\}^2 \bigg| \hat \bw \right] &=  \frac{1}{n}\sum_{i=1}^n \sum_{k=1}^m \lambda_k (w_k - \hat w_k)^2 \\
				&= O_P\left(\sum_{k=1}^m \lambda_k (w_k - \hat w_k)^2 \right). 
			\end{align*}
			Applying similar arguments in the proof of Lemma \ref{lem:rate_w}, we obtain
			\[ \sum_{k=1}^m \lambda_k (w_k - \hat w_k)^2 = O_P\bigg(\frac{m + m^{3-\alpha}h^2}{N} \bigg). \] 
			And by Lemma \ref{lem:rate_w} and Lemma \ref{lem:eigfun},
			\begin{align*}
				\frac{1}{n}\sum_{i=1}^n \left\{ \sum_{k=1}^m (w_k - \hat w_k) \int (X_i - \mu) (\hat \phi_k - \phi_k) \right\}^2 & = O_P\left( \sum_{k=1}^m (w_k - \hat w_k)^2\sum_{k=1}^m \|\hat \phi_k - \phi_k\|_2^2 \right) \\
				& = O_P\bigg(\frac{m^{\alpha+4}}{N^2} + \frac{m^6 h^2}{N^2} \bigg) = o_P\bigg(\frac{m + m^{3-\alpha}h^2}{N} \bigg),
			\end{align*}
			due to $m^{3+\alpha}N^{-1} = o(1)$. Thus,
			\[ \frac{1}{n} \|\hat{\Xi}(\bw-\hat{\bw}) \|_2^2  = O_P\bigg(\frac{m + m^{3-\alpha}h^2}{N} \bigg). \]
			
			%	Note that from Proposition \ref{prop:rec} with probability at least $1-c_4\exp(-c_5n)$,
			%	\begin{align*}
				%	\frac{1}{n}\|\hat{\Xi} (\bdelta - \hat{\bdelta}) \|_2^2 & \ge \frac{1}{2n}\|\Xi_\circ(\hat \bdelta - \bdelta)\|_2^2 - \frac{1}{n}\|(\Xi_\circ-\hat \Xi)(\hat \bdelta - \bdelta) \|_2^2 \\
				%	& \ge \frac{1}{16} \|D(\hat \bdelta -\bdelta)\|_2^2 - \frac{c_K\|D\|_F}{2n^{1/2}}\|\hat \bdelta - \bdelta\|_1\|D(\hat \bdelta -\bdelta)\|_2 - \frac{1}{n}\|(\Xi_\circ-\hat \Xi)(\hat \bdelta - \bdelta) \|_2^2.   
				%	\end{align*}
			%	Applying the Cauchy-Schwarz inequality leads to
			%	\begin{eqnarray*}
				%		& & \frac{1}{n}\|(\Xi_\circ-\hat \Xi)(\hat \bdelta - \bdelta) \|_2^2 \\
				%		& = & \frac{1}{n} \sum_{i=1}^n \left[ \sum_{k=1}^m \left\{ \int_\tdomain \big(X_i(t)-\bar X(t)\big) (\phi_k(t) - \hat \phi_k(t))dt \right\} (\hat \delta_k - \delta_k) \right]^2 \\
				%		& \le & \|\hat \bdelta - \bdelta\|^2 \frac{1}{n}\sum_{i=1}^n\sum_{k=1}^m \left\{ \int_\tdomain \big(X_i(t)-\bar X(t)\big) (\phi_k(t) - \hat \phi_k(t))dt  \right\}^2  \\
				%		& \le & c\|\hat \bdelta - \bdelta\|^2 \frac{1}{n}\sum_{i=1}^n \big(\|X_i - \bar X\|^2 \sum_{k=1}^m \|\hat \phi_k - \phi_k\|^2 \big) \\
				%		& = & O_P\left(\frac{m^3}{N} \|\bdelta-\hat \bdelta\|_2^2  \right).
				%	\end{eqnarray*}
			
			Note that from Proposition \ref{prop:re} with probability at least $1-c_4\exp(-c_5n)$,
			\begin{align*}
				\frac{1}{n}\|\hat{\Xi} (\bdelta - \hat{\bdelta}) \|_2^2 & \ge \frac{1}{2n}\|\Xi(\hat \bdelta - \bdelta)\|_2^2 - \frac{1}{n}\|(\Xi-\hat \Xi)(\hat \bdelta - \bdelta) \|_2^2 \\
				& \ge \frac{1}{8} \|D(\hat \bdelta -\bdelta)\|_2^2 - \frac{c_K\|D\|_F}{2n^{1/2}}\|\hat \bdelta - \bdelta\|_1\|D(\hat \bdelta -\bdelta)\|_2 - \frac{1}{n}\|(\Xi-\hat \Xi)(\hat \bdelta - \bdelta) \|_2^2.   
			\end{align*}
			Applying the Cauchy-Schwarz inequality leads to
			\begin{eqnarray*}
				& & \frac{1}{n}\|(\Xi-\hat \Xi)(\hat \bdelta - \bdelta) \|_2^2 \\
				& = & \frac{1}{n} \sum_{i=1}^n \left[ \sum_{k=1}^m \left\{ \int_\tdomain \big(X_i(t)-\mu(t)\big) (\phi_k(t) - \hat \phi_k(t))dt +  \int_\tdomain \big(\bar X(t) - \mu(t)\big) \hat \phi_k(t)dt \right\} (\hat \delta_k - \delta_k) \right]^2 \\
				& \le & \|\hat \bdelta - \bdelta\|_2^2 \frac{1}{n}\sum_{i=1}^n\sum_{k=1}^m \left\{ \int_\tdomain \big(X_i(t)-\mu(t)\big) (\phi_k(t) - \hat \phi_k(t))dt +  \int_\tdomain \big(\bar X(t) - \mu(t)\big) \hat \phi_k(t)dt \right\}^2  \\
				& \le & c\|\hat \bdelta - \bdelta\|_2^2 \frac{1}{n}\sum_{i=1}^n \big(\|X_i - \mu\|_2^2 \sum_{k=1}^m \|\hat \phi_k - \phi_k\|_2^2 + \|\bar X-\mu \|_2^2 \big) \\
				& = & O_P\left(\frac{m^3}{N} \|\bdelta-\hat \bdelta\|_2^2 + \frac{1}{n} \|\hat \bdelta - \bdelta\|_2^2 \right).
			\end{eqnarray*}
			
			\begin{itemize}
				\item Suppose $3n^{-1} \| \tilde{\Xi} \tilde{\bb}\|_2^2 + 3n^{-1}\|(\Xi - \hat{\Xi})\bb\|_2^2 + 3n^{-1} \|\hat{\Xi}(\bw-\hat{\bw}) \|_2^2 \le 2\tau\|\bdelta\|_1 $. Then,
				\[ \frac{1}{4n}\|\hat{\Xi} (\bdelta - \hat{\bdelta}) \|_2^2 \le 4\tau\|\bdelta\|_1, ~~~ \|\hat \bdelta - \bdelta\|_1 \le 8 \|\bdelta\|_1. \] 
				Consequently,
				\begin{align*}
					\frac{1}{32} \|D(\hat \bdelta -\bdelta)\|_2^2  & \le  4\tau\|\bdelta\|_1 + \frac{c_K\|D\|_F}{8n^{1/2}}\|\hat \bdelta - \bdelta\|_1\|D(\hat \bdelta -\bdelta)\|_2 + O_P\left(\frac{m^3}{N} \|\bdelta-\hat \bdelta\|_2^2 + \frac{1}{n} \|\hat \bdelta - \bdelta\|_2^2  \right) \\
					& \le 4\tau h +  \frac{c_6 h}{n^{1/2}} \|D(\hat \bdelta -\bdelta)\|_2 + O_P\left(\frac{m^3}{N} \|\bdelta-\hat \bdelta\|_2^2+ \frac{1}{n} \|\hat \bdelta - \bdelta\|_2^2  \right).
				\end{align*}
				Moreover, since $m^{3+\alpha}N^{-1} = o(1)$, direct calculations yield
				\[ \|\bdelta - \hat \bdelta\|_2^2 = O_P\left( \frac{m^{\alpha} h}{n^{1/2}} \wedge h^2 \right). \]
				
				\item Suppose $ 2\tau\|\bdelta\|_1 \le 3n^{-1} \| \tilde{\Xi}_\circ \tilde{\bb}\|_2^2 + 3n^{-1}\|(\Xi_\circ - \hat{\Xi})\bb\|_2^2 + 3n^{-1} \|\hat{\Xi}(\bw-\hat{\bw}) \|_2^2  $. 
				%		By applying the same reasoning used in \eqref{eq:basic_ineq} and $\|\hat \bdelta\|_1 \ge \|\bdelta - \hat \bdelta\|_1 - \|\bdelta\|_1$, we obtain 
				%		\begin{eqnarray*} 
					%			& & \frac{1}{2n}\|\hat{\Xi} (\bdelta - \hat{\bdelta}) \|_2^2 \\
					%			& \le & \tau \|\bdelta\|_1 - \tau \|\hat{\bdelta}\|_1 + \frac{1}{n}\|\hat \Xi^\t (\bepsilon - \bar{\bepsilon})\|_\infty \|\bdelta - \hat \bdelta\|_1 + \frac{1}{n}|\langle \tilde{\Xi}_\circ \tilde{\bb} + (\Xi_\circ - \hat{\Xi})\bb + \hat{\Xi}(\bw-\hat{\bw}), \hat{\Xi}(\hat{\bdelta}-\bdelta)  \rangle|  \\
					%			& \le & 2\tau \|\bdelta\|_1 - \frac{\tau}{2} \|\hat{\bdelta} - \bdelta\|_1 + \frac{1}{4n} \|\hat{\Xi} (\bdelta - \hat{\bdelta}) \|_2^2 + \frac{1}{n}\| \tilde{\Xi}_\circ \tilde{\bb} + (\Xi_\circ - \hat{\Xi})\bb + \hat{\Xi}(\bw-\hat{\bw}) \|_2^2.
					%		\end{eqnarray*}
				By \eqref{eq:basic_ineq}, we have
				\[ \frac{1}{4n}\|\hat{\Xi} (\bdelta - \hat{\bdelta}) \|_2^2  \le  \frac{6}{n} \| \tilde{\Xi}_\circ \tilde{\bb}\|_2^2 + \frac{6}{n}\|(\Xi_\circ - \hat{\Xi})\bb\|_2^2 + \frac{6}{n} \|\hat{\Xi}(\bw-\hat{\bw}) \|_2^2 - \frac{\tau}{2} \|\hat{\bdelta} - \bdelta\|_1 .     \]
				Thus,
				\[ \|\bdelta - \hat \bdelta\|_1^2 = O_P\left(nm^{2(1-\alpha-2\beta)} + \frac{nm^2}{N^2} + \frac{nm^{2(3-\alpha)}h^4}{N^2} \right). \]
				Moreover,
				\begin{align*}
					\frac{1}{32} \|D(\hat \bdelta -\bdelta)\|_2^2  & \le  \frac{c_K\|D\|_F}{8n^{1/2}}\|\hat \bdelta - \bdelta\|_1\|D(\hat \bdelta -\bdelta)\|_2 + O_P\left(\frac{m^3}{N} \|\bdelta-\hat \bdelta\|_2^2 + \frac{1}{n} \|\hat \bdelta - \bdelta\|_2^2 \right) + \\
					& O_P(n^{-1} \| \tilde{\Xi}_\circ \tilde{\bb}\|_2^2 + n^{-1}\|(\Xi_\circ - \hat{\Xi})\bb\|_2^2 + n^{-1} \|\hat{\Xi}(\bw-\hat{\bw}) \|_2^2) .
				\end{align*}
				We obtain
				\[ \|\bdelta - \hat \bdelta\|_2^2 = O_P\left\{\bigg(\frac{m^{1+\alpha}}{N} + m^{1-2\beta} + \frac{m^3 h^2}{N} \bigg) \wedge \bigg( nm^{2(1-\alpha-2\beta)} + \frac{nm^2}{N^2} + \frac{nm^{2(3-\alpha)}h^4}{N^2} \bigg) \right\}. \]
			\end{itemize}
		\end{proof}
		
		\begin{lemma}\label{lem:Fnorm}
			Recall that $D$ is a diagonal matrix with elements $\lambda_1^{1/2}, \dots, \lambda_m^{1/2}$. Write $\hat \Xi = \breve \Xi D$. Then 
			\[ \| n^{-1}\breve \Xi^\t \breve \Xi - I_m \|_F^2 = O_P\left( \frac{m^2}{n} + \frac{m^{3+\alpha}}{N} + \frac{m^{\alpha+4}}{N n} \right). \]
		\end{lemma}
		
		\begin{proof}
			Note that 
			\begin{align}\label{eq:Fnorm-decom}
				& \quad \frac{1}{n}\sum_{i=1}^n \frac{\int (X_i - \bar X) \hat \phi_j \int(X_i - \bar X) \hat \phi_k}{\lambda_j^{1/2}\lambda_k^{1/2}} \nonumber\\
				& = \frac{1}{n}\sum_{i=1}^n \frac{\int (X_i - \mu + \mu - \bar X) (\hat \phi_j - \phi_j + \phi_j) \int(X_i - \mu + \mu - \bar X) (\hat \phi_k-\phi_k + \phi_k)}{\lambda_j^{1/2}\lambda_k^{1/2}} \nonumber\\
				& =  \frac{1}{n}\sum_{i=1}^n \frac{\int (X_i - \mu) \phi_j \int(X_i - \mu) \phi_k}{\lambda_j^{1/2}\lambda_k^{1/2}} + \frac{1}{n}\sum_{i=1}^n \frac{\int (X_i - \mu)\phi_j \int(X_i - \mu) (\hat \phi_k - \phi_k)}{\lambda_j^{1/2}\lambda_k^{1/2}} + \nonumber\\
				& \quad \frac{1}{n}\sum_{i=1}^n \frac{\int (X_i - \mu) \phi_j \int(\mu - \bar X) \phi_k}{\lambda_j^{1/2}\lambda_k^{1/2}} + \frac{1}{n}\sum_{i=1}^n \frac{\int (X_i - \mu) \phi_j \int(\mu - \bar X) (\hat \phi_k - \phi_k)}{\lambda_j^{1/2}\lambda_k^{1/2}} +  \nonumber\\
				& \quad \frac{1}{n}\sum_{i=1}^n \frac{\int (X_i - \mu) (\hat \phi_j - \phi_j) \int(X_i - \mu) \phi_k}{\lambda_j^{1/2}\lambda_k^{1/2}}  + \frac{1}{n}\sum_{i=1}^n \frac{\int (X_i - \mu) (\hat \phi_j - \phi_j) \int(X_i - \mu) (\hat \phi_k - \phi_k)}{\lambda_j^{1/2}\lambda_k^{1/2}} + \nonumber\\
				& \quad  \frac{1}{n}\sum_{i=1}^n \frac{\int (X_i - \mu) (\hat \phi_j - \phi_j) \int(\mu - \bar X) \phi_k}{\lambda_j^{1/2}\lambda_k^{1/2}} + \frac{1}{n}\sum_{i=1}^n \frac{\int (X_i - \mu) (\hat \phi_j - \phi_j) \int(  \mu - \bar X) (\hat \phi_k - \phi_k)}{\lambda_j^{1/2}\lambda_k^{1/2}}  + ...,
			\end{align}
			where we omit some terms for simplicity. By independence between $X_i, i=1,\dots,n$,
			\[	\expect \left\{   \frac{1}{n}\sum_{i=1}^n \frac{\int (X_i - \mu) \phi_j \int(X_i - \mu) \phi_k}{\lambda_j^{1/2}\lambda_k^{1/2}} - \delta_{jk} \right\}^2 = O(n^{-1}), \]
			where $\delta_{jk} = 1$ if $j=k$ and 0 otherwise. Observe that
			\begin{align}\label{eq:Fnorm_ineq1}
				& \quad \left\{   \frac{1}{n}\sum_{i=1}^n \frac{\int (X_i - \mu)\phi_j \int(X_i - \mu) (\hat \phi_k - \phi_k)}{\lambda_j^{1/2}\lambda_k^{1/2}} \right\}^2 \nonumber\\
				& = \frac{1}{n^2}\sum_{i=1}^n \frac{\xi_{i,j}^2 \{ \int(X_i - \mu) (\hat \phi_k - \phi_k) \}^2  }{\lambda_j\lambda_k} + \frac{1}{n^2}\sum_{i \neq i'}\frac{ \xi_{i,j}\xi_{i',j}\int(X_i - \mu) (\hat \phi_k - \phi_k)\int(X_{i'} - \mu) (\hat \phi_k - \phi_k) }{\lambda_j \lambda_k}.
			\end{align}
			To control the second term, by independence between $\hat \phi_k$ and $\xi_{i,j}, X_i, i=1,\dots,n$, we have 
			\begin{align*}
				& \quad \expect\left\{\xi_{i,j}\xi_{i',j}\int(X_i - \mu) (\hat \phi_k - \phi_k)\int(X_{i'} - \mu) (\hat \phi_k - \phi_k) \mid \hat \phi_k \right\} \\
				& = \left[\expect \left\{\xi_{i,j}\int(X_i - \mu) (\hat \phi_k - \phi_k) \mid \hat \phi_k \right\} \right]^2 \\
				& = \left[ \expect \left\{ \xi_{i,j} \sum_{l=1}^\infty \xi_{i,l}\int \phi_l(\hat \phi_k - \phi_k)  \mid \hat \phi_k \right\} \right]^2 \\
				& \le  \lambda_j^2 \|\hat \phi_k - \phi_k\|_2^2.
			\end{align*}
			Thus, the second term in \eqref{eq:Fnorm_ineq1} is bounded by $O_P(j^{-\alpha}k^{2+\alpha}N^{-1})$. Now, we address the first term in \eqref{eq:Fnorm_ineq1}. 	By the Cauchy-Schwarz inequality,
			\begin{align*}
				& \quad \frac{1}{n^2}\sum_{i=1}^n \frac{\xi_{i,j}^2 \{ \int(X_i - \mu) (\hat \phi_k - \phi_k) \}^2  }{\lambda_j\lambda_k} \\
				& \le \frac{1}{n^2} \left(\sum_{i=1}^n \frac{\xi_{i,j}^4}{\lambda_j^2}\right)^{1/2} \left(\sum_{i=1}^n \frac{\{ \int(X_i - \mu) (\hat \phi_k - \phi_k) \}^4}{\lambda_k^2}  \right)^{1/2} \\
				& \le \frac{1}{n^2} \left(\sum_{i=1}^n \frac{\xi_{i,j}^4}{\lambda_j^2}\right)^{1/2} \left(\sum_{i=1}^n \frac{\|X_i - \mu\|_2^4}{\lambda_k^2}  \right)^{1/2} \|\hat \phi_k - \phi_k\|_2^2 \\
				& = O_P\left( \frac{k^{2+\alpha}}{N n} \right).
			\end{align*}
			In a similar manner, we can control the other terms in \eqref{eq:Fnorm-decom}.
			Combining the above pieces together yields, 
			\[ \| n^{-1}\breve \Xi^\t \breve \Xi - I_m\|_F^2 = O_P\left(\frac{m^2}{n} + \frac{m^{3+\alpha}}{N} + \frac{m^{4+\alpha}}{N n}\right).  \]
		\end{proof}
		
		\begin{lemma}\label{lem:w/o_lasso}
			Under Assumptions \ref{assump:xdist}-\ref{assump:x-moment}, if $\tau = 0$, $m^2n^{-1} = o(1)$ and $m^{\alpha+3}N^{-1}=o(1)$, then
			\[ \|\bb - \hat \bb\|_2^2 = O_P\bigg(  m^{1-2\beta} + \frac{m^{1+\alpha}}{n}  + \frac{m^{1+\alpha}}{N} \bigg).  \]
		\end{lemma}
		
		\begin{proof}
			If $\tau=0$, then
			\[ \hat \bdelta = \mathop{\arg\min}_{\bdelta} \frac{1}{2n}\| \bY - \bar \bY - \hat \Xi (\hat \bw + \bdelta) \|_2^2. \]
			Then by the optimality condition, we obtain
			\[ \frac{1}{n}\hat \Xi^\t \hat \Xi (\hat \bw + \hat \bdelta ) = \frac{1}{n}\hat \Xi^\t (\bY - \bar \bY). \]
			Recall that $D$ is a diagonal matrix with elements $\lambda_1^{1/2}, \dots, \lambda_m^{1/2}$. Let $\Xi_\circ$ and $\tilde \Xi_\circ$ be as defined in the proof of Lemma \ref{lem:rate_delta}. Write $\hat \Xi = \breve \Xi D$ and $\Xi_\circ = \breve \Xi_\circ D$. A direct calculation yields
			\[ D(\hat \bb - \bb) = (\breve \Xi^\t \breve \Xi)^{-1} \breve \Xi^\t \big( \bepsilon - \bar{\bepsilon} + \tilde \Xi_\circ \tilde \bb + (\breve \Xi_\circ - \breve \Xi)D\bb \big).   \]
			Consequently, 
			\begin{align*}
				& \quad \|\hat \bb - \bb\|_2^2  \le m^\alpha \|D(\hat \bb - \bb)\|_2^2 \\
				&  \le 3m^\alpha \bigg( \| (\breve \Xi^\t \breve \Xi)^{-1} \breve \Xi^\t ( \bepsilon - \bar{\bepsilon})\|_2^2 + \| (\breve \Xi^\t \breve \Xi)^{-1} \breve \Xi^\t\tilde \Xi_\circ \tilde \bb   \|_2^2 + \| (\breve \Xi^\t \breve \Xi)^{-1}\breve \Xi^\t(\breve \Xi_\circ - \breve \Xi)D\bb  \|_2^2 \bigg). 
			\end{align*}
			It remains to bound the three terms, respectively.
			Due to the independence between $\bepsilon$ and $\breve \Xi$, we have
			\begin{align*}
				\expect ( \| (\breve \Xi^\t \breve \Xi)^{-1} \breve \Xi^\t ( \bepsilon - \bar{\bepsilon})\|_2^2 \mid \breve \Xi) &= \expect ( ( \bepsilon - \bar{\bepsilon})^\t \breve \Xi (\breve \Xi^\t \breve \Xi)^{-2} \breve \Xi^\t ( \bepsilon - \bar{\bepsilon}) \mid \breve \Xi) \\
				& = \expect \big\{ \trace\big( (\breve \Xi^\t \breve \Xi)^{-2} \breve \Xi^\t ( \bepsilon - \bar{\bepsilon}) ( \bepsilon - \bar{\bepsilon})^\t \breve \Xi  \big)  \mid \breve \Xi \big\} \\
				& = \trace\big( (\breve \Xi^\t \breve \Xi)^{-2} \breve \Xi^\t \expect\{( \bepsilon - \bar{\bepsilon}) ( \bepsilon - \bar{\bepsilon})^\t\} \breve \Xi  \big)  \\
				& = \frac{\sigma^2}{n}\trace((n^{-1}\breve \Xi^\t \breve \Xi)^{-1})  = O_P(mn^{-1}).
			\end{align*}
			due to Lemma \ref{lem:Fnorm} and $\expect\{ (\bepsilon - \bar \bepsilon)(\bepsilon - \bar \bepsilon)^\t \} = \sigma^2 I_n - n^{-1}\sigma^2 1_n 1_n^\t$, where $I_n$ is a $n \times n$ identity matrix and $1_n$ is a $n$-dimensional vector with entries all equal to 1.
			
			By Lemma \ref{lem:Fnorm},  we obtain
			\[ \| (\breve \Xi^\t \breve \Xi)^{-1} \breve \Xi^\t\tilde \Xi_\circ \tilde \bb   \|_2^2 = \| (n^{-1}\breve \Xi^\t \breve \Xi)^{-1} n^{-1}\breve \Xi^\t\tilde \Xi_\circ \tilde \bb   \|_2^2 = O_P(n^{-1}\|\tilde \Xi_\circ \tilde \bb\|_2^2) = O_P(m^{1-\alpha-2\beta}). \]
			Lemma \ref{lem:Fnorm} also implies that
			\[  \| (\breve \Xi^\t \breve \Xi)^{-1}\breve \Xi^\t(\breve \Xi_\circ - \breve \Xi)D\bb  \|_2^2 = O_P(n^{-1}\|(\Xi_\circ - \hat \Xi)\bb\|_2^2) = O_P(m N^{-1}), \]
			due to \eqref{eq:error1} in Lemma \ref{lem:rate_delta}.
			
			Thus, we conclude that 
			\[ \|\hat \bb - \bb\|_2^2 = O_P\bigg(  m^{1-2\beta} + \frac{m^{1+\alpha}}{n} + \frac{m^{1+\alpha}}{N} \bigg) . \]
		\end{proof}

		\section{Proof of Proposition \ref{prop:re} and Corollary \ref{cor:rate}}
		
		\begin{proof}[Proof of Proposition \ref{prop:re}]
			Let $\bgamma = D^{-1}\bxi$ and $\Gamma = \Xi D^{-1}$. According to the sub-Gaussian property of $\bxi$ in Assumption \ref{assump:xdist}, we conclude that $\sup_{\bv} \|\bv^{\t} \bgamma\|_{\psi_2} \le K\|\bv\|_2$.
			
			It suffices to prove the statement for $\|D\bv\|_2 = 1$. Indeed, if $\|D\bv\|_2 = 0$, the claim holds trivially. Otherwise, we may consider $\check \bv = \bv/\|D\bv\|_2$ and prove
			\[ \frac{1}{n}\|\Xi \check \bv \|_2^2 \ge \frac{1}{4} - \frac{c_K\|D\|_F}{n^{1/2}}\|\check \bv\|_1. \]
			If the above inequality holds for $\check \bv$, then \eqref{eq:re} holds for $\bv$ by scale invariance.	
			Further, it suffices to prove that
			\[ n^{-1}\sum_{i=1}^n \{(\bv^{\t}D \bgamma_i)^2\mathds{1}(|\bv^{\t}D \bgamma_i| \le T)\} \ge \frac{1}{4} - \frac{c_K\|D\|_F}{n^{1/2}}\|\bv\|_1, ~~~~ \|D\bv\|_2=1,  \]
			where $T>0$ is to be determined later.
			
			Define $\varsigma(\bv)=|n^{-1}\sum_{i=1}^n \{(\bv^{\t}D \bgamma_i)^2\mathds{1}(|\bv^{\t}D \bgamma_i| \le T)\} - \expect\{(\bv^{\t}D\bgamma)^2\mathds{1}(|\bv^{\t}D \bgamma| \le T)\}$ and $Z(r) = \sup_{\bv \in V(r)} \varsigma(\bv) $, where $V(r) = \{ \bv: \|D\bv\|_2=1, \|\bv\|_1 \le r  \}$. 
			We outline the main steps of the proof.
			\begin{enumerate}
				\item[(1)] Given a fixed radius $r$, prove
				\begin{equation}\label{eq:Z}
					\prob \left(Z(r) \ge \frac{1}{8} + \frac{16rT\|D\|_F}{n^{1/2}}\right) \le 2\exp\bigg(-\frac{n(1/8+16rT\|D\|_Fn^{-1/2})^2}{2T^4}\bigg) .
				\end{equation}
				
				\item[(2)] We use a peeling argument to verify the claim holds uniformly over all possible choices of $r$ with high probability.
				
				\item[(3)] Prove that with the chosen $T$ later and for $\bv$ such that $\|D\bv\|_2=1$,
				\begin{equation}\label{eq:T-mean-lower}
					\expect\{(\bv^{\t}D\bgamma)^2\mathds{1}(|\bv^{\t}D \bgamma| \le T)\} \ge \frac{1}{2}. 
				\end{equation}
			\end{enumerate} 
			
			We first prove \eqref{eq:Z}.
			Applying the Azuma-Hoeffding inequality \citep[Corollary 2.20]{wainwright2019high} yields
			\[ \prob(Z(r)-\expect Z(r) \ge t) \le 2\exp\bigg(-\frac{2nt^2}{T^4}\bigg). \]
			Taking $t = t(r) =  1/8 + 8rT\|D\|_Fn^{-1/2} $, it remains to show that $\expect Z(r) \le 8rT\|D\|_F n^{-1/2}$ to prove \ref{eq:Z}.
			Let $\varepsilon_1, \dots, \varepsilon_n$ be i.i.d. Rademacher variables. By the symmetrization technique and the Ledoux-Talagrand contraction inequality \citep[Proposition 5.28]{wainwright2019high},
			\begin{eqnarray*}
				\expect Z(r) & \le & 2 \expect \sup_{\bv \in V(r)} \bigg| n^{-1}\sum_{i=1}^n \{\varepsilon_i(\bv^{\t}D \bgamma_i)^2\mathds{1}(|\bv^{\t}D \bgamma_i| \le T)\} \bigg| \\
				& \le & 8T \expect \sup_{\bv \in V(r)} \bigg| n^{-1}\sum_{i=1}^n \varepsilon_i\bv^{\t}D \bgamma_i \mathds{1}(|\bv^{\t}D \bgamma_i| \le T) \bigg| \\
				& \le & 8r T \expect  \big \|n^{-1}\sum_{i=1}^n \varepsilon_i \bxi_i \mathds{1}(|\bv^{\t}D \bgamma_i| \le T) \big\|_{\infty} \\
				& \le & 8r T \left( \sum_{k=1}^m  n^{-2}\sum_{i=1}^n \expect \xi_{i,k}^2  \right)^{1/2} = \frac{8rT\|D\|_F}{n^{1/2}}.
			\end{eqnarray*}
			
			Denote the event 
			\[ \mathcal T = \left\{ \exists \bv ~s.t.~ \|D\bv\|_2=1 \text{~and~} \varsigma(\bv) \ge \frac{1}{4} +  \frac{32 T\|D\|_F\|\bv\|_1}{n^{1/2}} \right\}. \]
			Applying the peeling argument in Lemma 3 of \citet{raskutti2010restricted} leads to
			\[ \prob(\mathcal T) \le c_4\exp(-c_5n), \]
			for some positive constants $c_4, c_5>0$.
			
			Now we prove \eqref{eq:T-mean-lower}. Note that for $\bv$ satisfying $\|D\bv\|_2=1$,
			\begin{eqnarray*} 
				\expect\{(\bv^{\t}D\bgamma)^2\mathds{1}(|\bv^{\t}D \bgamma| \le T)\} &= & \expect(\bv^{\t}D\bgamma)^2 - \expect \{(\bv^{\t}D\bgamma)^2\mathds{1}(|\bv^{\t}D \bgamma| > T)\} \\
				& \ge & 1 - \expect \{(\bv^{\t}D\bgamma)^2\mathds{1}(|\bv^{\t}D \bgamma| > T)\}.   
			\end{eqnarray*}
			To obtain the lower bound of the expectation on the left side, it suffices to upper bound the last term. By the sub-Gaussian property and the Cauchy-Schwarz inequality,
			\begin{eqnarray*}
				\expect \{(\bv^{\t}D\bgamma)^2\mathds{1}(|\bv^{\t}D \bgamma| > T)\} & \le &  \{\expect (\bv^{\t}D\bgamma)^4\}^{1/2}\{ \prob(|\bv^{\t}D \bgamma| > T) \}^{1/2} \\
				& \le & C_1 K^2 \exp\bigg(-\frac{T^2}{K^2}\bigg).
			\end{eqnarray*}
			Let $T = K\sqrt{\log(2C_1K^2)}$, then we conclude that $\expect \{(\bv^{\t}D\bgamma)^2\mathds{1}(|\bv^{\t}D \bgamma| \le T)\} \ge 1/2$.
			
			Combining the above arguments completes the proof.
		\end{proof}

		\begin{proof}[Proof of Corollary \ref{cor:rate}]
			From Theorem \ref{thm:upper}, if we take $\tau = 0$, then 
			\[ \|b - \hat b\|_2^2  = O_P\bigg( m^{1-2\beta} + \frac{m^{1+\alpha}}{n} + \frac{m^{1+\alpha}}{N} \bigg). \]
			Taking $m \asymp n^{1/(\alpha+2\beta)}$ yields $\|\hat b - b\|_2^2 = O_P(n^{-(2\beta-1)/(\alpha+2\beta)})$.
			Thus, we conclude that if $h$ is sufficiently large so that 
			\[  \frac{m^\alpha h}{n^{1/2}} \wedge h^2 + \frac{m^{1+\alpha}}{N} + \frac{m^3 h^2}{N} + m^{1-2\beta} \gg n^{-\frac{2\beta-1}{\alpha+2\beta}}, \]
			then we should directly take $\tau=0$.
			To understand how large $h$ needs to be to lead to the above situation, we should examine the rate on the left side.
			 
			Case 1. 
			Suppose $h^2 \lesssim n^{-1/2}m^\alpha h$, that is, $h \lesssim n^{-1/2}m^\alpha$. Then
			\[ \|b - \hat b\|_2^2 = O_P\bigg(h^2 + \frac{m^{1+\alpha}}{N} + \frac{m^3 h^2}{N} + m^{1-2\beta} \bigg) = O_P \bigg(h^2 + \frac{m^{1+\alpha}}{N} + m^{1-2\beta} \bigg), \]
			due to $m^{2(\alpha+1)}N^{-1} \to 0$.
			\begin{itemize}
				
				\item If $h \lesssim N^{-(2\beta-1)/(2\alpha+4\beta)}$, then 
				\[ h^2 \lesssim  N^{-(2\beta-1)/(\alpha+2\beta)} \lesssim  \frac{m^{1+\alpha}}{N} + m^{1-2\beta} . \]
				Taking $m \asymp N^{1/(\alpha+2\beta)}$ yields $\|\hat b - b\|_2^2 = O_P(N^{-(2\beta-1)/(\alpha+2\beta)})$.
				Moreover, $h\lesssim n^{-1/2}m^\alpha$ requires $h \lesssim n^{-1/2}N^{\alpha/(\alpha+2\beta)}$. This condition is indeed satisfied since $ n \lesssim N $ and $n^{-1/2}N^{\alpha/(\alpha+2\beta)} \gtrsim N^{-(2\beta-\alpha)/(2\alpha+4\beta)} \gtrsim N^{-(2\beta-1)/(2\alpha+4\beta)}$.
				
				\item If $h \gtrsim N^{-(2\beta-1)/(2\alpha+4\beta)}$, then taking $  h^{-2/(2\beta-1)} \lesssim m \lesssim (Nh^2)^{1/(\alpha+1)}$ and $m^{2(\alpha+1)}N^{-1} = o(1)$ yields $\|\hat b - b\|_2^2 = O_P(h^2)$. Moreover, $h\lesssim n^{-1/2}m^\alpha$ is satisfied when $h \lesssim n^{-(2\beta-1)/(2(2\alpha+2\beta-1))}$. 
				
%				\item If $h \gtrsim N^{-(2\beta-1)/(2\alpha+4\beta)}$, then taking $ m \asymp h^{-2/(2\beta-1)}$ yields $\|\hat b - b\|_2^2 = O_P(h^2)$, since $m^3N^{-1} \asymp h^{-6/(2\beta-1)} N^{-1} = o(1)$. Moreover, $h\lesssim n^{-1/2}m^\alpha$ requires $h \lesssim n^{-(2\beta-1)/(2(2\alpha+2\beta-1))}$. 
			\end{itemize}
			Case 2. Suppose $h^2 \gtrsim n^{-1/2}m^\alpha h$, that is, $h \gtrsim n^{-1/2}m^\alpha$. Then
			\[ \|b - \hat b\|_2^2 = O_P\bigg(\frac{m^\alpha h}{n^{1/2}} + \frac{m^{1+\alpha}}{N} + \frac{m^3 h^2}{N} + m^{1-2\beta} \bigg) . \]
			\begin{itemize}
				\item If $n^{-1/2}m^\alpha h \lesssim m^{1+\alpha}N^{-1}$, then $h$ should satisfy $h \lesssim mn^{1/2}N^{-1}$. However, $h$ can not satisfy $h \lesssim mn^{1/2}N^{-1}$ and $h \gtrsim n^{-1/2}m^\alpha$ simultaneously because $n \lesssim N$.
				
				\item If $n^{-1/2}m^\alpha h \gtrsim m^{1+\alpha}N^{-1}$ and $n^{-1/2}m^\alpha h \gtrsim m^3h^2N^{-1}$, then $h$ should satisfy $mn^{1/2}N^{-1} \lesssim  h \lesssim  m^{\alpha-3}Nn^{-1/2}$ and $\|b - \hat b\|_2^2 = O_P(n^{-1/2}m^\alpha h + m^{1-2\beta})$. Taking $m \asymp (n^{1/2}h^{-1})^{1/(\alpha+2\beta-1)}$ yields 
				$$\|b - \hat b\|_2^2 = O_P\bigg( n^{-\frac{2\beta-1}{2(\alpha+2\beta-1)}}h^{\frac{2\beta-1}{\alpha+2\beta-1}} \bigg).$$
				The condition $ h \gtrsim mn^{1/2}N^{-1} $ implies that $ h \gtrsim n^{1/2} N^{-(\alpha+2\beta-1)/(\alpha+2\beta)}$. The condition $h \gtrsim n^{-1/2}m^\alpha$ implies $h \gtrsim n^{-(2\beta-1)/(2(2\alpha+2\beta-1))}$. Combining the above pieces together, we require $h \gtrsim n^{-(2\beta-1)/(2(2\alpha+2\beta-1))}$ if $n \lesssim N$. Under this condition of $h$, we have $m^{2(\alpha+1)}N^{-1}=o(1)$. Moreover, the condition $ h \lesssim m^{\alpha-3}Nn^{-1/2}$ requires $ h \lesssim N^{(\alpha+2\beta-1)/(2(\alpha+\beta-2))} n^{-(\beta+1)/(2(\alpha+\beta-2))}$.
				
				\item If $n^{-1/2}m^\alpha h \gtrsim m^{1+\alpha}N^{-1}$ and $n^{-1/2}m^\alpha h \lesssim m^3h^2N^{-1}$, then $h$ should satisfy $h \gtrsim mn^{1/2}N^{-1}$ and $ h \gtrsim  m^{\alpha-3}Nn^{-1/2}$ and $\|b - \hat b\|_2^2 = O_P(m^3 h^2 N^{-1} + m^{1-2\beta})$. Taking $m \asymp (Nh^{-2})^{1/(2\beta+2)}$ yields
				\[ \| b - \hat b\|_2^2 = O_P\bigg(  N^{-\frac{2\beta-1}{2(\beta+1)}} h^{\frac{2\beta-1}{\beta+1}} \bigg). \]
				The condition $h \gtrsim  m^{\alpha-3}Nn^{-1/2}$ implies $h \gtrsim N^{(\alpha+2\beta-1)/(2(\alpha+\beta-2))} n^{-(\beta+1)/(2(\alpha+\beta-2))}$. The condition $h \gtrsim n^{-1/2}m^\alpha$ implies $ h \gtrsim  N^{\alpha/(2(\alpha+\beta+1))} n^{-(\beta+1)/(2(\alpha+2\beta-1))}$. The condition $ h \gtrsim mn^{1/2}N^{-1}$ implies $h \gtrsim N^{-(2\beta+1)/(2\beta+4)} n^{(\beta+1)/(2\beta+4)}$. By calculation, we only require $ h \gtrsim N^{(\alpha+2\beta-1)/(2(\alpha+\beta-2))} n^{-(\beta+1)/(2(\alpha+\beta-2))}$. Under this condition of $h$, we have $m^{2(\alpha+1)}N^{-1}=o(1)$.
			\end{itemize}
			Based on the results above, we also conclude that if $ h \gtrsim  n^{-(2\beta-1)/(2(2\alpha+2\beta-1))}$,
			\[ n^{-\frac{2\beta-1}{2(\alpha+2\beta-1)}}h^{\frac{2\beta-1}{\alpha+2\beta-1}} \gg n^{-\frac{2\beta-1}{\alpha+2\beta}},  \] 
			and if $ h \gtrsim N^{(\alpha+2\beta-1)/(2(\alpha+\beta-2))} n^{-(\beta+1)/(2(\alpha+\beta-2))}$, 
			\[ N^{-\frac{2\beta-1}{2(\beta+1)}} h^{\frac{2\beta-1}{\beta+1}}  \gg n^{-\frac{2\beta-1}{\alpha+2\beta}}. \]
			
			Combining the pieces, we complete the proof.
		\end{proof}
		
		\section{Theoretical Analysis of Prediction Performance}
		
		Let $(X^\star, Y^\star) \sim Q$, where $(X^\star, Y^\star) $ is independent of the training data. We measure the prediction performance as follows,
		\[  \expect_\star \{( Y^\star - \langle X^\star, b\rangle)^2 - (Y^\star - \langle X^\star, \hat b\rangle )^2\} = \expect_\star (\langle X^\star, \hat b - b\rangle^2),  \] 
		where $\langle X^\star, b\rangle = \int_\tdomain X^\star(t) b(t)dt$, and $\expect_\star$ denotes the expectation with respect to the randomness of $(X^\star, Y^\star)$.
		
		\begin{theorem}\label{thm:pre_error}
			Suppose Assumptions \ref{assump:xdist}-\ref{assump:x-moment} hold. 
			If $\tau \asymp n^{-1/2}$, $N^{-1}m^{2(\alpha+1)} = o(1)$ and $h = O(1)$,  then
			\[ \expect_\star (\langle X^\star, \hat b - b\rangle^2) = O_P\bigg(\frac{h}{n^{1/2}} \land h^2 + m^{1-\alpha-2\beta} + \frac{m}{N} + \frac{1+m^{3-\alpha}}{N}h^2 \bigg).  \]
			If $\tau=0$, $n^{-1}m^2=o(1)$ and $N^{-1}m^{\alpha+3}=o(1)$, then 
			\[ \expect_\star (\langle X^\star, \hat b - b\rangle^2) = O_P\bigg(m^{1-\alpha-2\beta} + \frac{m}{n} + \frac{m}{N} \bigg). \]
		\end{theorem}
		
		Analogous to the estimation error in Theorem \ref{thm:upper}, the sparsity parameter $\tau$ and the truncation parameter $m$ play a critical role in determining the convergence rate of the prediction error, as shown in Theorem \ref{thm:pre_error}. The optimal choices of $\tau$ and $m$ under different regimes of the contrast $h$ are detailed in Corollary \ref{cor:pre_error}. Specifically, when $h \ll n^{-(\alpha+2\beta-2)/(2\alpha+4\beta)}$ and $N \gg n$, the convergence rate for the prediction error is faster than that of the classical estimator using only target data.
		
		\begin{corollary}\label{cor:pre_error}
			Suppose Assumptions \ref{assump:xdist}-\ref{assump:x-moment} hold. 
			\begin{enumerate}
				\item Assume $n \lesssim N^{(\alpha+2\beta-1)/(\alpha+2\beta)}$.
				\begin{itemize}
					\item If $h \lesssim N^{-\frac{\alpha+2\beta-1}{2(\alpha+2\beta)}}$,  we take $\tau \asymp n^{-1/2}$ and $ m \asymp N^{1/(\alpha+2\beta)}$, then
					\[ \expect_\star (\langle X^\star, \hat b - b\rangle^2) = O_P\big( N^{-\frac{\alpha+2\beta-1}{\alpha+2\beta}} \big). \]
					\item If $ N^{-\frac{\alpha+2\beta-1}{2(\alpha+2\beta)}} \lesssim  h \lesssim n^{-1/2}$, we take $\tau \asymp n^{-1/2}$ and $ h^{-2/(\alpha+2\beta-1)} \lesssim m \lesssim Nh^2$, $m^{2(\alpha+1)}N^{-1}=o(1)$, then
					\[ \expect_\star (\langle X^\star, \hat b - b\rangle^2) = O_P( h^2 ). \]
					\item If $  n^{-1/2} \lesssim  h \lesssim n^{-\frac{\alpha+2\beta-2}{2(\alpha+2\beta)}}$, we take $\tau \asymp n^{-1/2}$ and $ (n^{-1/2}h)^{-1/(\alpha+2\beta-1)} \lesssim  m \lesssim Nn^{-1/2}h $, $m^{2(\alpha+1)}N^{-1}=o(1)$, then
					\[ \expect_\star (\langle X^\star, \hat b - b\rangle^2) = O_P( n^{-1/2}h ). \]
					\item  If $ h \gtrsim n^{-\frac{\alpha+2\beta-2}{2(\alpha+2\beta)}}$, we take $\tau =0$ and $m \asymp n^{1/(\alpha+2\beta)}$, then
					\[ \expect_\star (\langle X^\star, \hat b - b\rangle^2) = O_P\big( n^{-\frac{\alpha+2\beta-1}{\alpha+2\beta}} \big). \]
				\end{itemize}
				\item Assume $ N^{(\alpha+2\beta-1)/(\alpha+2\beta)} \lesssim n \lesssim N$.
				\begin{itemize}
					\item If $h \lesssim n^{1/2}N^{-\frac{\alpha+2\beta-1}{\alpha+2\beta}}$,  we take $\tau \asymp n^{-1/2}$ and $ m \asymp N^{1/(\alpha+2\beta)}$, then
					\[ \expect_\star (\langle X^\star, \hat b - b\rangle^2) = O_P\big( N^{-\frac{\alpha+2\beta-1}{\alpha+2\beta}} \big). \]
					\item If $ n^{1/2}N^{-\frac{\alpha+2\beta-1}{\alpha+2\beta}} \lesssim  h \lesssim n^{-\frac{\alpha+2\beta-2}{2(\alpha+2\beta)}}$, we take $\tau \asymp n^{-1/2}$ and $ (n^{-1/2}h)^{-1/(\alpha+2\beta-1)} \lesssim  m \lesssim Nn^{-1/2}h $, $m^{2(\alpha+1)}N^{-1}=o(1)$, then
					\[ \expect_\star (\langle X^\star, \hat b - b\rangle^2) = O_P( n^{-1/2}h ). \]
					\item  If $ h \gtrsim n^{-\frac{\alpha+2\beta-2}{2(\alpha+2\beta)}}$, we take $\tau =0$ and $m \asymp n^{1/(\alpha+2\beta)}$, then
					\[ \expect_\star (\langle X^\star, \hat b - b\rangle^2) = O_P\big( n^{-\frac{\alpha+2\beta-1}{\alpha+2\beta}} \big). \]
				\end{itemize}
			\end{enumerate}
		\end{corollary}
		
		\begin{proof}[Proof of Theorem \ref{thm:pre_error}]
			The Karhunen-L$\grave{\mathrm o}$eve expansion leads to $X^\star(t) = \sum_{k=1}^\infty \xi_k^\star \phi_k(t)$, where $\xi_k^\star = \int_\tdomain X^\star(t)\phi_k(t)dt$. Recall that $b(t) = \sum_{k=1}^\infty b_k \phi_k(t)$ and $\hat b(t) = \sum_{k=1}^m \hat b_k \hat \phi_k(t)$, and denote $\hat \xi_k^\star = \int_\tdomain X^\star(t) \hat \phi_k(t)dt$. Then,
			\begin{align*}
				& \quad \expect_\star (\langle X^\star, \hat b - b\rangle^2) \\
				& =  \expect_\star \bigg( \sum_{k=1}^m \xi_k^\star b_k + \sum_{k=m+1}^\infty \xi_k^\star b_k - \sum_{k=1}^m \hat \xi_k^\star \hat b_k  \bigg)^2  \\
				& = \expect_\star \bigg\{ \sum_{k=m+1}^\infty \xi_k^\star b_k + \sum_{k=1}^m \xi_k^\star (b_k - \hat b_k) +  \sum_{k=1}^m ( \xi_k^\star - \hat \xi_k^\star) (\hat b_k - b_k) + \sum_{k=1}^m ( \xi_k^\star - \hat \xi_k^\star) b_k \bigg\}^2 \\
				& \le 4 \expect_\star \bigg( \sum_{k=m+1}^\infty \xi_k^\star b_k \bigg)^2 + 4 \expect_\star \bigg\{ \sum_{k=1}^m \xi_k^\star (b_k - \hat b_k) \bigg\}^2 + 4 \expect_\star \bigg\{ \sum_{k=1}^m ( \xi_k^\star - \hat \xi_k^\star) (\hat b_k - b_k) \bigg\}^2 + 4 \expect_\star \bigg\{ \sum_{k=1}^m ( \xi_k^\star - \hat \xi_k^\star) b_k \bigg\}^2 \\
				& \le 4 \sum_{k=m+1}^\infty b_k^2 \lambda_k + 4 \sum_{k=1}^m \lambda_k(b_k - \hat b_k)^2 + 4\expect_\star(\|X^\star\|_2^2)\sum_{k=1}^m \lambda_k^{-1}\|\phi_k - \hat \phi_k\|_2^2  \sum_{k=1}^m \lambda_k(\hat b_k - b_k)^2 + 4m\expect_\star(\|X^\star\|_2^2)\sum_{k=1}^m\|\phi_k - \hat \phi_k\|_2^2 b_k^2 . 
			\end{align*}
			
			Note that 
			\[ \sum_{k=m+1}^\infty b_k^2 \lambda_k = O(m^{1-\alpha-2\beta}), \]
			\[  m\expect_\star(\|X^\star\|_2^2)\sum_{k=1}^m\|\phi_k - \hat \phi_k\|_2^2 b_k^2 = O_P(mN^{-1}),\]
			\[ \expect_\star(\|X^\star\|_2^2)\sum_{k=1}^m \lambda_k^{-1}\|\phi_k - \hat \phi_k\|_2^2  = O_P(m^{3+\alpha}N^{-1}) = o_P(1). \]
			Consequently, it suffices to bound the term $\sum_{k=1}^m \lambda_k (b_k - \hat b_k)^2$.
			
			First consider the case of $\tau \asymp n^{-1/2}$. Based on the proof arguments in Lemma \ref{lem:rate_w} and Lemma \ref{lem:rate_delta}, we obtain 
			\[  \sum_{k=1}^m \lambda_k (b_k - \hat b_k)^2 = O_P\bigg(\frac{h}{n^{1/2}} \land h^2 + m^{1-\alpha-2\beta} + \frac{m}{N} + \frac{1+m^{3-\alpha}}{N}h^2 \bigg). \]
			
			Next, we focus on the case of $\tau = 0$. According to the proof arguments in Lemma \ref{lem:w/o_lasso}, we have
			\[ \sum_{k=1}^m \lambda_k (b_k - \hat b_k)^2 = O_P\bigg( m^{1-\alpha-2\beta} + \frac{m}{n} + \frac{m}{N} \bigg) . \]
			
			Combining the above pieces together completes the proof.
			
		\end{proof}
		
		\section{Auxiliary Lemmas and Proofs}
		
		\begin{lemma}\label{lem:Delta}
			Note that $\expect \hat \Delta_{kj}^2 = O(N^{-1}(kj)^{-\alpha} )$, $\expect \hat \Delta_k^2 = O(N^{-1}k^{-\alpha}) $ and $\expect \hat \Delta^2 = O(N^{-1})$.
		\end{lemma}
		
		\begin{proof}
			Note that
			\begin{align*}
				&\quad \hat{K}^{\mathcal A}(s,t) - K^{\mathcal A}(s,t)  \\
				&= \sum_{l=1}^L \pi_l \bigg[ \frac{1}{n_l-1} \sum_{i=1}^{n_l} \left\{ \big( X_i^{(l)}(s) - \bar{X}^{(l)}(s) \big) \big( X_i^{(l)}(t) - \bar{X}^{(l)}(t)  \big) \right\} - \\
				& \quad \expect \big\{\big(X^{(l)}(s) -\mu^{(l)}(s) \big) \big( X^{(l)}(t) -\mu^{(l)}(t) \big)\big\}  \bigg].
				%\sum_{l=1}^L \pi_l \left[ \frac{1}{n_l-1} \sum_{i=1}^{n_l} \left\{X_i^{(l)}(s)X_i^{(l)}(t) - E\big(X^{(l)}(s)X^{(l)}(t)\big) \right\} + \mu^{(l)}(s) \mu^{(l)}(t)- \bar{X}^{(l)}(s) \bar{X}^{(l)}(t) \right]. 
			\end{align*}
			By the independence and the Cauchy-Schwarz inequality, we have
			\begin{eqnarray*}
				\expect \hat \Delta_{kj}^2 & = & \expect \left\{\int \int (\hat{K}^{\mathcal A}(s,t) - K^{\mathcal A}(s,t)) \phi_k(s) \phi_j(t) dsdt \right\}^2\\
				%	& = & E\left\{ \sum_{l=1}^L \pi_l \frac{1}{n_l-1} \sum_{i=1}^{n_l} \big(\xi_{i,k}^{(l)}\xi_{i,j}^{(l)} - \lambda_k^{(l)}\mathds{1}(k=j) \big) - \sum_{l=1}^L \pi_l \frac{1}{n_l^2} \sum_{i=1}^{n_l} \xi_{i,k}^{(l)} \sum_{i=1}^{n_l} \xi_{i,j}^{(l)}\right\}^2 \\
				& = & \expect \left[ \sum_{l=1}^L \pi_l \left\{  \frac{1}{n_l-1} \sum_{i=1}^{n_l}\bigg( \big(\xi_{i,k}^{(l)} - \bar{\xi}_k^{(l)}\big)\big(\xi_{i,j}^{(l)} -\bar{\xi}_j^{(l)} \big) \bigg) - \lambda_k^{(l)}\mathds{1}(k=j)  \right\} \right]^2 \\
				& = & \expect \left[ \sum_{l=1}^L \pi_l \left\{  \frac{1}{n_l} \sum_{i=1}^{n_l} \big(\xi_{i,k}^{(l)} \xi_{i,j}^{(l)} - \lambda_k^{(l)}\mathds{1}(k=j) \big) - \frac{1}{n_l(n_l-1)} \sum_{i_1 \neq i_2} \xi_{i_1,k}^{(l)}\xi_{i_2, j}^{(l)}  \right\} \right]^2 \\
				& \le & c\sum_{l=1}^L \pi_l^2 \frac{1}{n_l} \big\{\expect(\xi_{k}^{(1)})^4 \expect(\xi_{j}^{(1)})^4\big\}^{1/2}  = O(N^{-1}(kj)^{-\alpha}),
			\end{eqnarray*}
			according to Assumptions \ref{assump:cov-decay} and \ref{assump:x-moment}.
			% unbiased estimate
			Notice that
			\[ \expect \hat \Delta_k^2 = \sum_{j=1}^{\infty}\expect \hat \Delta_{kj}^2 = O(N^{-1}k^{-\alpha}). \]
			Moreover, using similar arguments obtains $\expect \hat \Delta^2   = O(N^{-1}). $
		\end{proof}
		
		\begin{lemma}\label{lem:eigfun}
			If $m^{2(\alpha+1)}N^{-1} \to 0$, then for $k=1, \dots, m$,
			\[ \|\hat \phi_k - \phi_k\|_2^2 = O_P(k^2 N^{-1}). \]
		\end{lemma}
		
		\begin{proof}
			By (5.16) in \citet{hall2007methodology}, 
			\[ \|\hat \phi_k - \phi_k \|_2^2 \le 2 \hat \theta_k^2, \]
			where 
			\[ \hat \theta_k^2 = \sum_{j: j \neq k} (\hat \lambda_k^{\mathcal A} - \lambda_j^{\mathcal A})^{-2} \left( \int(\hat K^{\mathcal A} - K^{\mathcal A}) \hat \phi_k \phi_j \right)^2. \]
			%Define the event 
			%\[ \mathcal E_1 = \{ (\hat \lambda_k^{\mathcal A} - \lambda_j^{\mathcal A})^{-2} \le 2 (\lambda_j^{\mathcal A} - \lambda_k^{\mathcal A})^{-2} \le C m^{2(\alpha+1)}  \}. \]
			%If we have $m^{2(\alpha+1)} N^{-1} \to 0$, then we have $P(\mathcal E_1) \to 1$.
			Under the event $\mathcal E_1$ defined in \ref{eq:event1}, 
			\begin{align}\label{eq:hat_theta}
				\hat \theta_k^2 & \le 4 \sum_{j: j \neq k} (\lambda_k^{\mathcal A} - \lambda_j^{\mathcal A})^{-2} \left( \int(\hat K^{\mathcal A} - K^{\mathcal A})  \phi_k \phi_j \right)^2 + \nonumber\\
				&\quad 4 \sum_{j: j \neq k} (\lambda_k^{\mathcal A} - \lambda_j^{\mathcal A})^{-2} \left( \int(\hat K^{\mathcal A} - K^{\mathcal A}) (\hat \phi_k - \phi_k) \phi_j \right)^2.
			\end{align}
			The second term in \eqref{eq:hat_theta} is bounded by
			$ Cm^{2(\alpha+1)} \hat \Delta^2 \| \hat \phi_k - \phi_k \|_2^2.$
			Moreover, by Lemma \ref{lem:lam_k-lam_j},
			\[ \sum_{j: j \neq k} (\lambda_k^{\mathcal A} - \lambda_j^{\mathcal A})^{-2} \expect \left( \int(\hat K^{\mathcal A} - K^{\mathcal A})  \phi_k \phi_j \right)^2 = O(k^2 N^{-1}).  \]
			Using $\expect \hat \Delta^2 = O(N^{-1})$ in Lemma \ref{lem:Delta}, we have $\|\hat \phi_k - \phi_k\|_2^2 = O_P(k^2N^{-1})$ for $k=1, \dots, m$.
		\end{proof}
		
		\begin{lemma}\label{lem:g}
			Let $\hat g(t) = \sum_{l=1}^L \pi_l (n_l-1)^{-1} \sum_{i=1}^{n_l} \big\{(X_i^{(l)}(t) - \bar{X}^{(l)}(t) )(Y_i^{(l)} - \bar{Y}^{(l)} )\big\} $ and $g(t) = \sum_{l=1}^L \pi_l \expect\big\{\big(Y^{(l)} - \expect(Y^{(l)}) \big)\big( X^{(l)}(t) - \mu^{(l)}(t) \big) \big\}$. We have $\expect \langle \hat g -g, \phi_k\rangle^2 = O(k^{-\alpha} N^{-1})$ and $\expect \|\hat g -g\|_2^2 = O(N^{-1})$.
		\end{lemma}
		
		\begin{proof}
			Note that
			\begin{align*}
				& \hat g(t) - g(t) \\
				=& \sum_{l=1}^L \pi_l \left[ \left\{ \frac{1}{n_l-1}\sum_{i=1}^{n_l} \big\{(X_i^{(l)}(t) - \bar{X}^{(l)}(t) )(Y_i^{(l)} - \bar{Y}^{(l)} )\big\}  \right\} - \expect\big\{\big(Y^{(l)} - \expect(Y^{(l)}) \big)\big( X^{(l)}(t) - \mu^{(l)}(t) \big) \big\} \right].
			\end{align*}
			By the independence and the Cauchy-Schwarz inequality,
			\begin{align*}
				\expect\langle \hat g -g, \phi_k\rangle^2 & =\expect \left(\sum_{l=1}^L \pi_l \left[ \left\{ \frac{1}{n_l-1}\sum_{i=1}^{n_l} \big\{(\xi_{i,k}^{(l)}- \bar{\xi}_k^{(l)} )(Y_i^{(l)} - \bar{Y}^{(l)} )\big\}  \right\} - \expect\big\{\big(Y^{(l)} - \expect(Y^{(l)}) \big)\xi_k^{(l)} \big\} \right]\right) ^2 \\
				&=  \expect \left[ \sum_{l=1}^L \pi_l \left\{\frac{1}{n_l}\sum_{i=1}^{n_l} \big( \xi_{i,k}^{(l)}Y_i^{(l)} - \expect(\xi_{k}^{(l)}Y^{(l)}) ) - \frac{1}{n_l(n_l-1)}\sum_{i_1 \neq i_2} \xi_{i_1,k}^{(l)}Y_{i_2}^{(l)}   \right\} \right]^2 \\
				& = \expect \sum_{l=1}^L \pi_l^2 \left\{\frac{1}{n_l}\sum_{i=1}^{n_l} \big( \xi_{i,k}^{(l)}Y_i^{(l)} - \expect(\xi_{k}^{(l)}Y^{(l)}) ) - \frac{1}{n_l(n_l-1)}\sum_{i_1 \neq i_2} \xi_{i_1,k}^{(l)}Y_{i_2}^{(l)}   \right\} ^2  \\
				& \le c \sum_{l=1}^L \pi_l^2 \frac{1}{n_l} \big\{ \expect(\xi_k^{(l)})^2 \expect(\epsilon^{(l)})^2 + \expect \big((\xi_k^{(l)})^2\langle X^{(l)},w^{(l)}\rangle^2\big) \big\} \\
				& = O(N^{-1}k^{-\alpha}),
			\end{align*}
			due to Assumptions \ref{assump:cov-decay}-\ref{assump:x-moment}.
			Note that
			\[ \expect \|\hat g-g\|_2^2 = \sum_{k=1}^\infty 	\expect\langle \hat g -g, \phi_k\rangle^2 =  O(N^{-1}).   \]
		\end{proof}
		
		\begin{lemma}\label{lem:T1}
			Let $g_j = \langle g, \phi_j\rangle$ and
			\[T_{k1} = \sum_{j:j \neq k} (\lambda_k^{\mathcal A} - \lambda_j^{\mathcal A})^{-1}g_j\left( \int(\hat K^{\mathcal A} - K^{\mathcal A})  \phi_k \phi_j \right). \]
			We have $\expect T_{k1}^2 = O\big(N^{-1}(k^{-\alpha} + k^{-\alpha}h^2 + k^{2-2\alpha}h^2)\big)$.
		\end{lemma}
		\begin{proof}
			Note that
			\[ T_{k1} = \sum_{l=1}^L \pi_l  \sum_{j:j \neq k} (\lambda_k^{\mathcal A} - \lambda_j^{\mathcal A})^{-1}g_j  \left\{ \frac{1}{n_l-1}\sum_{i=1}^{n_l} (\xi_{i,k}^{(l)} -\bar \xi_k^{(l)}) ( \xi_{i,j}^{(l)} - \bar \xi_j^{(l)})  \right\}. \]
			Since $E\big|\xi_{j_1}^{(l)}\xi_{j_2}^{(l)}\xi_{j_3}^{(l)}\xi_{j_4}^{(l)}\big| \le c \Pi_{\kappa=1}^4 (\lambda_{j_\kappa}^{(l)})^{1/2} $,
			\begin{align*}
				\expect T_{k1}^2 = & \sum_{l=1}^L \pi_l^2 \expect \left\{ \sum_{j: j \neq k} (\lambda_k^{\mathcal A} - \lambda_j^{\mathcal A})^{-1}g_j \left( \frac{1}{n_l}\sum_{i=1}^{n_l} \big(\xi_{i,k}^{(l)} \xi_{i,j}^{(l)}  \big) - \frac{1}{n_l(n_l-1)} \sum_{i_1 \neq i_2} \xi_{i_1,k}^{(l)}\xi_{i_2, j}^{(l)}  \right)  \right\}^2 \\
				\le & 2 \sum_{l=1}^L \pi_l^2 \frac{1}{n_l} \expect \left\{ \sum_{j:j \neq k} (\lambda_k^{\mathcal A} - \lambda_j^{\mathcal A})^{-1}g_j  \big(\xi_{i,k}^{(l)} \xi_{i,j}^{(l)}  \big) \right\}^2 + \\
				& 2  \sum_{l=1}^L \pi_l^2 \frac{1}{n_l(n_l-1)} \expect\left\{ (\xi_k^{(l)})^2 \left( \sum_{j:j \neq k} (\lambda_k^{\mathcal A} - \lambda_j^{\mathcal A})^{-1}g_j \xi_j^{(l)}   \right)^2  \right\} \\
				\le & c \sum_{l=1}^L \pi_l^2 n_l^{-1} \big(\expect(\xi_k^{(l)})^4\big)^{1/2} \left[\expect \left\{ \sum_{j:j \neq k} (\lambda_k^{\mathcal A} - \lambda_j^{\mathcal A})^{-1}g_j  \xi_{i,j}^{(l)} \right\}^4 \right]^{1/2} \\
				\le & c N^{-1} \sum_{l=1}^L \pi_l \lambda_k^{(l)}  \left\{ \sum_{j_1: j_1 \neq k} \cdots \sum_{j_4: j_4 \neq k} \expect(\xi_{i,j_1}^{(l)} \xi_{i,j_2}^{(l)}  \xi_{i,j_3}^{(l)}  \xi_{i,j_4}^{(l)}  ) \times \Pi_{\ell=1}^4 ((\lambda_k^{\mathcal A} - \lambda_{j_\ell}^{\mathcal A})^{-1}g_{j_\ell} ) \right\}^{1/2} \\
				\le & c N^{-1} \sum_{l=1}^L \pi_l \lambda_k^{(l)} \left\{ \sum_{j:j \neq k} |\lambda_k^{\mathcal A} - \lambda_j^{\mathcal A}|^{-1}|g_j|(\lambda_j^{(l)})^{1/2} \right\}^2 \\
				\le & c\frac{k^{-\alpha}+ k^{-\alpha}h^2 + k^{2-2\alpha}h^2}{N},
			\end{align*}
			where the last inequality is obtained by using Lemma \ref{lem:lam_k-lam_j}.
		\end{proof}
		
		\begin{lemma}\label{lem:lam_k-lam_j}
			Under Assumption \ref{assump:cov-decay}, we have
			\[
			|\lambda_k - \lambda_j| \ge c_\alpha
			\left\{ 
			\begin{array}{cc}
				j^{-\alpha}, & j<k/2, \\
				k^{-\alpha}, & j>2k, \\
				|k-j|k^{-(\alpha+1)}, & k/2 \le j \le 2k, j\neq k.
			\end{array}
			\right.,
			\]
			for some positive constant $c_\alpha$ only depending on $\alpha$.
		\end{lemma}
		\begin{proof}
			See the arguments on page 85 in \citet{hall2007methodology} and Lemma 7 in \citet{dou2012estimation}.
		\end{proof}
		
		\begin{lemma}\label{lem:expansions}
			If we have
			\[ K^{\mathcal A}(s,t) = \sum_{k=1}^\infty \lambda_k^{\mathcal A} \phi_k(s)\phi_k(t), \quad \hat K^{\mathcal A} = \sum_{k=1}^\infty \hat \lambda_k^{\mathcal A} \hat \phi_k(s) \hat \phi_k(t),   \]
			then,
			\[ \bigg| \hat \lambda_k^{\mathcal A} - \lambda_k^{\mathcal A} - \int \big( \hat K^{\mathcal A} - K^{\mathcal A} \big) \phi_k\phi_k \bigg|  \le \|\hat \phi_k - \phi_k\|_2(|\hat \lambda_k^{\mathcal A} - \lambda_k^{\mathcal A}| + \hat \Delta_k ), \]
			where $\hat \Delta_k =  [\int  \{\int \big( \hat K^{\mathcal A}(s,t) - K^{\mathcal A}(s,t) \big) \phi_k(s) ds \}^2 dt]^{1/2} $.
			Furthermore, if $\inf_{k\neq j} | \hat \lambda_j^{\mathcal A} - \lambda_k^{\mathcal A} | >0$, then
			\[  \hat \phi_j(t) - \phi_j(t) = \sum_{k: k\neq j} (\hat \lambda_j^{\mathcal A} - \lambda_k^{\mathcal A})^{-1} \phi_k(t) \int(\hat K^{\mathcal A} - K^{\mathcal A}) \hat \phi_j \phi_k + \phi_j(t)\int (\hat \phi_j- \phi_j) \phi_j.  \]
		\end{lemma}
		\begin{proof}
			See Lemma 5.1 in \citet{hall2007methodology}.
		\end{proof}
	
\end{document}